\documentclass[11pt]{article}
\usepackage[english]{babel}
\usepackage[utf8]{inputenc}
\usepackage[T1]{fontenc}
\usepackage{mathrsfs}

\usepackage[leqno]{amsmath}
\usepackage{color}
\usepackage{todonotes} 
\usepackage{url}

\definecolor{dukeblue}{rgb}{0.0, 0.0, 0.61}
\definecolor{internationalkleinblue}{rgb}{0.0, 0.18, 0.65}
\definecolor{coolblack}{rgb}{0.0, 0.18, 0.39}
\definecolor{darkblue}{rgb}{0.0, 0.0, 0.55}

\usepackage{comment}
\usepackage{stmaryrd} 

\usepackage{geometry}
\usepackage{amsmath,amssymb}
\usepackage{bm} 

\usepackage{enumerate} 
\usepackage{amsthm}
\usepackage{mathtools} 

\usepackage{filemod} 

\usepackage{varwidth}

\usepackage{pdflscape}
\usepackage{afterpage}

\makeatletter
\newcommand{\leqnos}{\tagsleft@true\let\veqno\@@leqno}
\newcommand{\reqnos}{\tagsleft@false\let\veqno\@@eqno}
\reqnos
\makeatother

\mathtoolsset{showonlyrefs,showmanualtags}

\newtheoremstyle{boldremark}
{\dimexpr\topsep/2\relax} 
{\dimexpr\topsep/2\relax} 
{}          
{}          
{\bfseries} 
{.}         
{.5em}      
{}          

\theoremstyle{plain}
\newtheorem{thm}{Theorem}[section] 
\newtheorem{lemma}[thm]{Lemma} 
\newtheorem{prop}[thm]{Proposition} 
\newtheorem{cor}[thm]{Corollary} 
\newtheorem{sass}{Standing Assumption}

\theoremstyle{definition}
\newtheorem{defn}[thm]{Definition} 

\numberwithin{equation}{section}

\theoremstyle{remark}

\theoremstyle{boldremark}
\newtheorem{rem}[thm]{Remark}

\usepackage{dsfont}
\newcommand{\R}{\mathds{R}}
\newcommand{\N}{\mathds{N}}

\DeclareMathOperator{\id}{id}
\DeclareMathOperator*{\sgn}{sgn}
\DeclareMathOperator*{\argmax}{argmax}

\renewcommand{\epsilon}{\varepsilon}
\renewcommand{\P}{\mathbb{P}}

\newcommand{\E}{\mathbb{E}}

\newcommand{\hC}{\hat{C}}
\newcommand{\hX}{\hat{X}}
\newcommand{\hQ}{\hat{Q}}
\newcommand{\hPhi}{\hat{\Phi}}

\newcommand{\NN}{\mathds{N}}
\newcommand{\OO}{\mathds{O}}
\newcommand{\PP}{\mathds{P}}

\newcommand{\RR}{\mathds{R}}
\newcommand{\bS}{\mathds{S}}
\newcommand{\UU}{\mathds{U}}
\newcommand{\VV}{\mathds{V}}

\newcommand{\cF}{\mathcal{F}}

\newcommand{\sA}{\mathscr{A}}

\newcommand{\sC}{\mathscr{C}}

\newcommand{\sP}{\mathscr{P}}
\newcommand{\sS}{\mathscr{S}}

\newcommand{\sO}{\mathcal O}
\newcommand{\hq}{q_M}
\newcommand{\hatm}{m_M}

\newcommand{\diff}{\mathrm{d}}
\newcommand{\dd}{\,\mathrm{d}}

\newcommand*{\EX}[2][]{\E^{#1}\left [ #2 \right ]}
\newcommand*{\cEX}[2][]{\E\left[ #2 \,\middle\vert\,\cF_{#1} \right]}

\newcommand*{\as}[1]{#1\text{-a.s.}}

\newcommand*{\ol}[1]{\overline{#1}}
\newcommand*{\ul}[1]{\underline{#1}}

\usepackage{lipsum}

\definecolor{pipurple}{RGB}{128,0,128}
\definecolor{vargreen}{RGB}{0,150,0}
\definecolor{varyellow}{RGB}{180,120,0}

\makeatletter
\g@addto@macro\bfseries{\boldmath}
\makeatother

\usepackage{caption}
\usepackage{subcaption}

\usepackage{IEEEtrantools}

\begin{document}

\title{Portfolio Optimization under Transaction Costs with Recursive Preferences}
	
	\author{Martin Herdegen\thanks{University of Warwick, Department of Statistics, Coventry, CV4 7AL, UK. Email: m.herdegen@warwick.ac.uk.} \and
		David Hobson\thanks{University of Warwick, Department of Statistics, Coventry, CV4 7AL, UK. Email:  d.hobson@warwick.ac.uk.} \and
		Alex S.L. Tse\thanks{University College London, Department of Mathematics, London, WC1H 0AY, UK. Email: alex.tse@ucl.ac.uk.}
}
	\date{February 12, 2024}

	\maketitle
	
\begin{abstract}
The Merton investment-consumption problem is fundamental, both in the field of finance, and in stochastic control. An important extension of the problem adds transaction costs, which is highly relevant from a financial perspective but also challenging from a control perspective because the solution now involves singular control. A further significant extension takes us from additive utility to stochastic differential utility (SDU), which allows time preferences and risk preferences to be disentangled.

In this paper, we study this extended version of the Merton problem with proportional transaction costs and Epstein-Zin SDU. We fully characterise all parameter combinations for which the problem is well posed (which may depend on the level of transaction costs) and provide a full verification argument that relies on no additional technical assumptions and uses primal methods only. The case with SDU requires new mathematical techniques as duality methods break down.

Even in the special case of (additive) power utility, our arguments are significantly simpler, more elegant and more far-reaching than the ones in the extant literature. This means that we can easily analyse aspects of the problem which previously have been very challenging, including comparative statics, boundary cases which heretofore have required separate treatment and the situation beyond the small transaction cost regime. A key and novel idea is to parametrise consumption and the value function in terms of the shadow fraction of wealth, which may be of much wider applicability.
\end{abstract}


\section{Introduction}
\label{sec:intro}
One of the fundamental problems of economic theory is to determine the optimal consumption and investment behaviour of agents who are endowed with an initial capital, who face a stochastic opportunity set and who are averse to fluctuations of consumption both with respect to time and with respect to stochastic innovations. In a classic formulation of the problem, known as the Merton problem~\cite{Merton:69,Merton:71}, an agent with constant relative risk aversion (CRRA) seeks to maximise the expected, integrated, discounted utility of consumption over the infinite horizon, where consumption is financed from initial wealth and from investments in a risky asset driven by Brownian motion. The important outputs are the optimal instantaneous rate of consumption and the optimal fraction of wealth the agent should invest in the risky asset.

As well as being a key problem in economics, the Merton problem is also a classical problem in stochastic control. Merton~\cite{Merton:69} showed how to derive the Hamilton-Jacobi-Bellman equation for the problem and found the value function and optimal strategy. Although his arguments are heuristic\footnote{See Herdegen et al~\cite{herdegen2020elementary} for a modern, rigorous treatment of the problem.}, his solution is correct: namely, when the problem is well-posed the agent should consume at a rate which is a constant proportion of wealth and invest a constant fraction of wealth in the risky asset, where these constants can be expressed in terms of the parameters of the financial market and the agent's preferences.

Merton's solution of the investment-consumption problem is a beautiful application of stochastic control, and was one of the inspirations behind many subsequent developments of the theory, including the use of Hamilton-Jacobi-Bellman equations, verification arguments and the martingale optimality principle, and the dual approach, see, for example, Fleming and Soner~\cite{FlemingSoner}, and Pham~\cite{Pham}. But, from a practical perspective, the final answer is unsatisfactory in two important ways. First, the suggestion that agents should invest a constant fraction of wealth in the risky asset is infeasible in markets with transaction costs: if the asset price follows a semi-martingale with non-trivial local martingale term, then to keep a constant fraction of wealth in the risky asset involves a trading strategy with infinite variation. Such a strategy would immediately exhaust the wealth of the agent if there are positive transaction costs. Second, the fraction of wealth that Merton proposes that agents should invest in the risky asset does not match the actual fraction that the representative agent invests in the financial market (as taken from US market data) -- this is the equity premium puzzle of Mehra and Prescott~\cite{mehra1985equity}. The economics and mathematics literature has responded to these issues in two directions, first by adding transaction costs to the model, and second by assigning more sophisticated preferences to the agent in the form of recursive utility, and its continuous-time analogue stochastic differential utility.

Magill and Constantinides~\cite{MagillConstantinides:76} were the first to introduce transaction costs into the investment-consumption model and to formulate the intuition which underpins the solution, namely that in $(X, Y \Phi)$ = ({Cash}, {Value of investments in the risky asset}) space\footnote{Here $Y$ is the price of the risky asset, and $\Phi$ is the number of shares held.}, instead of trading to be on the Merton line $\frac{Y \Phi}{X+ Y \Phi}=p$ (where $p$ is a constant) agents should trade to keep the fraction of wealth in the risky asset in an interval (or equivalently to keep $(X, Y \Phi)$ in a no-transaction wedge). Davis and Norman~\cite{DavisNorman:90} formalised the model, and explained that the resulting problem involved singular stochastic control, where the controlled process $(X_t, Y_t \Phi_t)_{t \geq 0}$ is kept in the no-transaction region using local time pushes on the boundary. For some parameter combinations, Davis and Norman~\cite{DavisNorman:90} showed that the value function could be expressed via the solution of two coupled first order ODEs with free-boundaries, but in terms of understanding the solution, and even of numerically solving the problem, further progress is challenging with this formulation. Shreve and Soner~\cite{shreve:soner} improved on the Davis-Norman solution (in terms of the range of parameter combinations that were covered) using methods from viscosity solutions, but issues with giving a simple, accessible formulation of the value function and optimal strategy remained.
Instead efforts turned to using asymptotic methods to give an expansion for the value function (and thence to derive an expansion for the optimal no-transaction region and consumption rate) in the regime of small transaction costs. The small transaction cost literature will be discussed in more detail in Section~\ref{sec:smalltc}, but includes \cite{Rogers:04,shreve:soner,JanecekShreve:04,GerholdMuhleKarbeSchachermayer:18,Choi:14,hobson:tse:zhu:19A}.

Returning to the case of general proportional transaction costs, significant progress was made when Choi et al~\cite{ChoiSirbuZitkovic:13} and Hobson et al~\cite{hobson:tse:zhu:19A} separately showed how the problem could be decoupled into first, solving a initial-value problem for a first order ODE (and then adjusting the initial value so that an integral equation is satisfied) and second, solving a further first order ODE. Indeed, very important quantities such as the locations of the boundaries of the no-transaction region can be determined from the solution of the first ODE only. Choi et al~\cite{ChoiSirbuZitkovic:13} took a dual approach based on the shadow price for the asset with transaction costs whereas Hobson et al~\cite{hobson:tse:zhu:19A} use the primal approach. Understanding the final form of the problem (including the set of possible behaviours of the first order ODE, and how they relate to the fundamental properties of the financial problem, eg, when is it well-posed) is simpler in the setting used by Hobson et al~\cite{hobson:tse:zhu:19A}, but both approaches  can by criticised for appearing to use `magic' to transform the problem into a simpler version --- especially in \cite{hobson:tse:zhu:19A} these transformations appear to have no economic motivation.

The second response in the literature is to extend the analysis to allow for recursive preferences. Additive preferences can be criticised for the fact that a single parameter is used to describe the agents attitude to fluctuations in consumption due to stochastic risk and their attitude to fluctuations in consumption over time. Recursive preferences postulate that the utility of current consumption depends on the value of future consumption. This makes it possible to separate risk aversion from the elasticity of inter-temporal substitution (EIS). In fact, for notational convenience we express our problem in terms of the elasticity of
intertemporal complementarity (EIC) which is the reciprocal of the EIS. The analogue of recursive utility in continuous-time is stochastic differential utility, and in this setting the natural generalisation of the additive CRRA utility is Epstein-Zin stochastic differential utility (EZ-SDU) as proposed in Epstein and Zin~\cite{epstein1989substitution}.

The extension to EZ-SDU brings several challenges. In the finite-horizon setting, for a given consumption the associated value function is the solution of a backward stochastic differential equation (BSDE). But, over the infinite horizon it is not completely clear how to formulate an analogue of the terminal condition. Then, there are significant issues which arise concerning existence and uniqueness of solutions. Indeed, for many parameter combinations for EZ-SDU uniqueness fails for essentially every consumption stream and we must identify
which value process we wish to associate to a given consumption stream (using an economic criteria, rather than an ad-hoc selection). Only when these questions have been resolved is it possible to optimise over admissible investment-consumption strategies (i.e. those which keep wealth non-negative). Fortunately these issues have been considered in depth in Herdegen et al~\cite{herdegen:hobson:jerome:23A,herdegen:hobson:jerome:23B,herdegen:hobson:jerome:23C}. Since the real issue is to define a value process given a consumption process, we can `piggyback' on their analysis -- the main impact of the introduction of transaction costs is to change the set of attainable consumption streams, and the value associated to a given consumption stream is unchanged.

Although the literature on optimal investment and consumption with recursive utility is growing rapidly (see Chacko and Viceira~\cite{chacko2005dynamic}, Kraft et al~\cite{kraft2017optimal} and Xing~\cite{xing2017consumption}) it focuses on the frictionless case. The exception is Melnyk et al~\cite{MMKS:20} who look for an asymptotic expansion of the value function in the case of small transaction costs and derive expressions for the asymptotic expansion for the no-transaction region and the optimal consumption, albeit under a restricted set of parameter combinations.

In this article we consider the lifetime investment and consumption problem in a financial market with transaction costs for an agent with recursive preferences.
Our main contributions are as follows:
\begin{enumerate}
\item First, we give a complete solution of the consumption-investment problem in the case of EZ-SDU which is valid for all levels of transaction costs. This includes giving exact conditions for when the problem is well-posed. The methods are an extension of the analysis in the additive case in Hobson et al~\cite{hobson:tse:zhu:19A} and the solution involves finding a family of solutions to a first crossing problem for a first order ODE (indexed by its starting point on a curve) and fixing the solution for a given level of transaction costs by choosing the solution which satisfies an integral equation. The start and end point of this solution determine the boundaries of the no-transaction region. Other important quantities, such as the shadow price of the risky asset and the optimal consumption are given by the solutions of a further first order ODE (this time with fixed start and end points).
\item Second, we introduce a new key variable, namely the shadow fraction of wealth. If we make this the underlying variable (in preference, say, to the shadow price) then all other quantities can be expressed in a simple fashion in terms of this single variable. The idea of characterising problems via the shadow fraction of wealth is new and may have implications well beyond the investment-consumption problem setting, and may lead to significant simplifications in other settings too. One immediate candidate example is maximising the long run rate of portfolio growth under transaction costs.
\item Third, we `break the magician's code' and give interpretations to many of the quantities which arose `by magic' in the analysis in~\cite{hobson:tse:zhu:19A}. There, the transformations lead to a great simplification due to an act of serendipity; here, we explain where they come from and how they may be interpreted. Thus $q$, $n(q)$ and $m(q)$ in \cite{hobson:tse:zhu:19A} are here given interpretations as the shadow fraction of wealth, the optimal consumption rate per unit of wealth as a function of the shadow fraction of wealth, and the optimal consumption rate in a frictionless Merton model if the agent invests a constant fraction $q$ of wealth in the risky asset.
\item Fourth, use of the shadow fraction of wealth as a key variable simplifies many of the arguments. At points in the analysis we have to consider a highly non-standard ODE which passes through a singular point (and has multiple solutions to the left of this point, and just one to the right). In our coordinate system this analysis becomes relatively straightforward (whereas in other systems understanding the phase-space of the solutions can be difficult). One implication is that we only need a one-dimensional version of Tanaka's results to define the underlying processes via a local time, whereas to date much of the literature has been forced to rely on a multidimensional version. Further, we can treat several boundary cases as part of the main analysis, including the case when one of the boundaries of the no-transaction region corresponds to a portfolio with no cash component, and then the leading term in a small-transaction cost expansion for the locations of the boundary of the no-transaction region is of a different order to the standard case.
\item Fifth, we show how risk aversion $R$ mainly governs the optimal investment strategy and the no-transaction region (the co-efficient of elasticity of inter-temporal complimentarity $S$ does not appear directly in any of the first three terms of the expansion in small transaction costs for the boundaries of the no-trade interval) and $S$ mainly governs the rate of consumption. Further, (following on from Jane\v{c}ek and Shreve~\cite{JanecekShreve:04} in the additive case $R=S$ and Melnyk et al~\cite{MMKS:20} for EZ-SDU) we show that consumption may increase or decrease with the introduction of small transaction costs, and we give an explanation which shows that the introduction of transaction costs will increase the instantaneous consumption rate if $S$ is small and decrease it if $S$ is large.
\item Sixth, we investigate analytically the comparative statics with respect to the risk aversion parameter $R$ and the elasticity of inter-temporal complimentarity $S$. One of the advances of this paper is that we do not (only) consider the additive case, and therefore it is possible for the first time to disentangle the impact of risk aversion and the intertemporal substitution. We find that in many circumstances (for example, if $R < 1$ or $R \geq 2$) the boundaries of the no-transaction region are monotonic in the risk aversion parameter in the sense that the agent with higher risk aversion buys and sells the risky asset when the fraction of wealth in the risky asset is at a lower threshold than an agent with lower risk aversion investing and consuming in the same market. However, this is not universal and for $R \in (1,2)$ this monotonic relationship can fail. The dependence of the no transaction wedge on $S$ is less clear cut, and whether the no-transaction region moves towards or away from a cash portfolio depends on a crucial combination of parameters, see Section~\ref{sec:comparative}.
\end{enumerate}
Along the way we give several other interesting, original and important results:
\begin{enumerate}
\item We give precise statements for when the problem under EZ-SDU is well posed for all level of transaction costs, when it is ill-posed for all levels of transaction costs, and when it is well-posed for some levels of transaction costs and ill-posed for others. In the latter case we give an expression for the threshold level of transaction costs at which the problem becomes ill-posed.
\item Unlike most other authors we allow for different transaction costs on buying and selling; often the literature assumes that these costs are the same, or that the transaction cost on selling is zero. Whilst there is a straightforward transformation between the various parameterisations of transaction costs, covering the general case facilitates an easy comparison to the individual cases. It also highlights that, for small transaction costs, the individual transaction costs only enter the expansion at a higher order level, and the first few terms only involve the round-trip transaction cost.
\item We show that if the Merton line lies in the first quadrant (i.e. in the case of zero transaction costs the optimal portfolio involves a long position in both the risky asset and the risk-less asset) then it is always inside the no-transaction region; however if the Merton line is not in the first quadrant then it can lie outside the no transaction wedge (and thus becomes a very poor approximation of a reasonable investment strategy).
\item We show that if (for large enough transaction costs) the positions corresponding to all wealth being held in the risky asset lie inside the no-transaction region then (for large transaction costs) the location of the sale-boundary on the no-transaction region is independent of the transaction cost for buying.
\item We give the asymptotic expansion for the value function and other quantities in the small transaction cost regime to a higher order than has been given elsewhere.
\end{enumerate}

The remainder of the paper is structured as follows. In Section~\ref{sec:ezsdu} we introduce stochastic differential utility and discuss the formulation of the investment consumption problem under Epstein-Zin SDU over the infinite horizon. In Section~\ref{sec:frictionless} we solve the frictionless investment-consumption problem for EZ-SDU. The arguments are extended to the case with frictions in Sections~\ref{sec:tcheuristics} and \ref{sec:tcrigour} with Section~\ref{sec:tcheuristics} focussing on the heuristics and Section \ref{sec:tcrigour} concentrating on the rigorous statements. Section~\ref{sec:examples} gives some illustrative numerical examples and Section~\ref{sec:smalltc} covers the asymptotics in the small transaction regime. In Section~\ref{sec:comparative} we consider the comparative statics with the respect to the two key parameters of EZ-SDU, namely the risk aversion $R$ and the elasticity of intertemporal complementarity $S$. Some of the more technical material is relegated to the appendices.

\section{Epstein-Zin Stochastic Differential Utility}
\label{sec:ezsdu}
Let $R, S \in (0,\infty) \setminus \{1\}$ and $\delta \in \RR$ and set $\VV = (1-R)\ol\R_+$. The Epstein--Zin (EZ) aggregator (Epstein and Zin~\cite{epstein1989substitution} is the function $g_{\rm EZ}:\R_+ \times \R_+ \times \VV \to \VV$, given by
\begin{equation}\label{eq:Epstein--Zin aggregator R and S}
g_{\rm EZ}(t,c,v) \coloneqq e^{-\delta t}\frac{c^{1-S}}{{1-S}}\left( (1-R)v\right)^\frac{S-R}{1-R}.
\end{equation}
Here, $R \in (0,\infty) \setminus \{1\}$ denotes the coefficient of \emph{relative risk aversion}, $S\in (0,\infty) \setminus \{1\}$ denotes the coefficient of \emph{elasticity of intertemporal complementarity} (the reciprocal of the coefficient of intertemporal substitution) and $\delta \in \RR$ denotes some discount parameter.

It is convenient to introduce the parameters $\theta\coloneqq\frac{1-R}{1-S}$ and $\rho=\frac{S-R}{1-R} = 1- \frac{1}{\theta}$, so that \eqref{eq:Epstein--Zin aggregator R and S} simplifies to
\begin{equation}\label{eq:Epstein--Zin aggregator}
g_{\rm EZ}(t,c,v) = e^{-\delta t}\frac{c^{1-S}}{{1-S}}\left( (1-R)v\right)^\rho.
\end{equation}
Note that $\theta = 1$ (or equivalently, $\rho=0$) corresponds to the case $R = S$, which is additive power utility.

Given a nonnegative progressively measurable stochastic process $C =(C_t)_{t \geq 0}$, called a \emph{consumption stream}, we seek to find a utility process $V^C = (V^C_t)_{t \geq 0}$ that solves the equation
\begin{equation}\label{eq:Epstein--Zin SDU}
V_t = \E\left[\left.\int_t^\infty  e^{-\delta s}\frac{C_s^{1-S}}{{1-S}}\left( (1-R)V_s\right)^\rho \dd s\,\right| \,\cF_t \right].
\end{equation}

The problem is to find $\hat C \in \sC$ such that
\begin{equation}\label{eq:problem}
V^{\hat C}_0  =\sup_{C\in\sC}V^C_0,
\end{equation}
where $V^C = (V^C_t)_{t \geq 0}$ denotes an appropriate solution of \eqref{eq:Epstein--Zin SDU}, and $\sC$ is an appropriate class of consumption streams.

Throughout the paper, we make the following standing assumption
\begin{sass}
	\label{sass:rational} We assume $\theta \in (0,\infty)$.\end{sass}

The necessity of this standing assumption can be seen from multiplying \eqref{eq:Epstein--Zin SDU} on both sides by $(1-R)$. Then, with $\tilde{V}=(1-R)V$, which implicitly must be positive for it to be possible to take arbitrary powers, we have $\tilde{V}_t = \theta \E[ \int_t^\infty e^{-\delta s} C_s^{1-S} \tilde{V}_s^\rho \dd s | \cF_t]$ and we must have $\theta>0$ for the two sides of this equation to have the same sign. See Herdegen et al ~\cite{herdegen:hobson:jerome:23A} for further discussion of this issue.

For $\theta \in (0, 1]$, we have existence and uniqueness of a utility process $V^C = (V^C_t)_{t \geq 0}$ solving \eqref{eq:Epstein--Zin SDU} (in a generalised sense) for all consumption streams, see Herdegen et al ~\cite{herdegen:hobson:jerome:23A}. By saying that $V^C$ solves \eqref{eq:Epstein--Zin SDU} in a generalised sense we mean that we allow $(1-R)V^C$ to take the value $+\infty$ provided that $\E\left[\left.\int_t^\infty  e^{-\delta s}\frac{C_s^{1-S}}{{1-S}}\left( (1-R)V^C_s\right)^\rho \dd s\,\right| \,\cF_t \right]$ also takes the value $+\infty$.

\begin{thm}[{\cite[Theorem 6.9]{herdegen:hobson:jerome:23B}}]
\label{thm:SDU:existence:theta<1}
Suppose $\theta \leq 1$. For each $C \in \sP_+$, there exists
a unique generalised utility process $V^C = (V^C_t)_{t \geq 0}$ solving  \eqref{eq:Epstein--Zin SDU} (in a generalised sense).
\end{thm}

By contrast, if $\theta > 1$, the situation is more involved because  \eqref{eq:Epstein--Zin SDU} has multiple solutions. For example, $V \equiv 0$ is always a solution to \eqref{eq:Epstein--Zin SDU} when $\rho>0$. To get uniqueness we restrict ourselves to so-called proper solutions; see Herdegen et al~\cite{herdegen:hobson:jerome:23C} for a discussion of the issues.

\begin{defn}\label{def:proper}
	Suppose $C\in\sP_+$ and suppose that $V=(V_t)_{t\geq0}$ is a solution to \eqref{eq:Epstein--Zin SDU}. Then $V$ is called a \textit{proper solution} if $(1-R)V_t>0$ on $\{\cEX[t]{\int_t^\infty  e^{-\delta \theta s} C^{1-R}_s  \dd s } > 0\}$ up to null sets for all $t\geq 0$.
\end{defn}

Denote by $\UU\PP^*$ the collection of right-continuous processes $C \in \sP_+$ for which there exists a unique proper utility process $V^C = (V^C_t)_{t \geq 0}$ solving \eqref{eq:Epstein--Zin SDU}.   
The results of Herdegen et al~\cite{herdegen:hobson:jerome:23C} (especially Proposition 4.8) show that is a wide class of processes. When $\theta>1$ we will optimise over consumption streams in $\UU \PP^*$.

\section{Frictionless problem}
\label{sec:frictionless}
We consider a Black--Scholes--Merton market consisting of a risk free asset with interest rate $r\in\R$, whose price process $Y^0=(Y^0_t)_{t\geq0}$ is given by $Y^0_t = y^0_0 e^{r t}$ (for $y^0_0>0$), and a risky asset given by geometric Brownian motion with drift $\mu\in\R$, volatility $\sigma>0$ and initial value $y_0>0$, whose price process $Y=(Y_t)_{t \geq 0}$ is given by $Y_t = y_0\exp(\sigma B_t + (\mu-\frac{1}{2}\sigma^2))$, for a Brownian motion $B=(B_t)_{t\geq0}$.

At each time $t\geq0$, the agent chooses to consume at a rate $C_t\in\R_+$ per unit time, and to hold $\Phi_t\in\R$ numbers of shares of the risky asset. If we denote total wealth at time $t$ by $Z_t$, the number of shares invested into the riskless asset is given by $\Phi^0_t = \frac{Z_t - \Phi_t Y_t}{Y^0_t}$. Then, by the self-financing condition, the (total) wealth process of the agent $Z=(Z_t)_{t\geq0}$ satisfies the SDE
\begin{equation}\label{eqn:frictionless:wealth process}
\dd Z_t =  \Phi^0_t \dd Y^0_t + \Phi_t \dd Y_t - C_t \dd t = \Phi_t Y_t  \sigma \dd B_t + \left( Z_t  r + \Phi_t Y_t (\mu - r)  - C_t\right)\mathrm{d} t,	
\end{equation}
subject to the initial condition $Z_0 =z$, where $z>0$ is the agent's initial wealth.

\begin{rem}
It is customary in the literature to parameterise trading strategies in terms of fractions of wealth (rather than numbers of shares). However, this parametrisation is less suitable for the problem with transaction costs. For this reason, and to make an easy comparison with the frictionless case, we also parameterise trading in the frictionless problem in numbers of shares.
\end{rem}

\begin{defn}\label{def:frictionlessadmissible}
	Given $z >0$, an \emph{admissible investment-consumption strategy} (for $z$) is a pair $(\Phi,C)= (\Phi_t,C_t)_{t \geq 0}$ of progressively measurable processes, where $\Phi$ is real-valued and $C$ is nonnegative, such that the SDE \eqref{eqn:frictionless:wealth process} has a unique strong solution $Z^{z, \Phi, C} = (Z^{z, \Phi, C}_t)_{t \geq 0}$ that is $\as{\P}$ nonnegative. We denote the set of admissible investment-consumption strategies for $z > 0$ by $\sA_0(z)$.
\end{defn}
Since the utility value associated to an admissible investment-consumption strategy only depends upon consumption, and not upon the amount invested into the risky asset, we also introduce the following definition:
\begin{defn}\label{def:frictionlessattainable}
	Given $z > 0$, an \emph{attainable consumption stream} (for $z$) is a nonnegative progressively measurable stochastic process $C =(C_t)_{t \geq 0}$ for which there exists a progressively measurable process $\Phi = (\Phi_t)_{t\geq0}$ such that $(\Phi, C) \in \sA_0(z)$. Denote the set of attainable consumption streams for $z > 0$ by $\sC_0(z)$.
\end{defn}

The frictionless problem is to find $\hat C \in \sC_0(z)$ such that
\begin{equation}\label{eq:frictionless problem}
V^{\hat C}_0  =\sup_{C\in\sC^*_0(z)}V^C_0,
\end{equation}
where, if $\theta \in (0,1]$, $V^C = (V^C_t)_{t \geq 0}$ denotes the unique generalised utility process from Theorem \ref{thm:SDU:existence:theta<1} solving \eqref{eq:Epstein--Zin SDU} (in a generalised sense) and $\sC^*_0(z) = \sC_0(z)$, and, if $\theta \in (0,1]$, $V^C = (V^C_t)_{t \geq 0}$ denotes the unique proper solution of \eqref{eq:Epstein--Zin SDU} and $\sC^*_0(z) = \sC_0(z) \cap \UU \PP^*$.

\subsection{Solution for constant proportional strategies}
\label{ssec:frictionless}
In order to motivate our solution ansatz in the case with frictions, it is insightful to study the optimisation problem \eqref{eq:frictionless problem} for constant proportional strategies only. To this end, consider the investment-consumption strategy $(\Phi,C)$, where $\Phi = q \frac{Z}{Y}$ and $C = m Z$ for $q \in \RR$, $m \in\RR_{++}$. Then, the wealth process $Z^{z,\Phi,C} = Z = (Z_t)_{t \geq 0}$ from \eqref{eqn:frictionless:wealth process} is given by
\begin{equation}
Z_t  =  z\exp\left(q \sigma B_t + \left({r} + q  \sigma \lambda - m - \frac{q^2 \sigma^2}{2}\right)t \right),
\end{equation}
where $\lambda := \frac{\mu - r}{\sigma}$ denotes the \emph{market price of risk}.
We make the ansatz
\begin{equation}
\label{eq:ansatz}
V_t = a e^{-b t} \frac{Z_t^{1-R}}{1-R}
\end{equation}
for some constants $a, b \in \RR$ to be determined. Substituting this ansatz into \eqref{eq:Epstein--Zin SDU}, and recalling that $\rho = \frac{S-R}{1-R}$ yields
\begin{align}
\label{eq:Z^1-R equation constant strategy}
V_t      = \E\left[\left.\int_t^\infty e^{-\delta s} \frac{(m Z_s)^{1-S}}{{1-S}} \left( a e^{-b s} Z_s^{1-R}\right)^\rho \dd s\,\right| \,\cF_t \right] = a^\rho \frac{m^{1-S}}{1-S}  \E\left[\left.\int_t^\infty  e^{-(\delta+ b \rho)s}   Z_s^{1-R} \dd s\,\right| \,\cF_t \right].
\end{align}
A straightforward calculation gives
\begin{equation}
\label{eq:Z:cond exp}
	\cEX[t]{e^{-(\delta + b \rho)s}Z_{s}^{1-R}} = e^{-(\delta + b\rho)t} Z_t^{1-R} e^{- \theta \tilde{H}_b(q,m)(s-t)}, \quad s > t,
\end{equation}
where  $\tilde{H}_b: \R \times \R_{++} \mapsto \R$ is given by $\tilde{H}_b(q, m) = \frac{1}{\theta}(\delta+b \rho) + (S-1)\left({r} + \lambda \sigma q - m - \frac{q^2\sigma^2}{2}R\right)$.
Provided that $\tilde{H}_b(q, m) > 0$, \eqref{eq:Z^1-R equation constant strategy} simplifies to
\begin{equation}
V_t =  \frac{ a^\rho m^{1-S} }{\tilde{H}_b(q,m)}e^{-(\delta+ b \rho) t} \frac{Z_t^{1-R}}{1-R} .
\end{equation}
Our ansatz \eqref{eq:ansatz} now implies that $\frac{a^\rho m^{1-S} }{\tilde{H}_b(q,m)} = a $ and $\delta+ b \rho = b$. We deduce that $b = \delta \theta$ and that $a = (\frac{m^{1-S} }{\tilde{H}_b(q,m)})^\theta $.
If we substitute this value of $b$ into $\tilde{H}_b$ and define $H = \tilde{H}_{\delta \theta}$ then we have
\begin{equation}
\label{eq:H_nu}
H(q, m) = \delta + (S-1)\left({r} + \lambda \sigma q - m - \frac{q^2\sigma^2}{2}R\right).
\end{equation}
and provided $H(q,m)>0$,
\begin{equation}\label{eq:valfungenstrat}
V_t =  \left(\frac{m^{1-S} }{H(q,m)}\right)^\theta e^{-\delta \theta t} \frac{Z_t^{1-R}}{1-R} .
\end{equation}
If we now maximise \eqref{eq:valfungenstrat} for fixed $q$ over $m$, we obtain that the optimal consumption rate expressed per unit of total wealth is given by $m=m(q)$ where
\begin{align}
\label{eq:m}
m(q) &= \frac{\delta }{S} - \frac{1-S}{S} \left({r} + \lambda \sigma q  - \frac{q^2\sigma^2}{2}R\right) = \frac{\alpha}{S} - \frac{1-S}{S} \lambda \sigma q + R\frac{1-S}{S} \frac{\sigma^2}{2} q^2,
\end{align}
where $\alpha := \delta - r(1-S)$ denotes the \emph{effective impatience rate}.
Plugging \eqref{eq:m} into \eqref{eq:valfungenstrat} yields
\begin{equation}\label{eq:V0:opt:q}
V_t =   e^{-\delta \theta t} \frac{Z_t^{1-R}}{1-R} m(q)^{-\theta S}
\end{equation}
Differentiating \eqref{eq:V0:opt:q} with respect to $q$ and setting the derivative equal to zero,  gives that the turning point of $m$ is at $q_M = \frac{\lambda}{R \sigma}$. Clearly, $q_M$ is the famous \emph{Merton ratio}. Thus, in frictionless market the optimal fraction of wealth to invest in the risky asset is independent of the elasticity of intertemporal complimentarity.

The above analysis yields a candidate value function and a candidate optimal strategy. It is well known that
\begin{thm}[\cite{MMKS:20,herdegen:hobson:jerome:23B,herdegen:hobson:jerome:23C}]  Consider the frictionless market.
Suppose that $m_M = m(q_M) = m\left(\frac{\lambda}{\sigma R}\right)>0$. Then the optimal-investment consumption problem is well-posed, and the value function is given by
\[ V_0 = \sup_{C \in \sC_0^*(z)} V^C_0 = \frac{z^{1-R}}{1-R} m_M^{-\theta S} \]
The optimal strategy is $(C^* = m_M Z, \Phi = \frac{\lambda }{\sigma R} \frac{Z}{Y})$.

If $m_M \leq 0$ then the problem is ill-posed.
\end{thm}

\begin{rem}
\label{rem:shapeofm}
The shape of $m=m(q)$ will play a key role in later arguments.
The shape of $m$ is determined by the shape of $H(q,m)$ as given in \eqref{eq:H_nu}, and as a function parameterised by $R$ and $S$, the shape of $H$ is determined by the sign of $(S-1)$. In particular, when $S<1$, $H$ is convex in $q$, and this property is inherited by $m(q)$ (see Section~\ref{ssec:WhyCdependsonS} for details of how this inheritance works). Then, $m(q)$ has a {\em minimum} over $q \in \R$ at $q_M$. Conversely, when $S>1$, $H$ is concave in $q$, this property is inherited by $m$, and $m$ has a {\em maximum} (over $q \in \R$ such that $H(q,m)>0$) at $q_M$. Indeed, given the sign of $V$, this must be the case for $V$ in \eqref{eq:V0:opt:q} to have an interior {\em maximum} in $q$.

Economically, the explanation is as follows. If the agent invests a suboptimal fraction of wealth in the risky asset then this will reduce the certainty equivalent value of their holdings. This will have two opposite implications. First, since they are implicitly less wealthy (in certainty equivalent terms), they will reduce instantaneous consumption, and as a fraction of their current wealth their consumption rate will fall. Second, the primary cause of the reduction of certainty equivalent wealth is a reduction in the valuation of future (rather than current) consumption, and this will encourage them to increase the rate of instantaneous consumption. When the elasticity of inter-temporal complementarity $S$ is small (and the coefficient of inter-temporal substitution is large), the desire to smooth consumption over time is small (see Herdegen et al~\cite[Section 4.1]{herdegen:hobson:jerome:23A}) and the second effect dominates. When $S$ is larger, the desire to seek smooth consumption over time is increased and the first effect dominates.

We return to this issue in Remark~\ref{rem:consumption}(c) and Section~\ref{ssec:WhyCdependsonS}.
\end{rem}

\section{Problem with friction: notation and heuristics}
\label{sec:tcheuristics}
Now we consider the problem with transaction costs. We assume that in order to buy one unit of the risky asset at time $t$ one has to pay $\gamma^\uparrow Y_t$, and that the receipts from selling one unit of the risky asset at time $t$ are given by $\frac{1}{\gamma^\downarrow} Y_t$. Here $\gamma^\uparrow, \gamma^\downarrow \in [1, \infty)$ denote the transaction cost parameters for buying and selling. When $\gamma^\uparrow \gamma^\downarrow > 1$, we have to restrict ourselves to trading strategies $\Phi = (\Phi_t)_{t \geq 0}$ (in numbers of shares) that are of finite variation. In this case, we denote by $\Phi^\uparrow = (\Phi^\uparrow_t)_{t \geq 0}$ and $\Phi^\downarrow = (\Phi^\downarrow_t)_{t \geq 0}$ the increasing and decreasing part of $\Phi$, respectively, so that $\Phi = \Phi^\uparrow - \Phi^\downarrow$ where $\Phi^\uparrow$ and $\Phi^\downarrow$ are increasing processes. Moreover, in this case, it is useful not to describe the total wealth process but rather to consider the cash process and the holdings in the risky asset separately.

When choosing the investment-consumption strategy $(\Phi, C)$, the corresponding cash process $X = (X_t)_{t \geq 0}$ satisfies the dynamics
\begin{equation}
\label{eq:friction:cash process}
 \diff X_t= r X_t  \dd t - C_t \dd t - \gamma^\uparrow Y_t \dd \Phi_t^\uparrow  + \frac{Y_t}{\gamma^\downarrow}\dd \Phi_t^\downarrow.
\end{equation}
This is paired with the initial condition
\begin{equation}
\label{eq:frictionIC}
X_0 = x_0 - \gamma^\uparrow y (\Phi_0 - \phi_0)^{+} + \frac{y}{\gamma^\downarrow} (\Phi_0 - \phi_0)^-,
\end{equation}
where $x_0 \in \R$ denotes the initial cash position and $\phi_0$ denotes the initial stock position prior to trading at time $0$. Note that we allow for an initial block trade just before time zero (i.e. at time $0-$) that is not encoded in the strategy $\Phi = (\Phi_t)_{t \geq 0}$. This allows us to consider $\Phi$ to be a right-continuous process.

For $\gamma^\uparrow, \gamma^\downarrow \in [1, \infty)$, we denote by
\begin{equation*}
\sS_{\gamma^\uparrow, \gamma^\downarrow} := \left\{(x, y, \phi) \in \RR \times (0, \infty) \times \RR: x  +  \frac{1}{\gamma^\downarrow}\phi^+ y - \gamma^\uparrow \phi^-  y \geq 0\right\}
\end{equation*}
the \emph{solvency region}. We denote by $\sS^\circ_{ \gamma^\uparrow, \gamma^\downarrow}$ its interior, which satisfies
\begin{equation*}
\sS^\circ_{ \gamma^\uparrow, \gamma^\downarrow} = \left\{(x, y, \phi) \in \RR \times (0, \infty) \times \RR: x  +  \frac{1}{\gamma^\downarrow}\phi^+ y -
\gamma^\uparrow \phi^- y > 0\right\}.
\end{equation*}

By analogy with Definitions~\ref{def:frictionlessadmissible} and \ref{def:frictionlessadmissible}
we can define what it means for an investment-consumption strategy $(\Phi, C)$ to be admissible for $(\gamma^\uparrow, \gamma^\downarrow)$-transaction costs, and a consumption stream $C$ to be attainable for $(\gamma^\uparrow, \gamma^\downarrow)$-transaction costs.

\begin{defn}
	Given $(x, y, \phi) \in \sS^\circ_{\gamma^\uparrow, \gamma^\downarrow}$, a \emph{ $(\gamma^\uparrow, \gamma^\downarrow)$-admissible investment-consumption strategy} is a pair $(\Phi,C)= (\Phi_t,C_t)_{t \geq 0}$ of progressively measurable processes, where $\Phi$ is real-valued and of finite variation and $C$ is nonnegative, such that the SDE \eqref{eq:friction:cash process} has a unique strong solution $X^{x, y, \phi, \Phi, C} = (X^{x, y, \phi, \Phi, C}_t)_{t \geq0} $ such that $(X^{x, y, \phi, \Phi, C}_t, \Phi_t, Y_t) \in \sS_{ \gamma^\uparrow, \gamma^\downarrow}$ $\as{\P}$ for all $t \geq 0$. We denote the set of $(\gamma^\uparrow, \gamma^\downarrow)$-admissible investment-consumption strategies for $(x, s, \phi) \in \sS^\circ_{\gamma^\uparrow, \gamma^\downarrow}$ by $\sA_{\gamma^\uparrow, \gamma^\downarrow}(x, \phi, y)$.
\end{defn}

\begin{defn}
Given $(x, y, \phi) \in \sS^\circ_{\gamma^\uparrow, \gamma^\downarrow}$, a \emph{$(\gamma^\uparrow, \gamma^\downarrow)$-attainable consumption stream} (for $(x, y, \phi) $) is a nonnegative progressively measurable stochastic process $C =(C_t)_{t \geq 0}$ for which there exists a progressively measurable process $\Phi = (\Phi_t)_{t\geq0}$ such that $(\Phi, C) \in \sA_{\gamma^\uparrow, \gamma^\downarrow}(x, \phi, y)$. Denote the set of attainable consumption streams for $(x, y, \phi) \in \sS^\circ_{\gamma^\uparrow, \gamma^\downarrow}$ by $\sC_{\gamma^\uparrow, \gamma^\downarrow}(x, y, \phi)$.
\end{defn}

\begin{rem}
\label{rem:inclusion}
It is straightforward to check that for all $\gamma^\uparrow, \gamma^\downarrow \in [1, \infty)$, $\sA_{\gamma^\uparrow, \gamma^\downarrow}(x, \phi, y) \subseteq \sA_0(x + \phi y)$ and $\sC_{\gamma^\uparrow, \gamma^\downarrow}(x, \phi, y) \subseteq \sC_0(x + \phi y)$. More generally, if $\tilde{\gamma}^\uparrow, \tilde{\gamma}^\downarrow$ are such that $\tilde{\gamma}^\uparrow \geq {\gamma}^\uparrow$ and $\tilde{\gamma}^\downarrow \geq \tilde{\gamma}^\downarrow$ then $\sC_{\tilde{\gamma}^\uparrow, \tilde{\gamma}^\downarrow}(x, \phi, y) \subseteq \sC_{\gamma^\uparrow, \gamma^\downarrow}(x, \phi, y)$.
\end{rem}

For $(x, y, \phi) \in \sS^\circ_{\gamma^\uparrow, \gamma^\downarrow}$, the problem with friction is to find $\hat C \in \sC_{\gamma^\uparrow, \gamma^\downarrow}(x, \phi, y)$ such that
\begin{equation}\label{eq:frictional problem}
V^{\hat C}_0 =\sup_{C\in\sC^*_{\gamma^\uparrow, \gamma^\downarrow}(x, \phi, y)}V^C_0,
\end{equation}
where, if $\theta \in (0,1]$, $V^C = (V^C_t)_{t \geq 0}$ denotes the unique generalised utility process from Theorem \ref{thm:SDU:existence:theta<1} solving \eqref{eq:Epstein--Zin SDU} (in a generalised sense), and $\sC^*_{\gamma^\uparrow, \gamma^\downarrow}(x, \phi, y) = \sC_{\gamma^\uparrow, \gamma^\downarrow}(x, \phi, y)$, and if $\theta>1$, $V^C = (V^C_t)_{t \geq 0}$ denotes the unique proper utility process 
 solving \eqref{eq:Epstein--Zin SDU} and $\sC^*_{\gamma^\uparrow, \gamma^\downarrow}(x, \phi, y) = \sC_{\gamma^\uparrow, \gamma^\downarrow}(x, \phi, y) \cap \UU\PP^*$.

\subsection{Shadow fraction of wealth and relative shadow price}

For a cash process $X = (X_t)_{t \geq 0}$ corresponding to some  $(\gamma^\uparrow, \gamma^\downarrow)$-admissible investment-consumption strategy $(\Phi, C)$, define the \emph{fraction of wealth} $P = (P_t)_{t \geq 0}$  invested into the risky asset by
\begin{equation}
P_{t}= \frac{\Phi_t Y_t}{X_t + \Phi_t Y_t }.
\end{equation}
Let $\tilde Y = (\tilde Y_t)_{t \geq 0}$ denote a \emph{shadow stock price process}, which takes values in $Y [\frac{1}{\gamma^\downarrow}, \gamma^\uparrow]$.

Then the process $Q = (Q_t)_{t \geq 0}$ defined by
\begin{equation}
\label{eq:def:shadow fraction wealth}
Q_{t}= \frac{\Phi_t \tilde Y_t}{X_t + \Phi_t \tilde Y_t },
\end{equation}
denotes the corresponding \emph{shadow fraction of wealth}. For future reference, we note that the ratio of
shadow wealth to real wealth can be expressed as
\begin{equation}
\label{eq:shadow over real}
\frac{X+ \Phi \tilde Y}{X +\Phi Y} =  \frac{\frac{X}{X + \Phi Y}}{\frac{X}{X + \Phi \tilde Y}} = \frac{1 -P}{1- Q}.
\end{equation}
We make the ansatz that
the \emph{relative shadow price} $\tilde{Y}/Y$ can be written in terms of the shadow 
fraction of wealth, and thus is given by
\begin{equation}
\frac{\tilde Y}{Y} = \kappa(Q)
\end{equation}
for some function $\kappa$ to be determined later. Note that by \eqref{eq:shadow over real},
\begin{equation}
\label{eq:J q cons}
\kappa(Q)= \frac{\tilde Y}{Y} = \frac{\Phi \tilde Y}{\Phi Y} = \frac{Q (X +\Phi \tilde Y)}{P(X + \Phi  Y)}=  \frac{Q/(1-Q)}{P/(1-P)}.
\end{equation}
We also make the ansatz that
\begin{equation}
Q = q(P)
\end{equation}
for some function $q$ to be determined later. Then \eqref{eq:J q cons} gives
\begin{equation}
\label{eq:consitency:cond}
\kappa(q(p))= \frac{q(p)/(1-q(p))}{p/(1-p)}
\end{equation}
These equations may be inverted to express $p$ in terms of $q$ in which case we find $P=p(Q)$ where
\begin{equation}
\label{eq:pintermsofq}
p(q) = \frac{q}{1 + q(1- \kappa(q)))}, \hspace{20mm} \kappa(q) = \frac{q}{1-q} \frac{1 - p(q)}{p(q)} .
\end{equation}
Differentiating \eqref{eq:consitency:cond} and rearranging yields the key identity
\begin{equation}
\label{eq:consitency:cond:diff}
p(1-p) q'(p)=\left(\frac{1}{q(p)(1-q(p))}-\frac{\kappa'(q(p))}{\kappa(q(p))}\right)^{-1}.
\end{equation}

\subsection{Solution ansatz}

Motivated by the structure \eqref{eq:V0:opt:q} of the frictionless solution, we make the first ansatz that the value process $V  = (V_t)_{t \geq 0}$ is given by
\begin{equation}
\label{eq:ansatz:val proc}
V_t := e^{- \delta \theta t} \frac{(X_t + \Phi_t \tilde Y_t)^{1-R}}{1-R} n\left(Q_t\right)^{-\theta S}
= e^{- \delta \theta t} \frac{(X_t + \Phi_t Y_t \kappa(Q_t))^{1-R}}{1-R} n\left(Q_t\right)^{-\theta S}
,
\end{equation}
where the function $n$ is to be determined. Here $n$ represents the rate of consumption per unit of wealth, when wealth is measured using the shadow stock price.

Thus, setting $p = \frac{ \phi y}{x +\phi y}$, we look for a value function of the form $V = V(I, X,\Phi,Y)$ where
\begin{equation}
\label{eq:ansatz:val proc2}
V(t, x, \phi, y)=e^{- \delta \theta t}\frac{(x + \phi y \kappa(q(p)))^{1-R}}{1-R} n\left(q(p)\right)^{-\theta S}.
\end{equation}
Inside the no trade region we make the second ansatz that the value function should be independent of the shadow fraction of wealth. Thus, we assume that
\begin{equation}
0 =\frac{\partial V}{\partial q} = (1-R) \frac{\phi y \kappa'(q)}{(x+\phi y \kappa(q))} V - \theta S\frac{n^{\prime}(q)}{n(q)} V,
\end{equation}
which, using $\frac{\phi y}{x + \phi y \kappa} = \frac{q}{\kappa}$, is equivalent to
\begin{equation}
\label{eq:key equation}
\frac{\kappa'(q)}{\kappa(q)} =\frac{S}{1-S} \frac{1}{q} \frac{n^{\prime}(q)}{n(q)}.
\end{equation}
Using \eqref{eq:ansatz:val proc2} together with the second ansatz we can easily compute the first-order partial derivatives of $V$
\begin{align}
V_t &= -\delta \theta V, \label{eq:Vt} \\
x V_{x} & =(1-R) V \frac{x}{(x+\phi y \kappa(q))}=(1-R) V(1-q), \label{eq:Vx}\\
\phi V_{\phi} & =(1-R) V \frac{\phi y \kappa(q)}{(x+\phi y \kappa(q))}=(1-R) V q, \label{eq:Vphi} \\
y V_{y} & =(1-R) V \frac{\phi y \kappa(q)}{(x+\phi y \kappa(q))}=(1-R) V q, \label{eq:Vy}
\end{align}
and then using \eqref{eq:consitency:cond:diff} and \eqref{eq:key equation},
\begin{align}
y^{2} V_{y y} &=(1-R)^{2} q^{2} V-(1-R) V q+(1-R) V p(1-p) q'(p) \\
&= (1-R)V \left((1-R)q^2 - q + \left(\frac{1}{q(1-q)}-\frac{\kappa'(q)}{\kappa(q)}\right)^{-1} \right) \\
&=(1-R)V \left((1-R)q^2 - q + \left(\frac{1}{q(1-q)}-\frac{S}{1-S} \frac{1}{q} \frac{n^{\prime}(q)}{n(q)}\right)^{-1} \right). \label{eq:Vyy}
\end{align}

\subsection{Martingale optimality principle}
Define the process $M = (M_t)_{t \geq 0}$ by
\begin{equation*}
  M_t := \int_0^t e^{-\delta s}\frac{C_s^{1-S}}{1-S} ((1-R)V(s,X_s,Y_s,\Phi_s))^\rho \dd s + V(t, X_t, Y_t \Phi_t).
  \end{equation*}
By the martingale optimality principle, $M$ is a supermartingale for each pair $(C, \Phi)$ and a martingale for the optimal pair $(C^*, \Phi^*)$. It\^{o}'s formula and \eqref{eq:Vt} -- \eqref{eq:Vy} and \eqref{eq:Vyy} give
\begin{align} \diff M_t & =  \frac{C_t^{1-S}}{1-S}((1-R)V)^{\rho} e^{-\delta t} \dd t+ V_t \dd t + V_x \dd X_t + V_y \dd Y_t + V_\phi \dd\Phi_t + \frac{1}{2} V_{yy} \dd \langle Y \rangle_t \\
 & =  \left(\frac{C_t^{1-S}}{1-S}((1-R)V)^\rho e^{-\delta t} -\delta \theta V  + r(1-R)V(1-	q) - V_x  C_t + \mu (1-R)q V \right) \!\dd t  \\
&\hspace{5mm}+ \sigma Y_t V_y\dd B_t+(V_\phi - \gamma^\uparrow V_x Y_t) \dd\Phi^\uparrow_t + \left(\frac{V_x}{\gamma^\downarrow}Y_t - V_\phi \right) \dd\Phi^\downarrow_t \\
     &\hspace{5mm} + \frac{\sigma^2}{2} (1-R) V \left((1-R)Q_t^2 - Q_t+ \left(\frac{1}{Q_t(1-Q_t)}-\frac{S}{1-S} \frac{1}{Q_t} \frac{n^{\prime}(Q_t)}{n(Q_t)}\right)^{-1} \right) \!\dd t \\
 & =  \left(\frac{C_t^{1-S}}{1-S}((1-R)V)^\rho e^{-\delta t} - V_x  C_t  \right) \!\dd t  + \sigma Y_t V_y\dd B_t+(V_\phi - \gamma^\uparrow V_x Y_t) \dd\Phi^\uparrow_t + \left(\frac{V_x}{\gamma^\downarrow}Y_t - V_\phi \right) \dd\Phi^\downarrow_t \\
     &\hspace{5mm} - \theta S V \left(\ell (Q_t) -\frac{\sigma^2}{2}\frac{1-S}{S}\left(\frac{1}{Q_t(1-Q_t)}-\frac{S}{1-S} \frac{1}{Q_t} \frac{n^{\prime}(Q_t)}{n(Q_t)}\right)^{-1} \right)\!\dd t
\label{eq:heuristics:M}
\end{align}
where the function $\ell$ is given by
\begin{equation}
\label{eq:ell}
\ell(q) = \frac{\alpha}{S} - \frac{1-S}{S} \left( \lambda\sigma - \frac{\sigma^2}{2} \right) q - (1-R)\frac{1-S}{S} \frac{\sigma^2}{2} q^2.
\end{equation}

Maximising the first $\diff t$-term on the right hand side of \eqref{eq:heuristics:M} over $C_t$ shows that
the candidate optimal consumption stream  $C^* = (C^*_t)_{t \geq 0}$ satisfies
\begin{equation}
\label{eq:defC*}
C_{t}^{*}:=\left(X_{t}+\Phi_{t} \tilde{Y}_{t}\right) n\left(Q_{t}\right)=\left(X_{t}+\Phi_{t} Y_{t} \kappa\left(Q_{t}\right)\right) n\left(Q_{t}\right),
\end{equation}
This implies in particular that $n$ can be interpreted as the optimal consumption rate in terms of shadow wealth. Then we have
\begin{align}
\left(\frac{(C^*_t)^{1-S}}{1-S}(e^{-\delta t}((1-R)V)^\rho) - V_x  C^*_t \right) &= \frac{S}{1-S} (X_t + \Phi_t \tilde Y_t) n(Q_t) V_x = \theta S V n(Q_t)
\end{align}
Plugging this into \eqref{eq:heuristics:M} and using that the $\diff t$-term must vanish, after dividing by $\theta S V$, we obtain
\begin{equation}
\label{eq:HJB}
0 = n(Q_t) -\ell (Q_t) +\frac{\sigma^2}{2}\frac{1-S}{S}\left(\frac{1}{Q_t(1-Q_t)}-\frac{S}{1-S} \frac{1}{Q_t} \frac{n^{\prime}(Q_t)}{n(Q_t)}\right)^{-1},
\end{equation}
Rearranging \eqref{eq:HJB} and noting that $\ell(q)- m(q)  = \frac{1-S}{S} \frac{\sigma^2}{2} q(1-q)$, we obtain after a rearrangement that $\kappa$ satisfies the ODE
\begin{equation}
\label{eq:ODE kappa}
\kappa'(q) =\frac{\kappa(q)}{q(1-q)} \frac{m(q)-n(q)}{\ell(q)-n(q)}.
\end{equation}
Combining this with \eqref{eq:key equation} we find that $n$ satisfies the autonomous ODE
\begin{equation}
\label{eq:node}
n'(q) = O(n,q) \hspace{5mm} \mbox{where} \hspace{5mm} O(n,q) = \frac{1-S}{S} \frac{n}{1-q} \frac{m(q) - n}{\ell(q) -n}.
\end{equation}
Plugging \eqref{eq:ODE kappa} into \eqref{eq:consitency:cond:diff} and noting the simple identify
\begin{equation}
\frac{m(q(p)) - \ell(q(p))}{m(p) - \ell(p)} = \frac{q(p)(1-q(p))}{p(1-p)} \label{eq:m-l:fraction}
\end{equation}
we obtain after a rearrangement that $q$ satisfies the ODE
\begin{equation}
\label{eq:heur:ODE q}
q'(p) = \frac{\ell(q(p)) - n(q(p))}{\ell(p) - m(p)},
\end{equation}
which may be re-expressed as
\begin{equation}
\label{eq:heur:ODE p}
p'(q) = \frac{\ell(p(q)) - m(p(q))}{\ell(q) - n(q)}.
\end{equation}
We find a candidate solution therefore by first solving the autonomous first equation \eqref{eq:node} for $n$, and then solving \eqref{eq:ODE kappa} to determine $\kappa$. This then yields an expression for $p$ in terms of $q$ via \eqref{eq:pintermsofq} or \eqref{eq:heur:ODE p}. The next issue is to find the initial condition for the ODE \eqref{eq:node} which will depend on the round-trip transaction cost $\xi$.

\subsection{No-trade region and boundary conditions}
We assume that the no-trade region in fraction of shadow wealth is given by $[q_*, q^*]$ for some $q_* < q^* \in \R$. This will correspond to a no-trade region in fractions of real wealth of the form $[p_*, p^*]$ for some $p_* < p^* \in \R$.

Outside the no trade region, the shadow price coincides with the bid or ask price (depending on the side). Hence, we get the following boundary condition for the function $\kappa$.
\begin{equation}
\label{eq:kappa:boundary}
\kappa(q) =  \gamma^\uparrow, \quad q \leq q_*, \quad \text{and} \quad \kappa(q) =  \frac{1}{\gamma^\downarrow},  \quad q \geq q^*
\end{equation}
Considering, \eqref{eq:kappa:boundary} for $q = q_*$ and $q = q^*$, integrating $\kappa'(q)/\kappa(q)$ between  $q_{*}$ and $q^{*}$, and using the ODE \eqref{eq:ODE kappa} for $\kappa$, we obtain
\begin{equation}
\label{eq:intcondition}
\gamma^{\uparrow} \gamma^{\downarrow}= \frac{\kappa(q_*)}{\kappa(q^*)} = \exp\left(-\int_{q_{*}}^{q^{*}}\kappa'(q)/\kappa(q) \dd q\right) = \exp\left(\int_{q_{*}}^{q^{*}} \frac{1}{q(1-q)} \frac{n(q)-m(q)}{\ell(q)-n(q)} \dd q\right).
\end{equation}
Moreover, assuming smooth fit of $\kappa$ at $q_*$ and $q^*$, \eqref{eq:kappa:boundary} implies that $\kappa'(q_*) = 0 = \kappa'(q^*)$. This together with the ODE \eqref{eq:ODE kappa} for $\kappa$, this implies that
\begin{equation}
\label{eq:bc}
n(q_*) = m(q_*) \quad  \text{and} \quad n(q^*) = m(q^*).
\end{equation}
The idea is that the integral condition \eqref{eq:intcondition} together with the boundary condition \eqref{eq:bc} fixes $q_*$ and $q^*$ and then also the particular solution $n=n_{\gamma^{\uparrow} \gamma^{\downarrow}}$ we want. The boundaries $p_*$ and $p^*$ can then be found by
\begin{equation*}
p_* = \frac{q_*}{\gamma^\uparrow  + (1-\gamma^\uparrow) q_*} \quad \text{and} \quad p^* = \frac{\gamma^\downarrow q^*}{1 + (\gamma^\downarrow-1) q^*}.
\end{equation*}
It remains to analyse solutions to $n'(q)=O(n,q)$ started from points $(z,n(z)=m(z))$, and to look for the initial point $z=q_*$ such that \eqref{eq:intcondition} holds.

\begin{rem}
\label{rem:qMnotinwedge}
It will follow from later analysis (see Proposition~\ref{prop:n}) that $q_M \in [q_*,q^*]$ so that the Merton ratio lies inside the no-transaction region expressed in terms of fractions of shadow wealth. But this does imply that $q_M \in [p_*,p^*]$, and it is possible that the Merton line lies outside the no-transaction region. This can only happen when $q_M \notin [0,1]$. For example, in the case $q_M>1$, if $\gamma^\downarrow > \frac{q_M}{q_M-1}$ then the Merton line does not even lie in the solvency region $\sS_{\gamma^\uparrow,\gamma^\downarrow}$. Further, if $q_M>1$ then there are cases when $p^*=p^*(\gamma^\uparrow,\gamma^\downarrow)$ is a non-monotonic function of transaction costs. The fact that it is initially increasing will follow from the small transaction cost results and especially Corollary~\ref{cor:smalltcP}; the fact that it must subsequently decrease will follow from the fact that the no-transaction region must lie inside the solvency region. See also Figure~\ref{fig:eg2_tcost}(d) for a numerical example.
\end{rem}

\section{Problem with friction: rigorous statements}
\label{sec:tcrigour}

We assume henceforth that market price of risk $\lambda:= \frac{\mu -r}{\sigma}$ is non-zero. Otherwise the problem is trivial because the solution for the problem with frictions coincides with the one without friction. At $t=0$ the agent would re-balance their portfolio to leave a portfolio consisting of investments in the riskless asset only; thereafter they would continue to not invest in the risky asset.

\begin{sass}
\label{saa:Merton}
We assume $\lambda \neq 0$.
\end{sass}

Note that Standing Assumption \ref{saa:Merton} implies in particular that the Merton ratio $q_M = \frac{\lambda}{\sigma R}$ is non-zero.

\medskip{}

For better readability, all proofs and some results of a more technical nature can be found in Appendix~\ref{sec:app:tcrigour}.

\subsection{Well-posedness conditions}
We begin our discussion by defining when the problem with friction \eqref{eq:frictional problem} is well posed. It turns out that this does not depend on the transaction cost parameters $\gamma^\uparrow$ and $\gamma^\downarrow$ for buying and selling separately but only on their product $\xi := \gamma^\uparrow \gamma^\downarrow$.
\begin{defn}
\label{def:well posed}
Define the \emph{threshold transaction cost} by\footnote{Here, we agree that $\exp(\int_{q^m_-}^{q^m_+} -\frac{1}{q(1-q)} \frac{m(q)}{\ell(q)} \dd q) := \infty$ in case that $\int_{q^m_-}^{q^m_+} |\frac{1}{q(1-q)} \frac{m(q)}{\ell(q)}| \dd q = \infty$.}
\begin{equation}
\label{eq:ol xi}
\ol \xi := \begin{cases}
1 &\text{if $m$ has at most one zero,} \\
\exp\left(\int_{q^m_-}^{q^m_+} -\frac{1}{q(1-q)} \frac{m(q)}{\ell(q)} \dd q\right) &\text{if $m$ has two distinct zeros $q^m_- < q^m_+$.}
\end{cases}
\end{equation}
Then $\xi \in (1, \infty)$ is said to be in the \emph{well-posedness range of transaction costs} if one of the following two conditions is satisfied
\begin{enumerate}
\item $R <1$, $m(0), m(1) > 0$ and $\xi > \ol \xi$.
\item $R > 1$, $m(q_M) > 0$ and $\xi < \ol \xi$.
\end{enumerate}
\end{defn}
To understand Definition \ref{def:well posed}, let us look at the cases $R < 1$ and $R > 1$ separately.

First, if $R < 1$, the problem with friction cannot be well posed if either investing nothing ($q = 0$) or everything ($q = 1$) in the risky asset leads to infinite wealth in the frictionless problem (which is the case if $m(0) \leq 0$ or $m(1) \leq 0$, respectively) since these are the two cases where the frictionless strategy is also admissible with friction; the corresponding wealth is either the same if $q  = 0$ or reduced by the factor $\frac{1}{\gamma^\downarrow}$ if $q = 1$. But even if $m(0), m(1) > 0$, the frictionless Merton problem will be ill-posed if $m(q_M) \leq 0$. Hence, the problem with friction is only well-posed if transaction costs are sufficiently large.

Next, if $R > 1$, the problem with friction cannot be well posed if the frictionless Merton problem is not well posed, i.e., if $m(q_M) \leq 0$. But even if $m(q_M) > 0$, the problem with friction is only well-posed if transaction costs are sufficiently small.

We shall show in Section \ref{subsec:ill-posed} that if $\xi := \gamma^\uparrow \gamma^\downarrow$ does not lie in the well-posedness range of transaction costs, the problem \eqref{eq:frictional problem} is ill-posed.

\begin{rem}
If $m$ in case (a) of Definition \ref{def:well posed} has two distinct zeros $q^m_- < q^m_+$, then $\ol \xi < \infty$ since either $q^m_- < q^m_+ < 0$ or $0 < q^m_- < q^m_+ < 1$ or $1 < q^m_- < q^m_+$ by the assumption that $m(0), m(1) > 0$.

In case (b) of Definition \ref{def:well posed}, however, if $m(q_M)>0$, then $m$ always has two distinct zeros $q^m_- < q^m_+$. Moreover,  $\ol \xi = \infty$ if and only if $m(0) \geq 0$ or $m(1) \geq 0$ or there exists $q \in (0, 1)$ with $\ell(q) \geq 0$. Using that $m(0) =\ell(0)$ and $m(1) = \ell(1)$ the latter three cases can be summarised by the single condition: $\ol \xi = \infty$ if and only if $\max_{q \in [0, 1]} \ell(q) \geq 0$.
\end{rem}

\subsection{The free boundary problem}

We first establish existence, uniqueness and further properties of the optimal consumption rate $n$. Mathematically, this is linked to a first-order free-boundary problem, which is additionally delicate due to the fact that the right-hand side of the ODE \eqref{eq:prop:n:ODE} is not defined for $q = 1$, which may lie inside the free-boundary region.

The region between the two free boundaries for $n$ turns out to be the shadow no-trade region. Note that the latter does not depend on the transaction costs parameters $\gamma^\uparrow, \gamma^\downarrow$ separately but only on their product $\xi = \gamma^\uparrow \gamma^\downarrow$.
\begin{prop}
\label{prop:n}
Let $\xi\in (1, \infty)$ be in the well-posedness range of transaction costs. Then there exists unique boundary points $q_*(\xi), q^*(\xi) \in \RR$ satisfying
\begin{align}
0 < q_*(\xi) < q_M < q^*(\xi) < 1\quad &\text{if } q_M \in (0, 1),\\
0 < q_*(\xi) <  q^*(\xi) = 1\quad &\text{if } q_M =1,\\
0 < q_*(\xi) < q_M < q^*(\xi) < 2q_M-1  \quad &\text{if } q_M >1, \\
2q_M < q_*(\xi) < q_M < q^*(\xi) < 0 \quad&\text{if } q_M<0,
\end{align}
as well as a unique continuously differentiable function $n_\xi : [q_*(\xi) , q^*(\xi) ] \to (0, \infty)$ with the following properties:
\begin{enumerate}
\item $n_\xi(q) \neq \ell(q)$ for $q \neq 1$.
\item $n_\xi$ is decreasing if $R <1$ and increasing if $R > 1$.
\item On $[q_*(\xi) , q^*(\xi) ] \setminus \{1\}$, $n_\xi$ satisfies the ODE
\begin{equation}
\label{eq:prop:n:ODE}
n'_\xi(q) =  \frac{1-S}{S} \frac{n_\xi}{1-q} \frac{m(q) - n_\xi(q)}{\ell(q) -n_\xi(q)}.
\end{equation}
\item $n_\xi$ satisfies the boundary conditions
\begin{align}
\label{eq:prop:n:bound}
n_\xi(q_*(\xi)) = m(q_*(\xi)) \quad &\text{and} \quad  n_\xi(q^*(\xi)) = m(q^*(\xi)), \\
n'_\xi(q_*(\xi)) = 0  \quad &\text{and} \quad  n'_\xi(q^*(\xi)) = 0.
\label{eq:prop:n:bound:deriv}
\end{align}
\item  If $1 \in [q_*(\xi), q^*(\xi)]$, $n_\xi$ satisfies
\begin{equation}
n_\xi(1) = m(1) \quad \text{and} \quad n'_\xi(1) = m'(1).
\end{equation}
\item $n_\xi$ satisfies the integral constraint
\begin{equation}
\label{eq:prop:n:int const}
\xi = \exp\left(\int_{q_*(\xi)}^{q^*(\xi)} -\frac{1}{q(1-q)} \frac{m(q) -n_\xi(q)}{\ell(q)-n_\xi(q)} \dd q\right). 
\end{equation}
\end{enumerate}
Moreover, $q_*(\xi)$ and $q^*(\xi)$ satisfy the bounds
\begin{align}
q_*(\xi) &< \frac{\xi q^*(\xi)}{1 + (\xi - 1)q^*(\xi)} < \frac{\xi}{\xi -1}, \label{eq:prop:n:q up boundE}\\
q^*(\xi) &> \frac{q_*(\xi)}{\xi - q_*(\xi)(\xi - 1)} > -\frac{1}{\xi-1}.   \label{eq:prop:n:q down boundE}
\end{align}
\end{prop}

We next turn to existence, uniqueness and further properties of the relative shadow price $\kappa$.

\begin{cor}
\label{cor:kappa xi}
Let  $\gamma^\uparrow, \gamma^\downarrow \in [1, \infty)$ with $\xi:= \gamma^\uparrow \gamma^\downarrow  > 1$ be such that $\xi$ lies in the well-posedness range of transaction costs.  Then there exists a unique continuously differentiable function $\kappa_{\gamma^\uparrow,\gamma^\downarrow} : [q_*(\xi) , q^*(\xi) ] \to (0, \infty)$ with the following properties:
\begin{enumerate}
\item $\kappa_{\gamma^\uparrow,\gamma^\downarrow}$ is decreasing.
\item On $[q_*(\xi) , q^*(\xi) ] \setminus \{1\}$, $\kappa_{\gamma^\uparrow,\gamma^\downarrow}$ satisfies the ODE
\begin{equation}
\label{eq:cor:kappa:ODE}
\kappa'_{\gamma^\uparrow,\gamma^\downarrow}(q) = \frac{\kappa_{\gamma^\uparrow,\gamma^\downarrow}(q)}{q(1-q)} \frac{m(q) - n_\xi(q)}{\ell(q)-n_\xi(q)},
\end{equation}
\item $\kappa_{\gamma^\uparrow,\gamma^\downarrow}$ satisfies the boundary conditions
\begin{align}
\label{eq:cor:kappa:bound}
\kappa_{\gamma^\uparrow,\gamma^\downarrow}(q_*(\xi)) = \gamma^\uparrow \quad &\text{and} \quad \kappa_{\gamma^\uparrow,\gamma^\downarrow}(q^*(\xi))  = \frac{1}{\gamma^\downarrow},
\\
\kappa'_{\gamma^\uparrow,\gamma^\downarrow}(q_*(\xi)) = 0 \quad &\text{and} \quad  \kappa'_{\gamma^\uparrow,\gamma^\downarrow}(q^*(\xi)) =0.
\label{eq:cor:kappa:bound:deriv}
\end{align}
\end{enumerate}
Moreover, if $\tilde \gamma^\uparrow,\tilde \gamma^\downarrow \in [1, \infty)$ are such that $\tilde \gamma^\uparrow\tilde \gamma^\downarrow = \xi$, then
\begin{equation}
\kappa_{\tilde \gamma^\uparrow,\tilde \gamma^\downarrow} = \kappa_{\gamma^\uparrow, \gamma^\downarrow} \times \frac{\tilde \gamma^\uparrow}{\gamma^\uparrow} = \kappa_{\gamma^\uparrow, \gamma^\downarrow} \times \frac{ \gamma^\downarrow}{\tilde \gamma^\downarrow}.
\end{equation}
\end{cor}

\subsection{No-transaction region}
We next turn to the no-transaction region in terms of real (as opposed to shadow) quantities. We shall see that the relationship between the real no-transaction region and the shadow no-transaction region is neatly described by real Möbius transformations.

\begin{defn}
Denote by $\RR \cup \{\ol \infty\}$ the compactified real line. For $c \in (0, \infty)$, define the real Möbius transformation  $\tau_c: \RR \cup \{\ol \infty\} \to \RR \cup \{\ol \infty\}$ by
\begin{equation*}
\tau_c(q): =
\begin{cases}
\frac{c q}{1 + (c-1)q} & \text{if }  q \in \RR \setminus \{\frac{1}{1-c}\}, \\
\ol \infty &\text{if } q = \frac{1}{1-c},\\
\frac{c}{c-1} &\text{if } q = \ol \infty.
\end{cases}
\end{equation*}
\end{defn}
Note that the pair $(\{\tau_c, c\in (0, \infty)\}, \circ)$ forms an Abelian group with the additional property
\begin{equation}
\tau_c \circ \tau_d = \tau_{c d}, \quad c, d \in (0, \infty).
\end{equation}
\begin{rem}
We write $\ol \infty$ to distinguish this  “point at infinity” from the points $-\infty$ and $\infty$ of the extended real line. We also agree that any real number $x$ satisfies both $x > \ol \infty$ and $x < \ol \infty$, and therefore interpret $(\ol \infty,x)$ as $(-\infty, x)$ and  $(x, \ol \infty)$ as $(x, \infty)$. We finally agree that $c/0 := \ol \infty$ for $c \in \RR \setminus \{0\}$ so that $\tau_1 = \id$.
\end{rem}

With the help of Möbius transformations, we may define the boundary points of the real no-trade region (which depends on the transaction cost parameters $\gamma^\uparrow$ and $\gamma^\downarrow$ separately) in term of the boundary points of the shadow no-trade region (which depends on the transaction cost parameters $\gamma^\uparrow$ and $\gamma^\downarrow$ only through their product $\xi = \gamma^\uparrow \gamma^\downarrow$).

\begin{defn}
Let  $\gamma^\uparrow, \gamma^\downarrow \in [1, \infty)$ with $\xi:= \gamma^\uparrow \gamma^\downarrow  > 1$ be such that $\xi$ lies in the well-posedness range of transaction costs.  Let $q_*(\xi), q^*(\xi)$ be as in Proposition \ref{prop:n}. Define $p_*(\gamma^\uparrow, \gamma^\downarrow),  p^*(\gamma^\uparrow, \gamma^\downarrow) \in \RR$ by
\begin{align}
p_*(\gamma^\uparrow, \gamma^\downarrow) &:= \tau_{\frac{1}{\gamma^\uparrow}}(q_*(\xi )) = \frac{q_*(\xi)}{\gamma^\uparrow+ (1-\gamma^\uparrow) q_*(\xi)}, \label{eq:def:p lower star}\\
p^*(\gamma^\uparrow, \gamma^\downarrow) &:=  \tau_{\gamma^\downarrow}(q^*(\xi ))  =\frac{\gamma^\downarrow q^*(\xi)}{1 + (\gamma^\downarrow-1) q^*(\xi)}.
\label{eq:def:p upper star}
\end{align}
\end{defn}

Note that $\tau_{\gamma^\uparrow}(p^*(\gamma^\uparrow, \gamma^\downarrow)) = \tau_{\gamma^\uparrow} \circ \tau_{\frac{1}{\gamma^\uparrow}}(q_*(\xi )) = \tau_1 (q_*(\xi )) = q_*(\xi )$ and similarly $\tau_{\frac{1}{\gamma^\downarrow}}(p^*(\gamma^\uparrow, \gamma^\downarrow)) = q^*(\xi)$.

The following result gives some estimates for the lower and upper bound of the no-transaction region.

\begin{lemma}
\label{lem:p low p up}
Let  $\gamma^\uparrow, \gamma^\downarrow \in [1, \infty)$ with $\xi:= \gamma^\uparrow \gamma^\downarrow  > 1$ be such that $\xi$ lies in the well-posedness range of transaction costs. Then
\begin{align}
0 < p_*(\gamma^\uparrow, \gamma^\downarrow) &<q_M < p^*(\gamma^\uparrow, \gamma^\downarrow) < 1, &&\quad\text{if } q_M \in (0, 1)\\
0 < p_*(\gamma^\uparrow, \gamma^\downarrow) &< p^*(\gamma^\uparrow, \gamma^\downarrow) =1, &&\quad\text{if } q_M =1 \\
0 < p_*(\gamma^\uparrow, \gamma^\downarrow) &< p^*(\gamma^\uparrow, \gamma^\downarrow) < \tau_{\gamma^\downarrow}(2q_M-1)  &&\quad\text{if } q_M >1,  \\
\tau_{\frac{1}{\gamma^\uparrow}} (2 q_M) < p_*(\gamma^\uparrow, \gamma^\downarrow) &< p^*(\gamma^\uparrow, \gamma^\downarrow)  < 0, &&\quad\text{if } q_M<0.
\end{align}
Moreover, $1 \in  [p_*(\gamma^\uparrow, \gamma^\downarrow),  p^*(\gamma^\uparrow, \gamma^\downarrow)]$ if and only if $1 \in  [q_*(\xi),  q^*(\xi)]$.
\end{lemma}

\begin{rem}
\label{rem:locationofNTwedge}
(a)  It follows from Lemma~\ref{lem:p low p up} (see also Davis and Norman~\cite[p704]{DavisNorman:90}) that if the Merton line lies in the first quadrant, then the no-transaction region also lies in the first quadrant (i.e., it is never optimal to have a short position in either the riskless asset or the risky asset) and contains the Merton line. A numerical illustration of this case will be given in Section~\ref{ssec:eg1}. However, if $q_M<0$ or if $q_M>1$, then it may be the case that $q_M$ lies outside the interval $[p_*(\gamma^\uparrow, \gamma^\downarrow),  p^*(\gamma^\uparrow, \gamma^\downarrow)]$ and that the Merton line lies \emph{outside} the no-transaction region. As discussed in Remark~\ref{rem:qMnotinwedge} this has to be the case if $\gamma^\downarrow > \frac{q_M}{q_M-1}$. If $q_M>1$, then Davis and Norman~\cite[p704]{DavisNorman:90} conjectured that the no-transaction region lies in the second quadrant, but this need not be the case (although it must always intersect the second quadrant, and indeed we must have $p^*(\gamma^\uparrow, \gamma^\downarrow)>1$). Conversely, Shreve and Soner~\cite[p675]{shreve:soner} conjectured that $q_M>p^*(\gamma^\uparrow, \gamma^\downarrow)$, whenever $q_M>1$, but this is not true for small transaction costs.

(b) If $q_M>1$ (and $\hatm,m(1)>0$), a new phenomenon arises. It follows from Proposition~\ref{prop:n} that if $\xi$ is large enough so that $q_*(\xi) < 1$, then $n_\xi$ passes through the singular point $(1,m(1))$ and on the interval $(1, q^*(\xi))$, $n_\xi$ does not depend on $\xi$ (and hence also $q^*(\xi)$ does not depend on $\xi$). It follows that once $\xi$ is sufficiently large, $p^*(\gamma^\uparrow,\gamma^\downarrow)$ is independent of $\gamma^\downarrow$.
In the additive case, Hobson et al~\cite[Corollary 5.4]{hobson:tse:zhu:19A} give a financial explanation behind this observation, which also applies to the case of stochastic differential utility. In particular, if $1 \in (q_*(\xi),q^*(\xi))$, then it is possible (indeed inevitable) that at some point the agent, who finances consumption from cash wealth, will exhaust holdings of the risk-free asset. At this point they continue consuming, and finance consumption by borrowing. Subsequently they will sell units of the risky asset in order to stop their short position in the riskless asset from getting too large, but they will always have a short position in cash. Put simply the process $Q$ never returns to the set $(q_*(\xi),1]$ once it has left it; see also Theorem \ref{thm:cand solution 1}(a) below. It follows that if the initial holdings involve a short cash position, then the trajectory of the process $(X_t,\Phi_t)$ will be such that $Q_t$ never hits the lower threshold $q_*(\xi)$, the agent will never purchase units of the risky asset, and the transaction costs on purchases is irrelevant.
\end{rem}

We proceed to describe existence, uniqueness and further properties of the shadow fraction of wealth in terms of the real fraction of wealth.

\begin{prop}
\label{prop:q of p}
Let  $\gamma^\uparrow, \gamma^\downarrow \in [1, \infty)$ with $\xi:= \gamma^\uparrow \gamma^\downarrow  > 1$ be such that $\xi$ lies in the well-posedness range of transaction costs.  Then there exists a unique, increasing, continuously differentiable function  $q_{\gamma^\uparrow, \gamma^\downarrow} :  [p_*(\gamma^\uparrow, \gamma^\downarrow), p^*(\gamma^\uparrow, \gamma^\downarrow)] \to [q_*(\xi), q^*(\xi)]$ with the following properties:
\begin{enumerate}
\item on $[p_*(\gamma^\uparrow, \gamma^\downarrow), p^*(\gamma^\uparrow, \gamma^\downarrow)] \setminus \{1\}$,  $q_{\gamma^\uparrow, \gamma^\downarrow}$ satisfies the ODE
\begin{equation}
\label{eq:prop:q of p}
q'(p) = \frac{\ell(q(p)) - n_{\xi}(q(p))}{\ell(p) - m(p)}.
\end{equation}
\item $q_{\gamma^\uparrow, \gamma^\downarrow}$ satisfies the boundary conditions
\begin{align}
q_{\gamma^\uparrow, \gamma^\downarrow}(p_*(\gamma^\uparrow, \gamma^\downarrow)) = \tau_{\gamma^\uparrow}(p_*(\gamma^\uparrow, \gamma^\downarrow))  = q_*(\xi) \quad &\text{and} \quad q_{\gamma^\uparrow, \gamma^\downarrow}(p^*(\gamma^\uparrow, \gamma^\downarrow)) = \tau_{\frac{1}{\gamma^\downarrow}}(p^*(\gamma^\uparrow, \gamma^\downarrow))  = q^*(\xi), \\
q'_{\gamma^\uparrow, \gamma^\downarrow}(p_*(\gamma^\uparrow, \gamma^\downarrow)) = \tau'_{\gamma^\uparrow}(p_*(\gamma^\uparrow, \gamma^\downarrow))  \quad &\text{and} \quad q'_{\gamma^\uparrow, \gamma^\downarrow}(p^*(\gamma^\uparrow, \gamma^\downarrow)) = \tau'_{\frac{1}{\gamma^\downarrow}}(p^*(\gamma^\uparrow, \gamma^\downarrow)).
\end{align}
\item  If $1 \in [p_*(\gamma^\uparrow, \gamma^\downarrow), p^*(\gamma^\uparrow, \gamma^\downarrow)]$, then $q_{\gamma^\uparrow, \gamma^\downarrow}(1) = 1$

\end{enumerate}
Moreover, if $\tilde \gamma^\downarrow,\tilde \gamma^\uparrow \in [1, \infty)$ are such that $\tilde \gamma^\downarrow\tilde \gamma^\uparrow = \xi$, then
\begin{equation}
q_{\tilde \gamma^\downarrow,\tilde \gamma^\uparrow} = q_{\gamma^\uparrow, \gamma^\downarrow} \circ \tau_{\frac{\tilde \gamma^\downarrow}{\gamma^\downarrow}} = q_{\gamma^\uparrow, \gamma^\downarrow} \circ \tau_{\frac{\gamma^\uparrow}{\tilde \gamma^\uparrow}}.
\end{equation}
\end{prop}

\subsection{The candidate value function and candidate optimal strategy}
In order define the candidate value function on the whole of $\sS^\circ_{\gamma^\uparrow, \gamma^\downarrow}$ (and not just inside the no-trade region), we need to extend $n$, $\kappa$ and $q$. The idea is to extend $n$ and $\kappa$ so that they are constant outside the no-transaction region. This in turn determines the extension of $q$.

\begin{defn}
\label{def:extension}
Let  $\gamma^\uparrow, \gamma^\downarrow \in [1, \infty)$ with $\xi:= \gamma^\uparrow \gamma^\downarrow  > 1$ be such that $\xi$ lies in the well-posedness range of transaction costs.
\begin{enumerate}
\item Define the function $\ol n_\xi: (-\infty, \infty) \to (0, \infty)$ by
\begin{equation*}
\ol n_\xi(q) :=
\begin{cases}
n_\xi(q_*(\xi)) = m(q_*(\xi)) &\text{if } q \in (-\infty, q_*(\xi)), \\
n_\xi(q) &\text{if } q \in [q_*(\xi), q^*(\xi)], \\
n_\xi(q^*(\xi)) = m(q^*(\xi)) &\text{if } q \in (q^*(\xi), \infty). \\
\end{cases}
\end{equation*}

\item Define the function $\ol \kappa_{\gamma^\uparrow, \gamma^\downarrow}: (-\infty, \infty) \to [\frac{1}{\gamma^\downarrow},
\gamma^\uparrow]$ by
\begin{equation*}
\ol \kappa_{\gamma^\uparrow, \gamma^\downarrow}(q) :=
\begin{cases}
\kappa_{\gamma^\uparrow, \gamma^\downarrow}(q_*(\xi)) = \gamma^\uparrow &\text{if } q \in (-\infty, q_*(\xi)), \\
\kappa_{\gamma^\uparrow, \gamma^\downarrow}(q) &\text{if } q \in [q_*(\xi), q^*(\xi)], \\
\kappa_{\gamma^\uparrow, \gamma^\downarrow}(q^*(\xi)) = \frac{1}{\gamma^\downarrow} &\text{if } q \in (q^*(\xi), \infty). \\
\end{cases}
\end{equation*}
\item Define the function $\ol q_{\gamma^\uparrow, \gamma^\downarrow}: (\tau_{\frac{1}{\gamma^\uparrow}}(\ol \infty), \tau_{\gamma^\downarrow}(\ol \infty)) = (\frac{-1}{\gamma^\uparrow-1}, \frac{\gamma^\downarrow}{\gamma^\downarrow-1}) \to \RR$ by
\begin{equation*}
\ol q_{\gamma^\uparrow, \gamma^\downarrow}(p) :=
\begin{cases}
\tau_{\gamma^\uparrow}(p) &\text{if } p \in (\tau_{\frac{1}{\gamma^\uparrow}(\ol \infty)}, p_*(\gamma^\uparrow, \gamma^\downarrow)), \\
q_{\gamma^\uparrow, \gamma^\downarrow}(p) &\text{if } p \in [p_*(\gamma^\uparrow, \gamma^\downarrow), p^*(\gamma^\uparrow, \gamma^\downarrow)]. \\
\tau_{\frac{1}{\gamma^\downarrow}}(p) &\text{if } p \in (p^*(\gamma^\uparrow, \gamma^\downarrow), \tau_{\gamma^\downarrow}(\ol \infty)), \\
\end{cases}
\end{equation*}

\end{enumerate}

\end{defn}

We can now define candidate value function on the whole of $\sS^\circ_{\gamma^\uparrow, \gamma^\downarrow}$.


\begin{defn}
Let  $\gamma^\uparrow, \gamma^\downarrow \in [1, \infty)$ with $\xi:= \gamma^\uparrow \gamma^\downarrow  > 1$ be such that $\xi$ lies in the well-posedness range of transaction costs. Define the value function $\hat V^{\gamma^\uparrow, \gamma^\downarrow}:  [0, \infty) \times \sS^\circ_{\gamma^\uparrow, \gamma^\downarrow} \to (1-R) (0, \infty)$ by
\begin{equation}
\label{def:hat V}
\hat V^{\gamma^\uparrow, \gamma^\downarrow}(t, x, y, \phi) := e^{- \delta \theta t} \frac{\Big(x + \phi y \ol \kappa_{\gamma^\uparrow, \gamma^\downarrow}\left(\ol q_{\gamma^\uparrow, \gamma^\downarrow}\left(\frac{\phi y}{x+ \phi y}\right)\right)\Big)^{1-R}}{1-R}  \ol n_{\xi}\left(\ol q_{\gamma^\uparrow, \gamma^\downarrow}\left(\frac{\phi y}{x+ \phi y}\right)\right)^{-\theta S}.
\end{equation}
\end{defn}

{We proceed to turn to the candidate optimal strategy. Since we consider all parameter combinations (in particular, $1 \in [p_*(\gamma^\uparrow, \gamma^\downarrow), p^*(\gamma^\uparrow, \gamma^\downarrow)]$ is allowed), this is quite delicate. Even in the additive case, the following result is not covered by standard results in the extant literature (see e.g. \cite{shreve:soner, MMKS:20}). For this reason, we provide a self-contained proof in Appendix~\ref{ssecapp:valuefn}. Note that our proof only uses results on reflected SDEs in one dimension and is therefore technically substantially easier than the results in the extant literature (including \cite{shreve:soner, MMKS:20}), which rely on reflection result in a two dimensional domain.

\begin{thm}
\label{thm:cand solution 1}
Let $\gamma^\uparrow, \gamma^\downarrow \in [1, \infty)$ with $\xi:= \gamma^\uparrow \gamma^\downarrow  > 1$ be such that $\xi$ is in the well-posedness range of transaction costs. Let $(x, y, \phi) \in  \sS^\circ_{\gamma^\uparrow, \gamma^\downarrow}$.

\begin{enumerate}
\item
There exist unique processes $(\hQ,G=G^\uparrow - G^\downarrow)$ with continuous paths and with initial values $\hQ_0 = \min(\max(q(\frac{\phi y}{x + \phi y}), q_*(\xi)), q^*(\xi))$, $G^\uparrow_0=0=G^\downarrow_0$, such that the process $\hat Q = (\hat Q_t)_{t \geq 0}$ defined by
\begin{equation}
\label{eq:dQdef}
d \hat{Q}_t = \frac{2S}{\sigma(1-S)} (\ell(\hQ_t)-n_\xi(\hQ_t)) dB_t + \hQ_t \left( n_\xi(\hQ_t) - 2 \theta S \{\ell(\hQ_t)-n_\xi(\hQ_t) \} \right) \dd t + \dd G^\uparrow_t - \dd G^\downarrow_t
\end{equation}
takes values in $[q_*(\xi), q^*(\xi)]$, and $G^\uparrow, G^\downarrow$ are nondecreasing and increase only when $\hQ_t = q_*(\xi)$ and $\hQ_t = q^*(\xi)$ respectively. Moreover, if $1 \in [q_*(\xi), q^*(\xi)]$, let $\tau := \inf \{t \geq 0: \hat Q_t \geq 1\}$. Then $\hat Q \in [q^*(\xi), 1)$ on
$\llbracket 0, \tau \llbracket$, and $\hat Q = 1$ on $\llbracket \tau, \infty \llbracket$ if $q^*(\xi) = 1$ and $\hat Q \in (1, q_*(\xi)]$ on $\rrbracket \tau, \infty \llbracket$ if $q^*(\xi) >1$.

\item
Set
\begin{equation}
\hPhi_0 :=
\begin{cases}
q_*(\xi)\left( \phi + \frac{x}{ y\gamma^\uparrow} \right) &\text{if} \quad \phi y \gamma^\uparrow (1 - q_*(\xi)) < q_*(\xi) x, \\
\phi & \text{if} \quad \phi y \gamma^\uparrow (1 - q_*(\xi)) \geq q_*(\xi) x \text{ and } \frac{\phi y}{\gamma^\downarrow} (1 - q^*(\xi) \leq q^*(\xi) x,\\
q^*(\xi) \left( \phi + \frac{x \gamma^\downarrow}{y} \right) &\text{if}  \quad \frac{\phi y}{\gamma^\downarrow} (1 - q^*(\xi)) > q^*(\xi) x,
\end{cases}
\end{equation}
and the processes $\hat \Phi = (\hPhi_t)_{t \geq 0}$, $\hat C =(\hC_t)_{t \geq 0}$ and $\hX = (\hX_t)_{t \geq 0}$  by
\begin{align}
\hPhi_t &:= \hPhi_0 \exp \left( \frac{G^\uparrow_t}{q_*} - \frac{G^\downarrow_t}{q^*} \right), \\
\hat C_t &:= \frac{Y_t \kappa(\hQ_t) \hPhi_t}{\hQ_t}  n_\xi(\hQ_t), \\
\label{eq:Xdef1}
\hX_t &:= \frac{1-\hQ_t}{\hQ_t} Y_t \kappa(\hQ_t) \hPhi_t .
\end{align}
Then
\begin{enumerate}
\item
 $(\hat \Phi, \hat C) \in \sA_{\gamma^\uparrow, \gamma^\downarrow}(x, y, \phi)$,
\item $\hX = X^{x, y, \phi, \hat \Phi, \hat C}$ and
\item $\hQ$ satisfies
\begin{equation}
\label{eq:Qselfconsistent}
\hQ_t = \frac{ Y_t  \kappa_{\gamma^\uparrow,\gamma^\downarrow}(\hQ_t) \hPhi_t }{\hX_t + Y_t  \kappa_{\gamma^\uparrow,\gamma^\downarrow}(\hQ_t) \hPhi_t}.
\end{equation}
\end{enumerate}
\end{enumerate}
\end{thm}

\begin{rem}
(a) The case $ \phi y \gamma^\uparrow (1 - q_*(\xi)) < q_* (\xi) x$ corresponds to starting in the region where the trading process $\Phi$ involves an initial purchase of the risky asset, and the case $\frac{\phi y}{\gamma^\downarrow} (1 - q^*(\xi)) > q^*(\xi) x$ corresponds to the case where the trading process involves an initial sale of the risky asset, and in each case the expression for $\hX_0$ given by \eqref{eq:Xdef1} evaluated at $t=0$ reflects the impact of the trade. Implicit in the statement is the fact that we cannot have both $ \phi y \gamma^\uparrow (1 - q_*(\xi)) < q_* (\xi)x$ and $\frac{\phi y}{\gamma^\downarrow} (1 - q^*(\xi)) > q^*(\xi) x$.

(b) We can have $ 1 \in (q_*(\xi),q^*(\xi))$. Then  the diffusion coefficient in \eqref{eq:dQdef} vanishes and the drift coefficient is positive for $\hat Q = 1$. In particular, the diffusion $\hQ$ is not regular for the domain $(q_*(\xi),q^*(\xi))$. If $\hQ_s \geq 1$ then $\hQ_t >1$ for $t > s$. This is equivalent to the fact that if $\hX_s \leq 0$ then for $t>s$, we have $\hX_t<0$, and the agent always takes a short position in the risky asset, see also Remark \ref{rem:locationofNTwedge}(b).

(c) We could instead start with dynamics for the candidate for $P = \frac{Y \Phi}{X+ Y \Phi}$. These are given by:
\[ \frac{\dd \hat P_t}{\hat P_t (1- \hat P_t)} = \left((\mu- r)  - \sigma^2 \hat P_t +  \frac{ n_\xi(q_{\gamma^\uparrow, \gamma^\downarrow}(\hat P_t))}{1- q_{\gamma^\uparrow, \gamma^\downarrow}(\hat P_t)} \right) \dd t
+ \sigma  \dd B_t +
\frac{dG_t}{q_{\gamma^\uparrow, \gamma^\downarrow}(\hat P_t) (1- q_{\gamma^\uparrow, \gamma^\downarrow}(\hat P_t))}.
\]
However, these dynamics are more complicated than those of $\hat{Q}$, not least because they involves expressions like $n_\xi(q_{\gamma^\uparrow, \gamma^\downarrow}( p))$. Moreover, for $\hat P = 1$, the right hand side of the SDE is not well-defined.
\end{rem}

We now link the candidate value function with the candidate optimal strategy.

\begin{thm}
\label{thm:cand solution 2}
Let  $\gamma^\uparrow, \gamma^\downarrow \in [1, \infty)$ with $\xi:= \gamma^\uparrow \gamma^\downarrow  > 1$ be such that $\xi$ is in the well-posedness range of transaction costs. Let $(x, y, \phi) \in  \sS^\circ_{(\gamma^\uparrow, \gamma^\downarrow)}$ and $(\hat \Phi, \hat C) \in \sA_{(\gamma^\uparrow, \gamma^\downarrow)}(x, y, \phi)$ be the corresponding candidate optimal strategy from Theorem \ref{thm:cand solution 1}. Then $\hat C$ is uniquely evaluable if $\theta \leq 1$ and uniquely proper with right-continuous paths if $\theta > 1$. Moreover,
\begin{align}
\hat V^{\gamma ^\uparrow, \gamma^\downarrow}(t, X^{x, y, \phi, \hat \Phi, \hat \Phi}_t, Y_t, \hat C_t) &= V^{\hat C}_t\;\; \as{\mathbb P}, \quad t \geq 0,
\end{align}
where $V^{\hat C}$ denotes the unique utility process if  $\theta \leq 1$ and the unique proper utility process if $\theta >1$.
\end{thm}

\subsection{Verification argument}
In order to complete the verification, we would like to use something like the martingale optimality principle, which in the additive case ($R=S$, $\rho=0$) says that $M$ given by $M_t = \int_0^t e^{-\delta s} \frac{C_s^{1-S}}{1-S} \dd s + \hat{V}(t,X_t,Y_t,\Phi_t)$ is always a supermartingale and a martingale for the optimal strategy. In the general case with $R \neq S$, classical martingale optimality arguments break down due to the fact that the value process enters the integrand. Instead we need a comparison theorem for solutions of BSDEs.

To this end, we recall the notions of a subsolution and a supersolution of a BSDE. Definitions of sub- and supersolutions of BSDEs in the literature can vary slightly depending on the context. The definition we use here is the one used in Herdegen et al~\cite[Definition 5.3]{herdegen:hobson:jerome:23A}, specialised to the case of Epstein-Zin SDU.

\begin{defn}
			Let $C \in \sP_+$ be a consumption stream. A $(1-R)\R_+$-valued, l\`ad, optional process $V$ is called
			\begin{itemize}
                \item  a \textit{subsolution} for $C$  if $\limsup_{t\to\infty}~ \EX{V_{t+}} \leq 0$ and for all bounded stopping times ${\tau_1}\leq{\tau_2}$,
				\begin{equation}\label{eq:subsolution equation}
				V_{\tau_1} ~\leq~ \cEX[{\tau_1}]{V_{{\tau_2}{+}} +  \int_{\tau_1}^{{\tau_2}} e^{-\delta s} \frac{C_s^{1-S}}{1-S} ((1-R)V_s)^\rho \dd s}.
				\end{equation}
				\item   a \textit{supersolution} for $C$  if $\liminf_{t\to\infty}~ \EX{V_{t+}} \geq 0$ and for all bounded stopping times ${\tau_1}\leq{\tau_2}$,
				\begin{equation}\label{eq:supersolution equation}
				V_{\tau_1} ~\geq~ \cEX[{\tau_1}]{V_{{\tau_2}{+}} +  \int_{\tau_1}^{{\tau_2}} e^{-\delta s} \frac{C_s^{1-S}}{1-S} ((1-R)V_s)^\rho \dd s}.
				\end{equation}
                \item a \emph{solution} for $C$ if it is both a subsolution and a supersolution and $\E[\int_0^\infty C_s^{1-S} ((1-R)V_s)^\rho \dd s] < \infty$.
			\end{itemize}
		\end{defn}
}
The next result essentially shows that for any $(\hat \Phi, \hat C) \in \sA_{\gamma^\uparrow, \gamma^\downarrow}(x, y, \phi)$ the process $(\hat V^{\gamma ^\uparrow, \gamma^\downarrow}(t, X_t, Y_t, \Phi_t))_{t \geq 0}$ is a supersolution for $C$. However, for technical reasons, this is more delicate and a stochastic perturbation argument (as in the classical Merton case, see \cite{herdegen2020elementary}) is needed.
\begin{prop}
\label{prop:ver:supersol}
Let  $\gamma^\uparrow, \gamma^\downarrow \in [1, \infty)$ and suppose that $\xi := \gamma^\uparrow \gamma^\downarrow \in (1, \infty)$ lies in the well-posedness range of transaction costs. Let $(x, y, \phi) \in  \sS^\circ_{\gamma^\uparrow, \gamma^\downarrow}$ and $(\hat \Phi, \hat C) \in \sA_{\gamma^\uparrow, \gamma^\downarrow}(x, y, \phi)$ be the corresponding candidate optimal strategy from Theorem \ref{thm:cand solution 1}. Let $(\Phi, C)\in \sA_{\gamma^\uparrow, \gamma^\downarrow}(x, y, \phi)$ and set $X :=  X^{x,y,\phi,\Phi, C}$ and  $\hat X :=  X^{x,y, \phi, \hat \Phi, \hat C}$. Then for every $\epsilon > 0$,
$(\hat V^{\gamma ^\uparrow, \gamma^\downarrow}(t, X_t+\epsilon \hat X_t, Y_t, \Phi_t + \epsilon \hat \Phi_t))_{t \geq 0}$ is a supersolution for $C+\epsilon \hat C$.
\end{prop}

We now can formulate a rigorous verification theorem. 
\begin{thm}
\label{thm:verification}
Let $\gamma^\uparrow, \gamma^\downarrow \in [1, \infty)$ and suppose that $\xi := \gamma^\uparrow \gamma^\downarrow \in (1, \infty)$ lies in the well-posedness range of transaction costs. Let $(x, y, \phi) \in  \sS^\circ_{\gamma^\uparrow, \gamma^\downarrow}$ and $\hat C \in \sC^*_{\gamma^\uparrow, \gamma^\downarrow}(x, y, \phi)$ be the corresponding candidate optimal strategy. Then
\begin{equation}
\sup_{C \in \sC^*_{\gamma^\uparrow, \gamma^\downarrow}(x, y, \phi)} V^C_0 = V^{\hat C}_0 = \hat V^{\gamma^\uparrow, \gamma^\downarrow}(0, x, y, \phi),
\end{equation}
where $\hat V^{\gamma^\uparrow, \gamma^\downarrow}$ is the candidate value function from \eqref{def:hat V}.

\end{thm}

\subsection{The ill-posed case}
\label{subsec:ill-posed}
We end this section by showing that if $\gamma^\uparrow \gamma^\downarrow  \in (1, \infty)$ does not lie in the well-posedness rage of transaction costs, the problem \eqref{eq:frictional problem} is ill posed. 

\begin{prop}
\label{prop:ill posed:theta<1}
Fix $(x, y, \phi) \in \mathscr{S}^\circ_{\gamma^\uparrow, \gamma^\downarrow}$ and let $\gamma^\uparrow,\gamma^\downarrow\in[1,\infty)$ be such that $\xi := \gamma^\uparrow\gamma^\downarrow > 1$. Suppose that $\xi$ does not lie in the well-posedness range of transaction costs given in Definition \ref{def:well posed}.

\begin{enumerate}[(a)]
	\item If $S \leq R<1$ or $R<S<1$, then there exists a sequence of consumption streams $(C^n)_{n \in \NN} \in\mathscr{C}^*_{\gamma^\uparrow, \gamma^\downarrow}(x, \phi, y)$ such that $\lim_{n \to \infty} V^{C^n}_0= +\infty$.
	
	\item If $S \geq R >1$ or $R>S>1$, then $V^{C}_0=-\infty$ for all  $C\in\mathscr{C}^*_{\gamma^\uparrow, \gamma^\downarrow}(x, \phi, y)$.
\end{enumerate}

\end{prop}

\section{Diagramatic representations and numerical illustrations}
\label{sec:examples}

In this section, we provide two numerical examples to illustrate some important qualitative properties of the solution, as represented by the function $n_\xi$ and the limits of the no-transaction region $\{q_*(\xi),q^*(\xi)\}$ (in the shadow-quantity of wealth coordinates) and $\{p_*(\gamma^\uparrow,\gamma^\downarrow),p^*(\gamma^\uparrow,\gamma^\downarrow)\}$ in the original coordinates. The examples are not meant to be exhaustive, and indeed several further canonical cases can arise depending on the signs of $\{1-R, q_M, 1-q_M \}$, and the locations of the quadratics $m(q)$ and $\ell(q)$. Some of the cases which are omitted here are treated more comprehensively in Hobson et al~\cite{hobson:tse:zhu:19A,hobson:tse:zhu:19B} in the special case of CRRA utility.

\subsection{First Numerical Example: $R,S<1$ and $q_M \in (0,1)$}
\label{ssec:eg1}
The model parameters are chosen such that $R,S<1$, $m(0)>0$, $m(1)>0$, $\min_{q\in(0,1)} m(q) <0$ and $q_M\in(0,1)$\footnote{Economically this means that the investement-consumption problem is well-posed when only investment in the risk-free asset is allowed, and also when only investment in the risky-asset is allowed, but the problem is ill-posed when transaction costs are zero and investment in both assets is permitted. Further, in the frictionless case, the optimal investment strategy involves a long position in both the risky asset and in cash.}. See the caption of Figure \ref{fig:eg1} for the precise values of each parameter used in the numerics. In this case, the problem is only well-posed for sufficiently large transaction costs (recall Definition \ref{def:well posed}). The solutions to the family of ODEs $(n_{(z)}(q))_{q\geq z}$ introduced in part (a) of Proposition \ref{prop:n first:R<1} are only defined for $z\in(0,q^m_-)$, where $q_-^m$ is the smaller root of $m(q)=0$. Each solution $(n_{(z)}(q))_{q\geq z}$ starts at $(z,m(z))$ and crosses $m$ again at some $(\zeta^+(z),m(\zeta^+(z)))$. Different choices of the left-starting point $z$ correspond to different levels of total transaction costs $\xi=\gamma^\uparrow\gamma^\downarrow$. The correct choice of $z$ can be identified from the equation $\Sigma^+(z)=\gamma^\uparrow\gamma^\downarrow$, where $\Sigma^{+}(\cdot)$ is defined in Proposition \ref{prop:n first:R<1}. Then the purchase and sale boundaries (in shadow portfolio weight) are given by $q_*=z$ and $q^*=\zeta^+(z)$. For example, Figure \ref{fig:eg1_n} shows two possible solutions to the ODE for $n$ with different left-starting points and in turn different points of crossing with $m$. (Note, such solutions cannot cross, as they solve a first-order ODE with regular co-efficients, at least away from 1.) They correspond to the total transaction cost levels of $\xi=1.69$ and $\xi=2.1125$ respectively. Note that if the transaction costs are such that $\xi \leq 1.1057$ then the problem is ill-posed.

Consider three combinations of transaction costs $(\gamma^\uparrow, \gamma^\downarrow)$ given by $(1.3, 1.3)$, $(1.625, 1.3)$ and $(1.625,1.04)$ respectively. They are constructed such that Set 1 and Set 3 give the same value of $\xi=1.69$ whereas Set 2 gives $\xi=2.1125$. Fundamental quantities of the optimal solution under each set of transaction costs are shown in Figure \ref{fig:eg1}. The dotted lines in each plot indicate the extension of the fundamental quantities beyond the no-transaction region (see Definition \ref{def:extension}). The cross, circle and square markers locate the boundaries of the no-transaction region under Set 1, 2 and 3 of transaction costs respectively.

In Figure \ref{fig:eg1_n}, the transaction costs parameters of Set 1 and Set 3 yield identical optimal consumption rates per unit shadow wealth since $n(\cdot)$ and $q_*^*$ only depend on the total transaction costs $\xi$. Set 2 carries a higher total transaction cost and then it follows that the optimal consumption rate (in shadow wealth terms) is higher (and the no-transaction interval $(q_*,q^*)$ is wider). In contrast, Figure \ref{fig:eg1_kappa} and \ref{fig:eg1_pq} show that the three sets of transaction costs result in different ratios of shadow-to-real price as well as mappings between real and shadow portfolio weights . This is because they are affected by $\gamma^\uparrow$ and $\gamma^\downarrow$ individually, not just via the total transaction costs $\xi=\gamma^\uparrow\gamma^\downarrow$. These two quantities are monotonic with respect to the shadow portfolio weight.

Figure \ref{fig:eg1_tcost} shows how the no-transaction region (in both real and shadow portfolio weights) changes with the transaction costs. In this case, $p_*$ and $q_*$ (resp. $p^*$ and $q^*$) are decreasing (resp. increasing) in both $\gamma^\uparrow$ and $\gamma^\downarrow$. But notably, Set 1 and 3 of the transaction costs will result in the same levels of $q_*$ and $q^*$, as indicated by the same levels of the cross and square markers in Figure \ref{fig:eq1_q_buy} and \ref{fig:eq1_q_sell}. In terms of real portfolio weights, Set 1 has a larger $p_*$ and $p^*$ than Set 3. This is due to the fact that $p_*$ and $p^*$ depend on $\gamma^\uparrow$ and $\gamma^\downarrow$ individually.
Note as well that in this example (with Merton ratio below one, i.e. $q_M \in (0,1)$) the no-transaction region $[p_*,p^*]$ contains the Merton ratio for all level of transaction costs, and is always contained in the first quadrant (i.e. $0<p_*<q_M<p^*<1$), recall Remark~\ref{rem:locationofNTwedge}.

\begin{figure}[!htbp]
	\captionsetup[subfigure]{width=0.45\textwidth}
	\centering
	\subcaptionbox{Optimal consumption (in shadow wealth) $n(q)$.\label{fig:eg1_n}}{\includegraphics[scale =0.5] {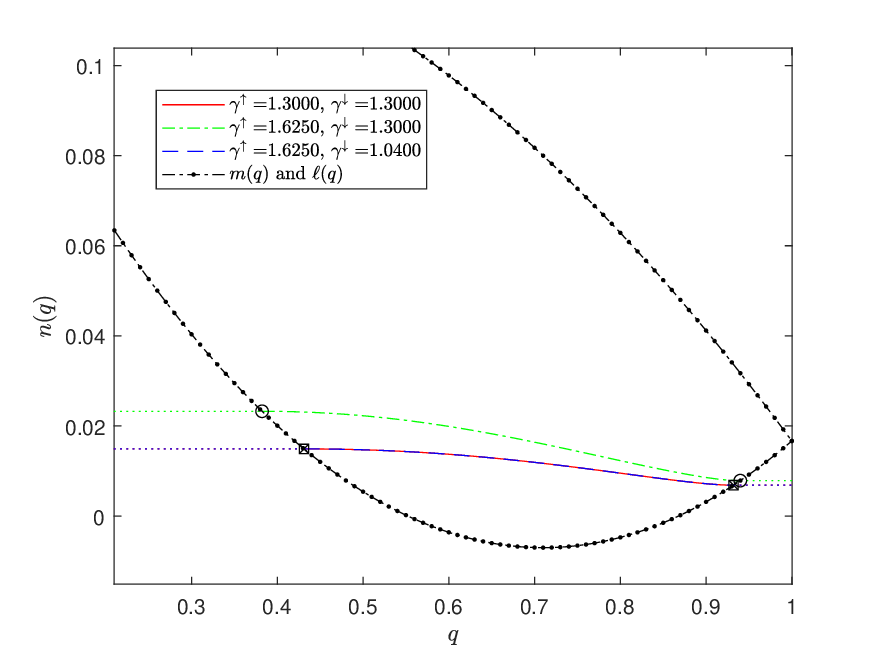}}
	\subcaptionbox{Ratio of shadow-to-real prices $\kappa(q)$.\label{fig:eg1_kappa}}{\includegraphics[scale =0.5] {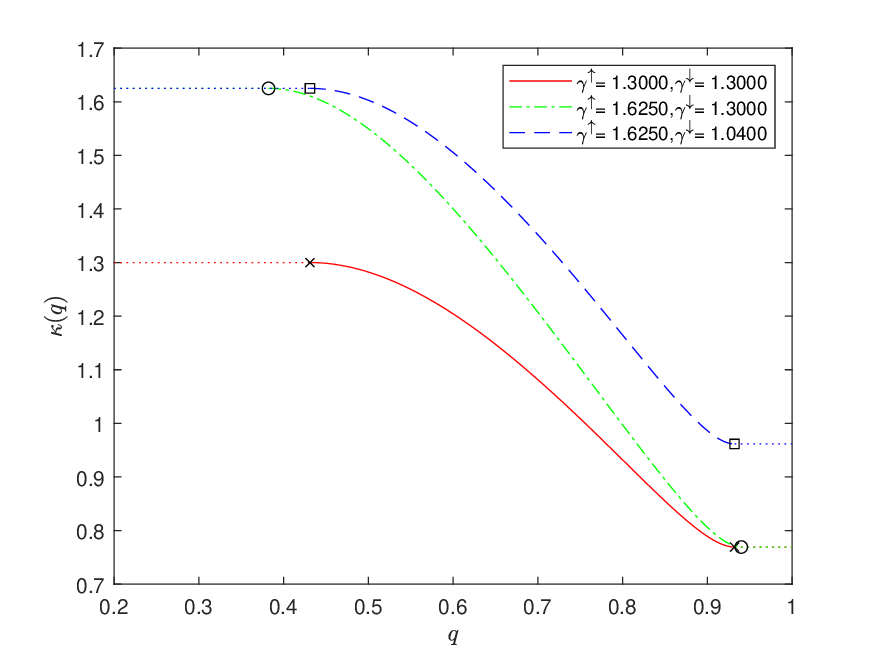}}
	\subcaptionbox{Real and shadow portfolio weight $p(q)$.\label{fig:eg1_pq}}{\includegraphics[scale =0.5] {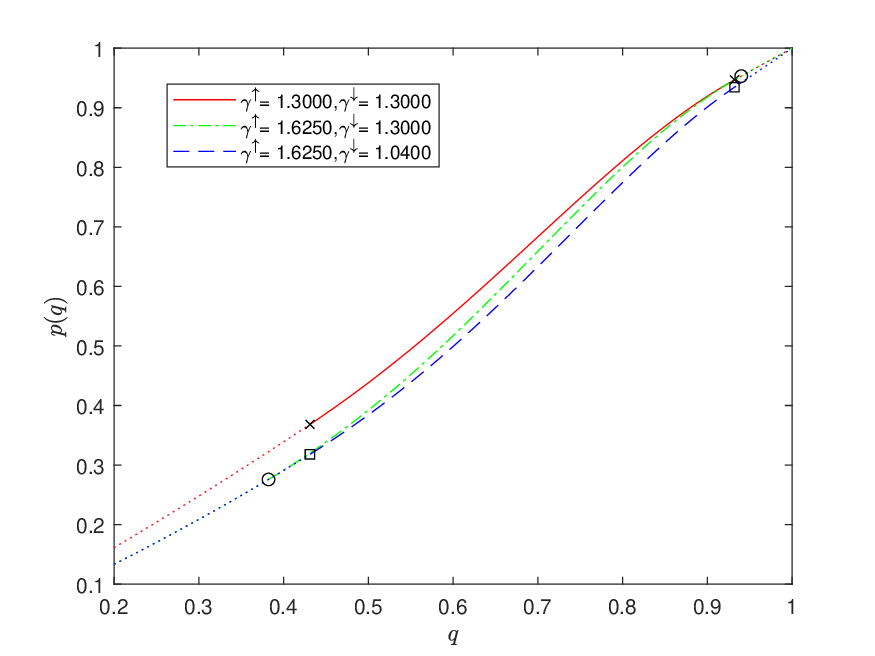}}
	
	\caption{The fundamental quantities as a function of shadow portfolio weights in Example 1. Parameters used: $R=2/3$, $S=1/3$, $\delta=0.045$, $r=0$, $\mu=0.2$, $\sigma=0.65$ (such that $\lambda=0.3077$ and $\alpha=0.045$). The cross, circle and square markers indicate the boundary points of the no-transaction region (in $p_*^*$ or $q_*^*$) under the set of transaction costs $(\gamma^\uparrow, \gamma^\downarrow)$ given by $(1.3, 1.3)$, $(1.625, 1.3)$ and $(1.625, 1.04)$ respectively.}
\label{fig:eg1}
\end{figure}

\begin{figure}[!htbp]
	\captionsetup[subfigure]{width=0.45\textwidth}
	\centering
	\subcaptionbox{$q^*$ and $q_*$ as $\gamma^\uparrow$ varies.\label{fig:eq1_q_buy}}{\includegraphics[scale =0.5] {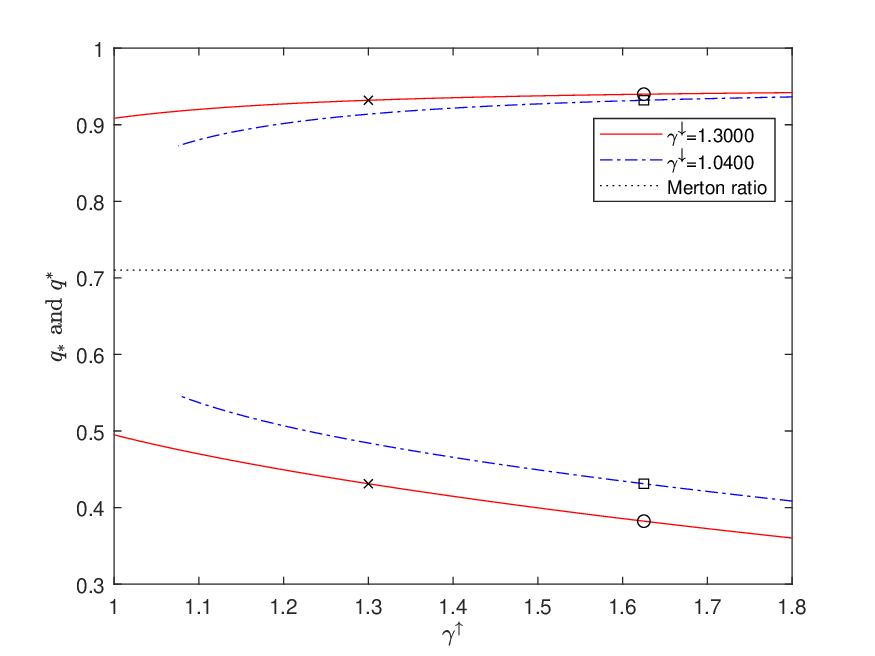}}
	\subcaptionbox{$q^*$ and $q_*$ as $\gamma^\downarrow$ varies.\label{fig:eq1_q_sell}}{\includegraphics[scale =0.5] {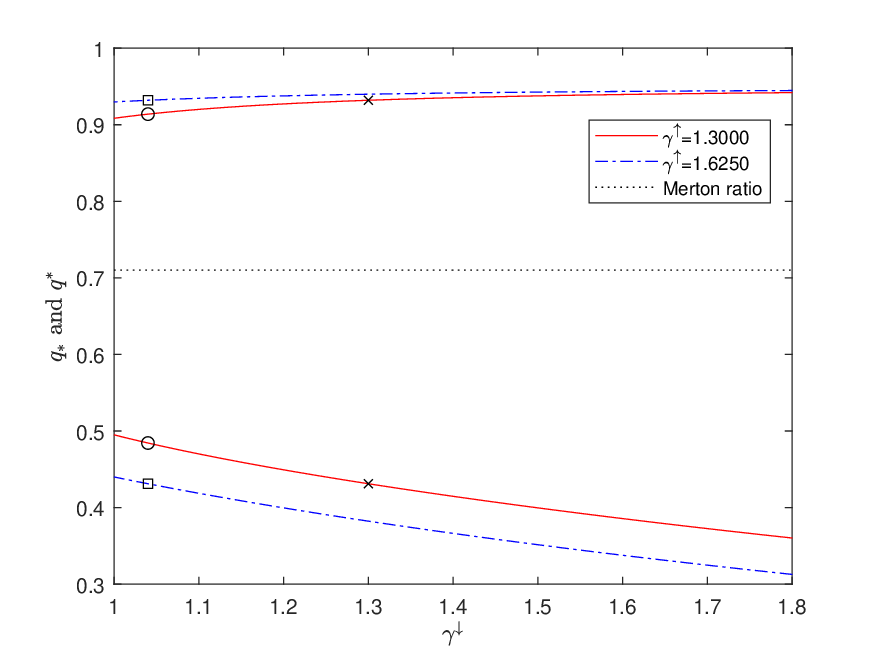}}
	\subcaptionbox{$p^*$ and $p_*$ as $\gamma^\uparrow$ varies.\label{fig:eq1_p_buy}}{\includegraphics[scale =0.5] {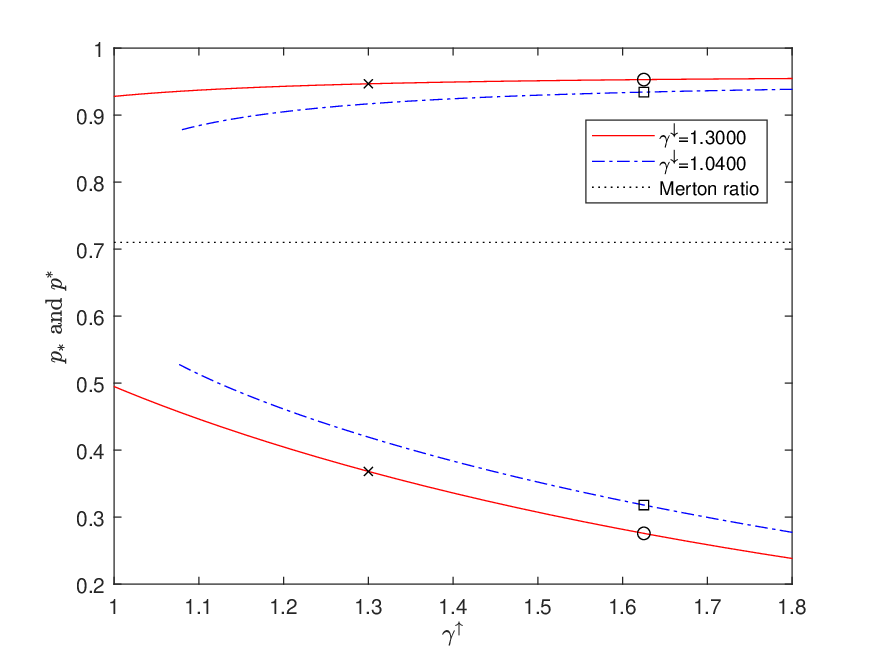}}
	\subcaptionbox{$p^*$ and $p_*$ as $\gamma^\downarrow$ varies.\label{fig:eq1_p_sell}}{\includegraphics[scale =0.5] {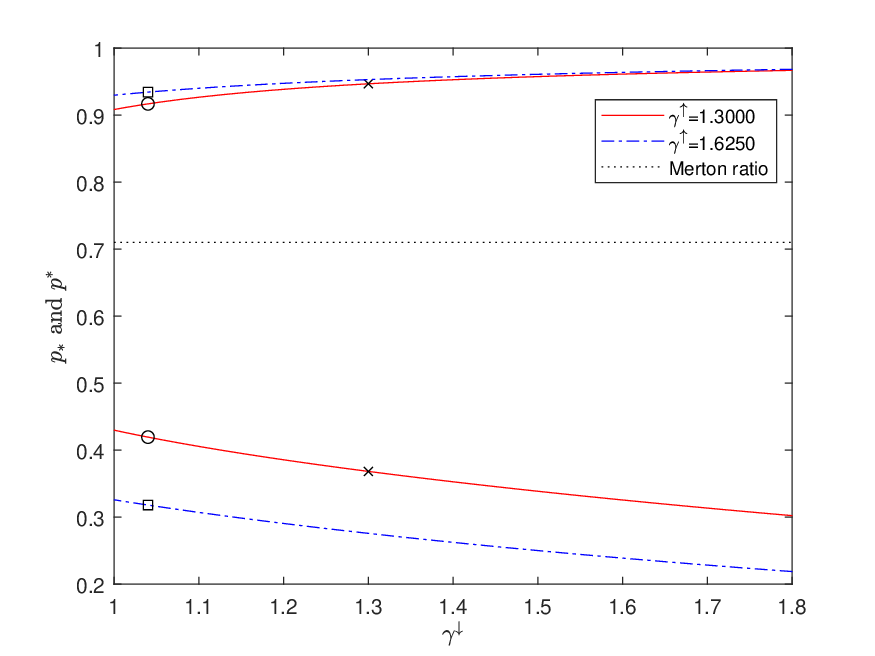}}
	
	\caption{Behaviours of the no-transaction region as a function of transaction costs in Example 1. Parameters used: $R=2/3$, $S=1/3$, $\delta=0.045$, $r=0$, $\mu=0.2$, $\sigma=0.65$ (such that $\lambda=0.3077$ and $\alpha=0.045$). The cross, circle and square markers indicate the boundary points of the no-transaction region (in $p_*^*$ or $q_*^*$) under the set of transaction costs $(\gamma^\uparrow, \gamma^\downarrow)$ given by $(1.3, 1.3)$, $(1.625, 1.3)$ and $(1.625, 1.04)$ respectively.}
\label{fig:eg1_tcost}
\end{figure}

Several qualitative features are also apparent from Figures~\ref{fig:eg1} and \ref{fig:eg1_tcost} which may be traced back to the fact that $n$ is decreasing and solutions $n$ with different starting points get closer as $q$ increases. In particular we observe that $q_*$ varies more than $q^*$ as $\xi$ changes, and whilst $\gamma^\uparrow$ affects both $p_*$ and $p^*$, $\gamma^\downarrow$ mainly affects $p_*$.

\subsection{Second Numerical Example: $R,S>1$ and $q_M>1$}
\label{ssec:eg2}

In this example, we consider a combination of model parameters\footnote{The parameters are such that the problem is well-posed for any value of transaction costs, and both when only investments in the risk-free asset are allowed and when only investments in the risky asset are allowed. Further, in the frictionless case, the optimal investment strategy involves borrowing to finance a leveraged position in the risky asset.} such that $R,S>1$, $m(q)>0$ for $q\in[0,1]$ and $q_M>1$. Specifically, the Merton ratio is larger than one and the problem is well-posed for any level of transaction costs. An interesting feature in the case of the Merton ratio being above one is that all solutions to the family of ODEs in part (c) of Proposition \ref{prop:n first:R>1} with left-starting points below one will pass through the singular point $(1,m(1))$. Consequently, for any $z<1$, we must have $n_{[z]}(q)=n_{[1]}(q)$ on $q\in(1,\zeta^+(1))$ and they all cross $m$ at the same coordinate $(\zeta^+(1),m(\zeta^+(1)))$. In the shadow fraction of wealth coordinates, the sale boundary $q^*(\xi)$ is constant as a function of $\xi$ for sufficiently high transaction costs. See Figure \ref{fig:eg2_n} for an example where the two plots of $n$ correspond to $\xi=1.69$ and $\xi=2.1125$. On $q<1$, the solution $n$ corresponding to smaller $\xi$ dominates the solution corresponding to higher $\xi$, although the magnitude of the difference is very small numerically. The two functions pass through the same singular point $(1,m(1))$ and coincide on $q>1$ until they both cross $m(q)$ again on $q>1$ at the point indicated by the markers.

We consider the same three sets of transaction costs as in Example 1. Set 1 and Set 3 again result in identical optimal consumption rates (as functions of shadow wealth) as represented by $n$ because they share the same level of $\xi=\gamma^\uparrow\gamma^\downarrow$. Furthermore, the shadow consumption rates are the same across all the three sets on $q>1$. In Figure \ref{fig:eg2_kappa}, $\kappa(q)$ is identical on $q>1$ under Sets 1 and 2 of transaction costs  (whilst the difference is numerically very small on $q<1$). This is not surprising since $\kappa(q)$ over $q>1$ only depends on $(\gamma^\uparrow,\gamma^\downarrow)$ via $q^*(\xi)$, $(n_{\xi}(\cdot))_{q\leq q^*(\xi)}$ and $\gamma^\downarrow$,\footnote{We have $\kappa(q)=\gamma^\uparrow\exp\left(\frac{S}{1-S}\int_{q_*(\xi)}^{q}\frac{n_{\xi}'(v)}{vn_{\xi}(v)}dv\right)=\frac{1}{\gamma^\downarrow}\exp\left(-\frac{S}{1-S}\int_{q}^{q^*(\xi)}\frac{n_{\xi}'(v)}{vn_{\xi}(v)}dv\right)$.} and Set 1 and 2 share the same value of $\gamma^\downarrow$.
Further, when we move to the transaction cost parameters in Set 3 we find that on $[1,q^*]$ the corresponding $\kappa$ is a multiplicative scaling of the $\kappa$ from Sets 1 and 2 (and this multiplicative scaling extends to the interval $[q_*(\xi),q^*(\xi)]$ when we compare the functions $\kappa$ that arise in Sets 1 and 3).
These properties are inherited by the map between the real and shadow portfolio weight $p=p(q)$, where comparing the results from Sets 1 and 2, the difference in cost of purchase has no impact on the function $p(q)$ on $q> 1$, and only a numerically small difference below $1$), see Figure \ref{fig:eg2_pq}.

The fact that the Merton ratio $q_M$ is above 1 brings new phenomena, especially concerning the location of the no-transaction region and how it varies with transaction costs. Figure \ref{fig:eq2_q_buy} and \ref{fig:eq2_q_sell} confirm that the sale boundary $q^*$ in shadow fractions of wealth units is totally insensitive to transaction costs (under our baseline setup with sufficiently large $\xi$), while $q_*$ is decreasing in transaction costs. Moreover, since $p^*$ only depends on transaction costs via $q^*(\xi)$ and $\gamma^\downarrow$, $p^*$ is totally insensitive to $\gamma^{\uparrow}$ as shown in Figure \ref{fig:eq2_p_buy} (recall Remark~\ref{rem:locationofNTwedge}.
Note also that $p^*$ is decreasing in $\gamma^\downarrow$ and this monotonicity is opposite to that in Example 1. Further, Figure \ref{fig:eq2_p_buy} and \ref{fig:eq2_p_sell} show that for some parameter values $q_M\notin[p_*,p^*]$, i.e. the Merton line is not necessarily contained in the no-transaction region, again recall Remark~\ref{rem:locationofNTwedge}.

\begin{figure}[!htbp]
	\captionsetup[subfigure]{width=0.45\textwidth}
	\centering
	\subcaptionbox{Optimal consumption (in shadow wealth) $n(q)$.\label{fig:eg2_n}}{\includegraphics[scale =0.5] {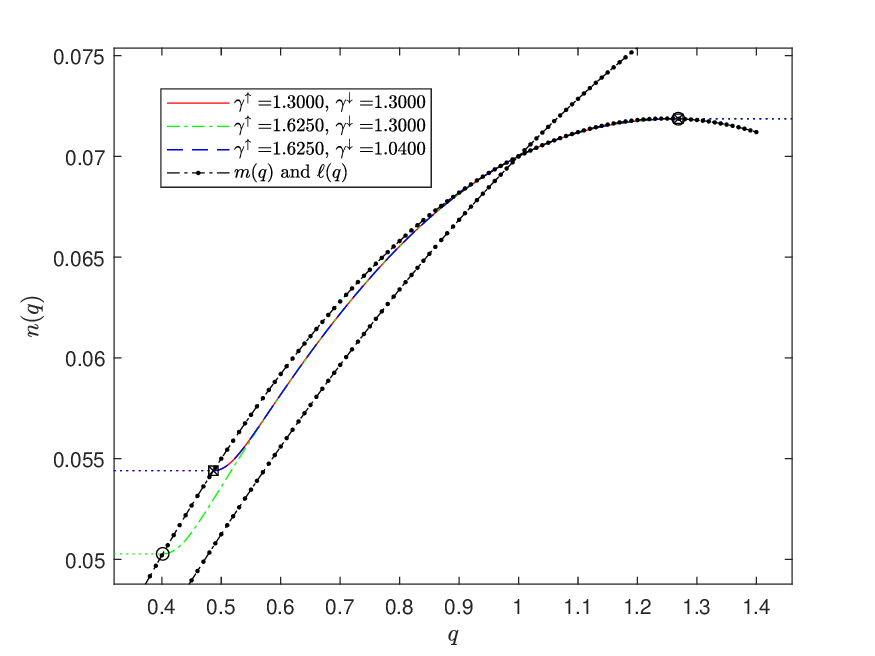}}
	\subcaptionbox{Ratio of shadow-to-real prices $\kappa(q)$.\label{fig:eg2_kappa}}{\includegraphics[scale =0.5] {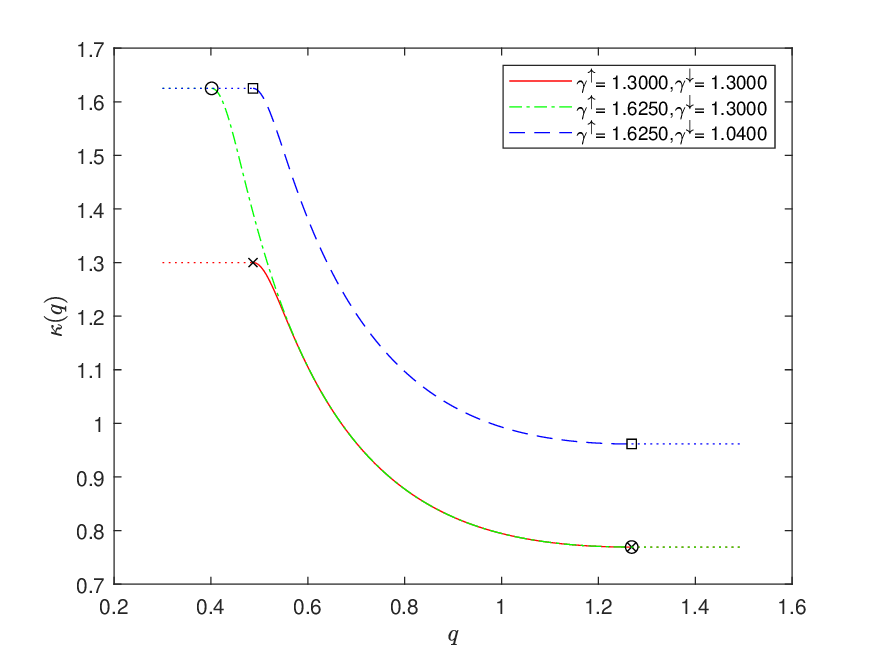}}
	\subcaptionbox{Real and shadow portfolio weight $p(q)$.\label{fig:eg2_pq}}{\includegraphics[scale =0.5] {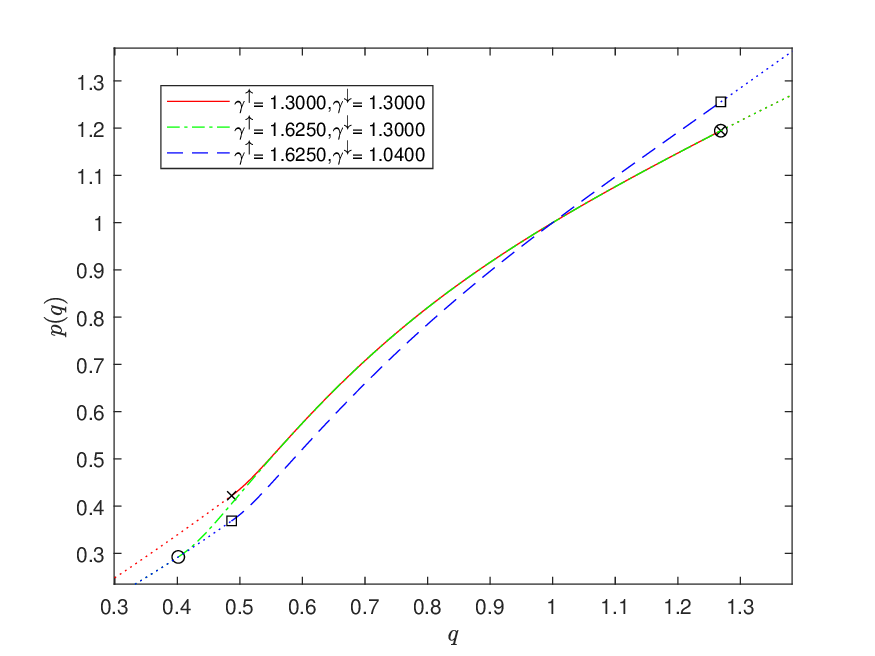}}
	
	\caption{The fundamental quantities as a function of shadow portfolio weights in Example 2. Parameters used: $R=2$, $S=4$, $\delta=0.1$, $r=0$, $\mu=0.1$, $\sigma=0.2$ (such that $\lambda=0.5$ and $\alpha=0.1$). The cross, circle and square markers indicate the boundary points of the no-transaction region (in $p_*^*$ or $q_*^*$) under the set of transaction costs $(\gamma^\uparrow, \gamma^\downarrow)$ given by $(1.3, 1.3)$, $(1.625, 1.3)$ and $(1.625, 1.04)$ respectively.}
\label{fig:eg2}
\end{figure}

\begin{figure}[!htbp]
	\captionsetup[subfigure]{width=0.45\textwidth}
	\centering
	\subcaptionbox{$q^*$ and $q_*$ as $\gamma^\uparrow$ varies.\label{fig:eq2_q_buy}}{\includegraphics[scale =0.5] {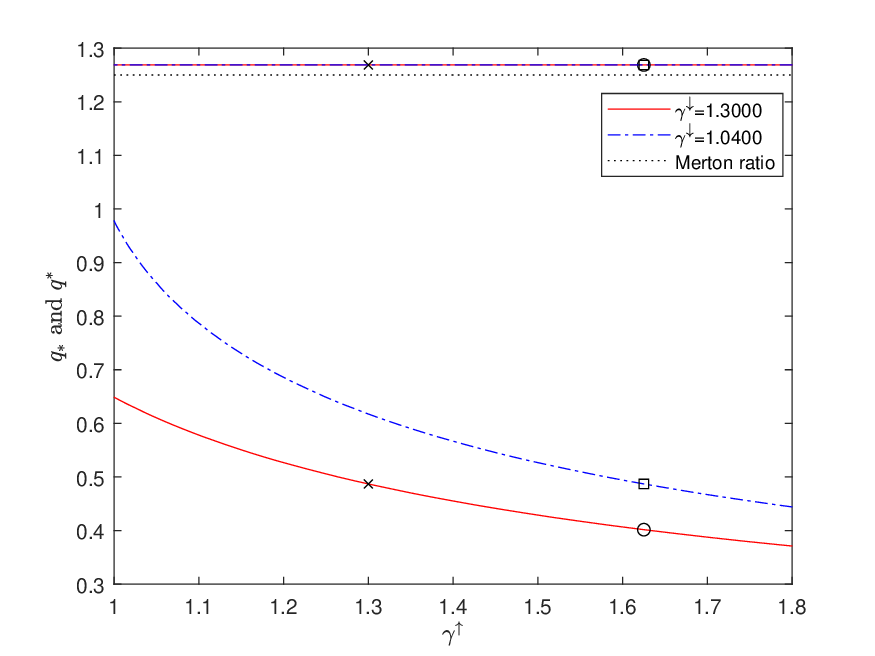}}
	\subcaptionbox{$q^*$ and $q_*$ as $\gamma^\downarrow$ varies.\label{fig:eq2_q_sell}}{\includegraphics[scale =0.5] {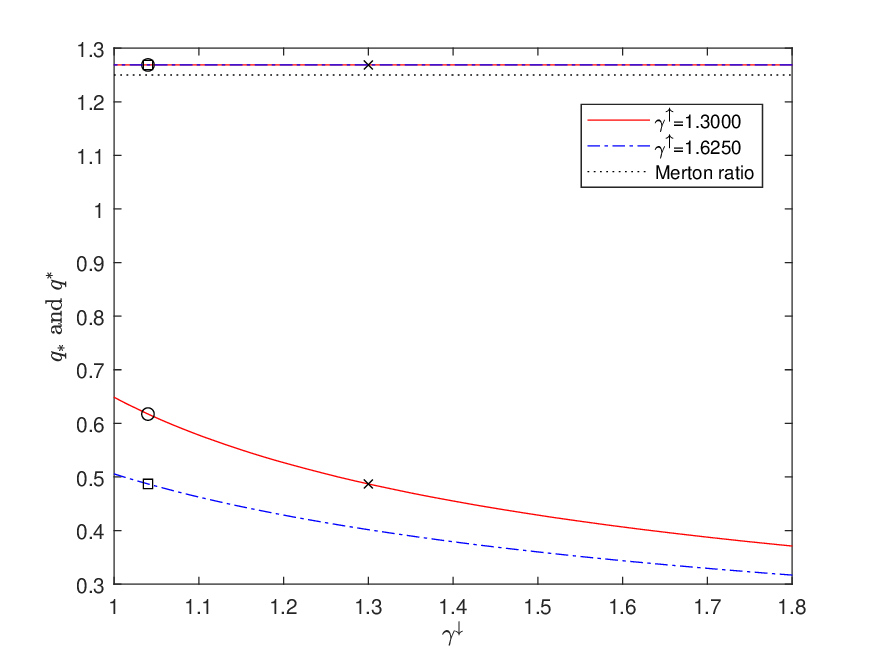}}
	\subcaptionbox{$p^*$ and $p_*$ as $\gamma^\uparrow$ varies.\label{fig:eq2_p_buy}}{\includegraphics[scale =0.5] {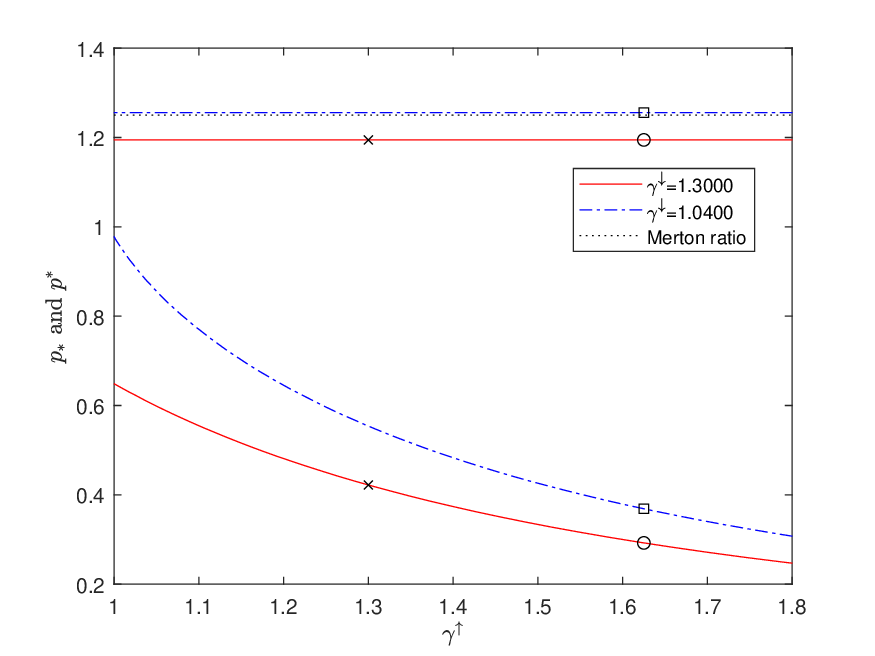}}
	\subcaptionbox{$p^*$ and $p_*$ as $\gamma^\downarrow$ varies.\label{fig:eq2_p_sell}}{\includegraphics[scale =0.5] {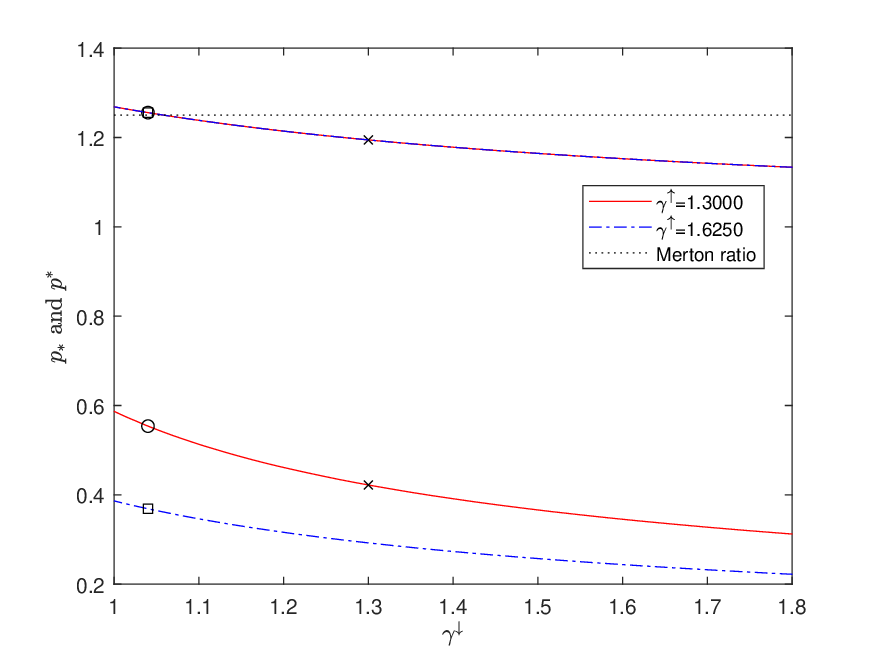}}
	
	\caption{Behaviours of the no-transaction region as a function of transaction costs in Example 2.  Parameters used:$R=2$, $S=4$, $\delta=0.1$, $r=0$, $\mu=0.1$, $\sigma=0.2$ (such that $\lambda=0.5$ and $\alpha=0.1$). The cross, circle and square markers indicate the boundary points of the no-transaction region (in $p_*^*$ or $q_*^*$) under the set of transaction costs $(\gamma^\uparrow, \gamma^\downarrow)$ given by $(1.3, 1.3)$, $(1.625, 1.3)$ and $(1.625, 1.04)$ respectively.}
\label{fig:eg2_tcost}
\end{figure}

\section{Small transaction costs}
\label{sec:smalltc}
In this section we discus the small transaction cost regime, and how the boundaries of the no-transaction region  depend on the transaction costs $\gamma^\uparrow$ and $\gamma^\downarrow$ when the transaction costs are small. In particular we look for an expansion in terms of $\epsilon^\uparrow=\gamma^\uparrow-1$ and $\epsilon^\downarrow = 1 - \frac{1}{\gamma^\downarrow}$. Let $\epsilon = \gamma^\uparrow \gamma^\downarrow - 1 = \frac{\epsilon^\uparrow + \epsilon^\downarrow}{1- \epsilon^\downarrow} = \epsilon^\uparrow + \epsilon^\downarrow + \sO(\epsilon^\downarrow(\epsilon^\uparrow + \epsilon^\downarrow))$.

Historically, the Merton investment-consumption problem with transaction costs has proved to be very challenging (even for the case of additive utility), and a natural response is to look for approximate solutions in the case of small transaction costs\footnote{It is important to note that our method gives a solution which is valid for all values of transaction costs, and then we look in the small transaction cost regime. Many other studies in the small transaction cost regime only construct solutions as expansions about the frictionless case.}. For additive utility,Rogers~\cite{Rogers:04} argued that it is natural to expect the width of the no-transaction region to be of the order of the size of the transcation cost to the power one-third. This result was formalised by Jane\v{c}ek and Shreve~\cite{JanecekShreve:04}. They showed (under an assumption that the transaction costs on buying and selling are identical, or in our setting that $\gamma^\downarrow  - 1 = 1 - \frac{1}{\gamma^\uparrow}$), and under a (almost necessary) assumption that the problem without transaction costs problem is well-posed) that $p_* = \hq - \Delta^{1/3} \epsilon^{1/3} + \sO(\epsilon^{2/3})$ and that $p^* = \hq + \Delta^{1/3} \epsilon^{1/3}+ \sO(\epsilon^{2/3})$, where $\Delta$ is defined below in \eqref{eq:defDelta}. In general, further progress was difficult because characterisations of the value function depended on finding a solution to a non-linear second order free-boundary problem, with two free boundaries. One case that is slightly simpler is the case of logarithmic utility (which formally may be considered as the case $R=1=S$ in our setting). In the case of logarithmic utility Gerhold et al~\cite{GerholdMuhleKarbeSchachermayer:18} focus on the dual problem to give an expansion for $p_*$, $p^*$ to order $\sO(\epsilon^{2/3})$ and an expansion for the optimal consumption which is valid up to order $\sO(\epsilon)$.

Still in the additive case, Choi et al~\cite{ChoiSirbuZitkovic:13} and Hobson et al~\cite{hobson:tse:zhu:19A}, showed how the problem for general transaction costs can be reduced to a first order equation. This facilitated a simpler derivation of the expansion for small transaction costs, see Choi~\cite{Choi:14} and Hobson et al~\cite{hobson:tse:zhu:19A}, and these papers calculated the second order term. Choi~\cite{Choi:14} assumed zero transaction costs on selling (i.e. that $\gamma^\downarrow = 1$), but his results can be translated to the more general case. Hobson et al~\cite{hobson:tse:zhu:19A} additionally considered the expansion in the case $\hq=1$, in which case the width of the no transaction wedge is of order $\epsilon^{1/2}$. (We could extend our results to this case similarly).

Melynyk et al~\cite{MMKS:20} considered the small transaction cost case for Epstein-Zin SDU. Under some restrictions on the paramaters (which go beyond well-posedness of the problem) and under the assumption that the transaction costs on buying and selling are identical, (i.e. $\gamma^\downarrow  - 1 = 1 - \frac{1}{\gamma^\uparrow}$) they derived an expansion to order $\epsilon^{2/3}$ for the value function and thence find expansions (again accurate to order $\epsilon^{2/3}$) for the boundaries of the optimal no-transaction region and the optimal consumption.

Define
\begin{align}
\Delta & =  \frac{3 \hq^2 (1 - \hq)^2}{4R}; \label{eq:defDelta} \\
\Sigma & =  \frac{4 \hatm}{3 \sigma^2 \hq(1 - \hq)^2} = \frac{\hatm \hq}{R \sigma^2 \Delta};
\label{eq:defSigma}  \\
\Psi 
&= \frac{1}{5} \Sigma^2 + \frac{1}{5} \frac{(5\hq-3)}{\hq(1-\hq)} \Sigma - \frac{(3 -10\hq + 10 \hq^2)+ 4R\hq(1-\hq)}{15\hq^2(1-\hq)^2}.
\label{eq:defPsi}
\end{align}

\begin{prop}
\label{prop:smalltcQ}
Suppose that $R,S<1$ and $\hatm \geq 0$ or $R,S>1$ and $\hatm>0$, and suppose that $q_M \notin \{0, 1 \}$. Then for all $\epsilon > 0$ sufficiently small
\begin{eqnarray*}
q_*(1+\epsilon) & = & \hq - \Delta^{\frac{1}{3}} \epsilon^{\frac{1}{3}} - \Sigma \Delta^{\frac{2}{3}} \epsilon^{\frac{2}{3}} - \Psi \Delta \epsilon + \sO\left( \epsilon^{\frac{4}{3}} \right) \\
q^*(1+\epsilon) & = & \hq  + \Delta^{\frac{1}{3}} \epsilon^{\frac{1}{3}} - \Sigma \Delta^{\frac{2}{3}} \epsilon^{\frac{2}{3}} + \Psi \Delta \epsilon + \sO\left( \epsilon^{\frac{4}{3}} \right)
\end{eqnarray*}
\end{prop}

\begin{rem}
\begin{enumerate}
\item The assumption of Proposition \ref{prop:smalltcQ} that either $R,S<1$ and $\hatm \geq 0$ or $R,S>1$ and $\hatm>0$ is equivalent to an assumption that the problem \eqref{eq:frictional problem} is well-posed for all sufficiently small transaction cots.
\item In the shadow-fraction of wealth co-ordinates, the width of no-transaction region is $\sO(\epsilon^{1/3})$ and is symmetric about the Merton proportion $\hq$. To order $\epsilon^{1/3}$ the width of the no-transaction region depends on risk aversion $R$ (both directly, and indirectly through $\hq = \frac{\lambda}{\sigma R}$), but not on the elasticity of inter-temporal complementarity $S$.
\item The $\sO(\epsilon^{2/3})$ term is same for both $q_*$ and $q^*$ and moves the no-transaction region in the direction of smaller positions in the risky asset. This term does not depend directly on $S$, although it does depend indirectly on $S$ via $\hat{m}$.
\item The $\sO(\epsilon^{2/3})$ term is same for both $q_*$ and $q^*$, except that the signs are opposite, so that rather than moving the no-transaction region it makes it wider or narrower. This term does not depend directly on $S$, although it does depend indirectly on $S$ via $\hat{m}$.
\item It is the straightforward to extend the arguments to higher order as required, or to the case $\hq=1$.  If $\lambda=0=\hq$ then the problem is degenerate and optimal strategy for the agent is to instantly eliminate any position (long or short) in the risky asset, and thereafter to keep all wealth in the bank account.
\end{enumerate}
\end{rem}

\begin{cor}
\label{cor:smalltcP}
Suppose that $R,S<1$ and $\hatm \geq 0$ or $R,S>1$ and $\hatm>0$, and suppose that $q_m \notin \{0,1 \}$. Then for all $\epsilon^\uparrow, \epsilon^\downarrow > 0$ sufficiently small
\begin{eqnarray*}  p_*(1 + \epsilon^\uparrow,1 + \epsilon^\downarrow) & = & \hq - \Delta^{\frac{1}{3}} (\epsilon^\uparrow+\epsilon^\downarrow)^{\frac{1}{3}} - \Sigma \Delta^{\frac{2}{3}} (\epsilon^\uparrow+\epsilon^\downarrow)^{\frac{2}{3}} - \Psi \Delta (\epsilon^\uparrow+\epsilon^\downarrow) - \epsilon^\uparrow \hq(1-\hq) + \sO(\epsilon^{\frac{4}{3}}) \\
p^*(1 + \epsilon^\uparrow,1 + \epsilon^\downarrow)  & = & \hq + \Delta^{\frac{1}{3}} (\epsilon^\uparrow+\epsilon^\downarrow)^{\frac{1}{3}} - \Sigma \Delta^{\frac{2}{3}} (\epsilon^\uparrow+\epsilon^\downarrow)^{\frac{2}{3}} + \Psi \Delta (\epsilon^\uparrow+\epsilon^\downarrow) + \epsilon^\downarrow \hq(1-\hq) + \sO(\epsilon^{\frac{4}{3}})
\end{eqnarray*}
\end{cor}

\begin{rem}
\begin{enumerate}
\item In the original coordinates, the width of the no-trade region is $\sO(\epsilon^{1/3})$ and to leading order is symmetric about the Merton proportion $\hq$. Also to leading order, the width of the no-transaction region depends on risk aversion $R$ (both directly, and indirectly through $\hq = \frac{\lambda}{\sigma R}$), but not on the elasticity of inter-temporal complementarity $S$. See also Jane\v{c}ek and Shreve~\cite[Remark 1]{JanecekShreve:04} in the additive case, and Melnyk et al~\cite{MMKS:20} for the corresponding result under EZ-SDU.
\item The $\sO(\epsilon^{2/3})$ term is same for both $p_*$ and $p^*$ and moves the no-transaction region in the direction of smaller positions in the risky asset. This term does not depend directly on $S$, although it does depend indirectly on $S$ via $\hat{m}$. See also \cite{JanecekShreve:04,MMKS:20}.
\item The individual transaction costs $\epsilon^\uparrow$ and $\epsilon^\downarrow$ only enter at linear order.
The linear order term does not depend on the elasticity of inter-temporal complementarity $S$, except
indirectly through $\hatm$. This is the final term with this property.
\end{enumerate}
\end{rem}


\begin{prop}
\label{prop:smalltcC}
Suppose that $R,S<1$ and $\hatm \geq 0$ or $R,S>1$ and $\hatm>0$, and suppose that $q_M \notin \{0,1 \}$.
Let $h_3(v) = (1+v)^2(2-v)$.  Then for all $\epsilon^\uparrow, \epsilon^\downarrow > 0$ sufficiently small
\begin{align}
C^*_{1+\epsilon^\uparrow,1+\epsilon^\downarrow}(t,x,y,\phi)
& = C^*_{1,1}(t,x,y,\phi) \bigg( 1 +  \frac{R(1-S)}{S} \frac{\sigma^2}{2 \hatm} \Delta^{2/3} (\epsilon^\uparrow+ \epsilon^\downarrow)^{2/3}  \\
& \hspace{18mm} +
\frac{\hq}{2S} \left( 1  - \frac{1}{2} h_3\left(\frac{p-\hq}{\Delta^{1/3} \epsilon^{1/3}}\right) \right) (\epsilon^\uparrow+ \epsilon^\downarrow)+ \frac{\hq}{2} (\epsilon^\uparrow - \epsilon^\downarrow) \bigg) + \sO(\epsilon^{4/3})
\end{align}
where $p = \frac{y \phi}{x + y \phi}$ and $C^*_{1,1}(t,x,y,\phi) :=(x+y \phi) \hatm$ is the optimal consumption under zero transaction costs.
\end{prop}

\begin{rem}
\label{rem:consumption}
\begin{enumerate}
\item To zeroth order consumption is equal to optimal consumption without transaction costs, and there is no term of order $\epsilon^{1/3}$. (See also Jane\v{c}ek and Shreve~\cite[Remark 2]{JanecekShreve:04} in the additive case and Melnyk et al~\cite[Corollary 6.1]{MMKS:20} in the case of EZ-SDU.)
\item The first non-zero correction term is independent of $p = \frac{y \phi}{x+y \phi}$ and is of order $\epsilon^{2/3}$, (\cite{JanecekShreve:04} and \cite{MMKS:20} make a similar observation.) It does not depend on the individual transaction costs. Mathematically, the fact that the consumption is constant arises from the fact that solutions $n_{(\hq - u)}^+$ are (approximately) horizontal for small $u$, (since $n_{(\hq - u)}^+$ has zero derivative at both endpoints where it intersects with $m$). Here $n_{(\hq - u)}^+$ is the solution to $n=O(n,q)$ started at $(\hq-u, n(\hq-u) = m(\hq-u))$, and is defined fully in Appendix~\ref{app:ode}.
\item The sign of the first non-zero correction term depends on the sign of $(1-S)$. If $S<1$ then the agent consumes more than in the zero-transaction case; if $S>1$ then the agent consume less\footnote{Jane\v{c}ek and Shreve~\cite[Remark 2]{JanecekShreve:04} and Melnyk et al~\cite{MMKS:20} also make this observation.  Jane\v{c}ek and Shreve write {\em `\ldots the existence of transaction costs increases the size of consumption for $R\in (0,1)$, while the consumption is decreased for $R>1$. This is explained by the fact that the index of inter-temporal substitution, $1/R$, is high for small $R$'} -- note that in their additive model $R \equiv S$. However, whilst the two facts (existence of transaction costs increases (respectively decreases) consumption for $R \in (0,1)$ (respectively $R \in (1,\infty)$), and the index of inter-temporal substitution $1/R$ is high for small $R$) are correct, no evidence is offered as to why the second fact is an explanation of the first, or why the sign of $R-1$ is critical. In the case of zero-transaction costs (and therefore also, approximately in the case of small transaction costs) the optimal consumption is $m_M = \left( r + \frac{\lambda^2}{2R} \right) + \frac{1}{S} \left( \delta - r - \frac{\lambda^2}{2R} \right)$. Assuming that $\delta - r - \frac{\lambda^2}{2R}>0$, the primary effect of increasing $S$ is to reduce the optimal consumption, but this does {\em not} give an explanation of why the sign of consumption changes relative to the zero-transaction cost case depend on the sign of $(1-S)$. For this a more subtle argument is needed. Melnyk et al~\cite[Page 1147]{MMKS:20} point out that in the case of stochastic differential utility the elasticity of intertemporal complementarity is captured by $S$ (and is separate to the risk aversion $R$) but defer to Jane\v{c}ek and Shreve for an explanation. We return to this issue in Section~\ref{ssec:WhyCdependsonS}.}.
    At its heart, the explanation for this relationship is that given in Remark~\ref{rem:shapeofm} about the impact of investing a sub-optimal fraction of wealth in the risky asset, and the effect this has on the certainty equivalent wealth of the agent.
    The impact of transaction costs is to cause the agent to invest a sub-optimal fraction of wealth in the risky asset. If $S<1$ the fact that the certainty equivalent value of the agent's holdings goes down (relative to the zero transaction cost case) leads the agent to increase their instantaneous consumption. When $S>1$, again the presence of transaction costs causes the agent to invest a sub-optimal fraction of wealth in the risky asset. This time this leads the agent to reduce instantaneous consumption. See Section~\ref{ssec:WhyCdependsonS} for further development of this argument.
\item We expect that when $R=S=1$ we recover the expansion for logarithmic utility, as studied by Gerhold et al~\cite{GerholdMuhleKarbeSchachermayer:18}. Indeed, in this case Gerhold et al argue that the leading order correction to the optimal consumption is of order $\epsilon$ (and not of order $\epsilon^{2/3}$). Our results are consistent with this claim.
\item Individual transaction costs enter at $\sO(\epsilon)$. This is also the first term at which the fraction $p$ of wealth invested in the risky asset affects the consumption.
\end{enumerate}
\end{rem}

\section{Comparative statics}
\label{sec:comparative}

\subsection{Comparative statics of the boundaries in $S$ }
\label{sect:comparative_S}

\begin{prop}
	Fix all the model parameters except $S$. Suppose $\hat{S}> \tilde{S}$ and the problems under $S=\hat{S}$ and $S=\tilde{S}$ are well-posed where the optimal purchase and sale boundaries are  $(\hat{p}_*,\hat{p}^*)$ and $(\tilde{p}_*,\tilde{p}^*)$ respectively.
	
	\begin{enumerate}
		\item If
		\begin{align}
			\delta-r-\frac{\lambda^2}{2R}\geq 0,
			\label{eq:suffcon}
		\end{align}
		then $\hat{p}_*\geq \tilde{p}_*$ and  $\hat{p}^*\geq \tilde{p}^*$.  
		
		\item Suppose instead
		\begin{align}
			\delta-r-\frac{\lambda^2}{2R}< 0.
			\label{eq:othercon}
		\end{align}
		If the problems under $S=\hat{S}$ and $S=\tilde{S}$ are well-posed for all sufficiently small transaction costs, then there exists $\epsilon>0$ (which may depend on $\hat{S},\tilde{S}$) such that $\hat{p}_*\leq \tilde{p}_*$ and  $\hat{p}^*\leq \tilde{p}^*$ for all $\xi=\gamma^{\uparrow}\gamma^{\downarrow}\in[1,1+\epsilon]$. 
	\end{enumerate}
	\label{prop:compstat_S}
\end{prop}

The proof of Proposition~\ref{prop:compstat_S} is given in Appendix~\ref{sapp:compstatS}.

Note that although we have made no assumptions about the sign of $r$ it is natural to assume that $r>0$ and then $r + \frac{\lambda^2}{2R}>0$. Then it is easily seen that the condition for the problem with (sufficiently) small transaction costs to be well-posed, namely $S\left(r + \frac{\lambda^2}{2R}\right) > r + \frac{\lambda^2}{2R} - \delta$, is equivalent to $S > 1 - \frac{\delta}{r + \frac{\lambda^2}{2R}}$.

\begin{cor}
Fix all the model parameters except $S$.

Suppose that $\delta \geq r + \frac{\lambda^2}{2R} >0$. Suppose $S<1$. Then the problem is always well-posed and $p_*$ and $p^*$ are increasing in $S$ for all levels of transaction costs. Conversely, suppose $S>1$. Then $p_*$ and $p^*$ are increasing in $S$ over the range of transaction costs for which the problem is well-posed. In particular $p_*$ and $p^*$ are increasing in $S$ for small transaction costs.

Suppose that $r + \frac{\lambda^2}{2R} > \max \{ \delta , 0 \}$. Then the problem is well-posed for small transaction costs for $S>  \max \{ 1 - \frac{\delta}{r + \frac{\lambda^2}{2R}}, 0 \}$ and then, still for small transaction costs, $p_*$ and $p^*$ are decreasing in $S$ over this range.
\end{cor}

Each of Figure \ref{fig:compstat_Sleq1} and \ref{fig:compstat_Sgeq1} shows the two cases described in Proposition \ref{prop:compstat_S} where the monotonicity behaviours of $p_*,p^*$ can be different. To understand why different cases can arise, recall that in the frictionless case the optimal consumption rate is given by
\begin{align*}
	m_M=m(q_M)=r+\frac{\lambda^2}{2R}+\frac{\delta-r- \frac{\lambda^2}{2R}}{S}.
\end{align*}
The sign of $\delta-r- \frac{\lambda^2}{2R}$ therefore governs whether, in the absence of transaction costs, the (constant) optimal consumption rate increases or decreases in $S$.	Suppose $\delta-r-\frac{\lambda^2}{2R}$ is positive. Then the agent with lower $S$ wants to consume at a higher rate whilst maintaining a constant fraction of wealth invested in the risky asset given by the Merton ratio. As trading does not incur transaction costs, the higher level of consumption can be supported interchangeably by sale of the risky asset and/or withdrawal of cash from the risk-free account.

With transaction costs, we expect that the sign of $\delta-r-\frac{\lambda^2}{2R}$ affects the monotonicity of the optimal consumption rate with $S$ in a similar fashion (at least when transaction costs are small). If $\delta-r-\frac{\lambda^2}{2R}$ is positive, the agent with lower $S$ again wants to consume more, but now such an agent will prefer to finance consumption from cash wealth rather than from sales of the risky asset because of the market frictions. Hence the agent with lower $S$ in general has a stronger incentive to hold cash in anticipation of the need to consume more, resulting in a downward shift of the no-transaction wedge $[p_*,p^*]$. The opposite will happen when $\delta-r-\frac{\lambda^2}{2R}$ is negative.

\begin{figure}[!htbp]
	\captionsetup[subfigure]{width=0.45\textwidth}
	\centering
	\subcaptionbox{$\delta=0.1$ such that \eqref{eq:suffcon} holds.\label{fig:Sleq1_increase}}{\includegraphics[scale =0.5] {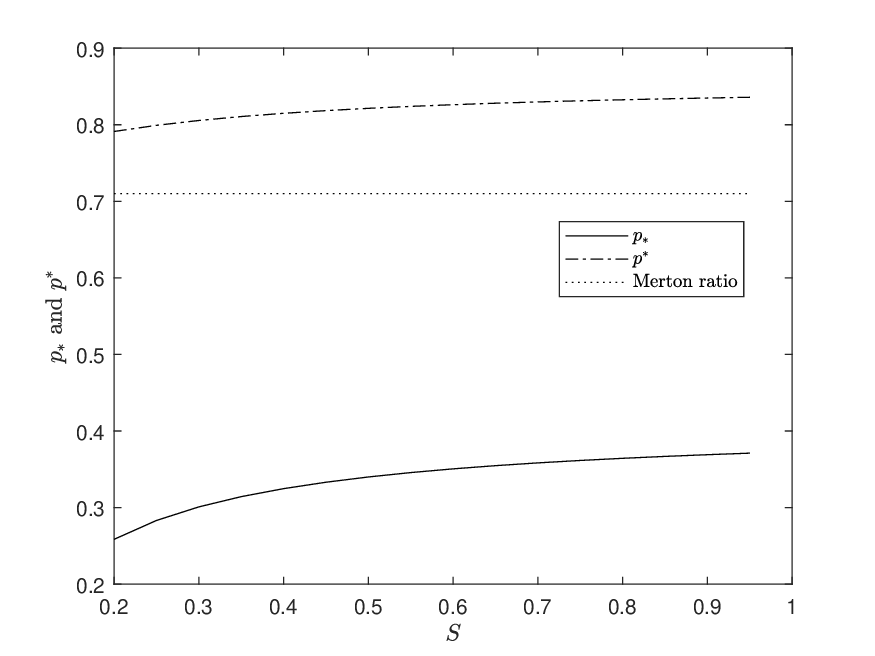}}
	\subcaptionbox{$\delta=0.06$ such that \eqref{eq:othercon} holds.\label{fig:Sleq1_decrease}}{\includegraphics[scale =0.5] {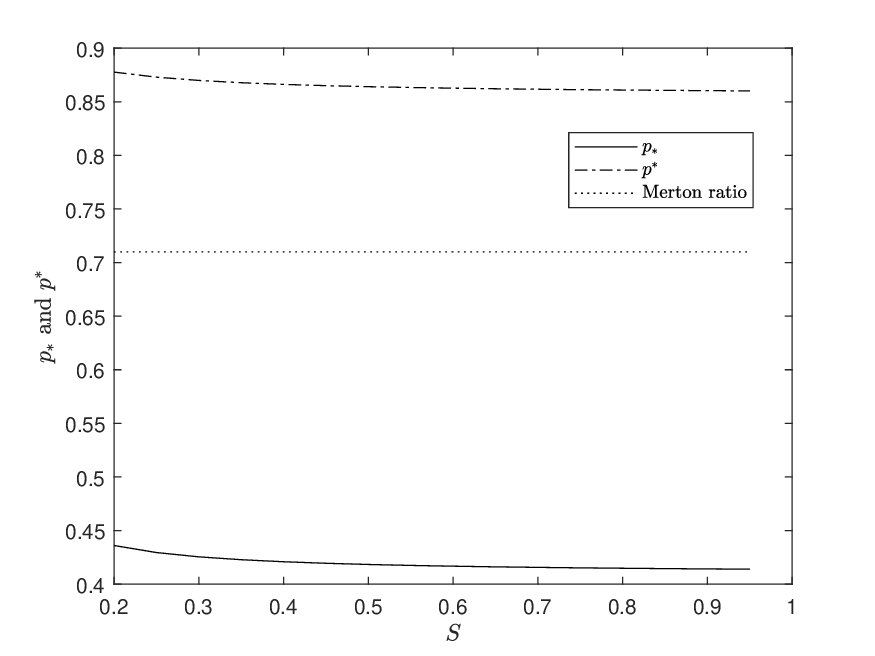}}
		
	\caption{$p_*$ and $p^*$ as $S$ varies when $R,S<1$. Base parameters used: $\gamma^\uparrow=\gamma^\downarrow=1.1$, $R=2/3$, $r=0$, $\mu=0.2$, $\sigma=0.6$.}
\label{fig:compstat_Sleq1}
\end{figure}

\begin{figure}[!htbp]
\captionsetup[subfigure]{width=0.45\textwidth}
\centering
\subcaptionbox{$\delta=0.1$ such that \eqref{eq:suffcon} holds.\label{fig:Sgeq1_increase}}{\includegraphics[scale =0.5] {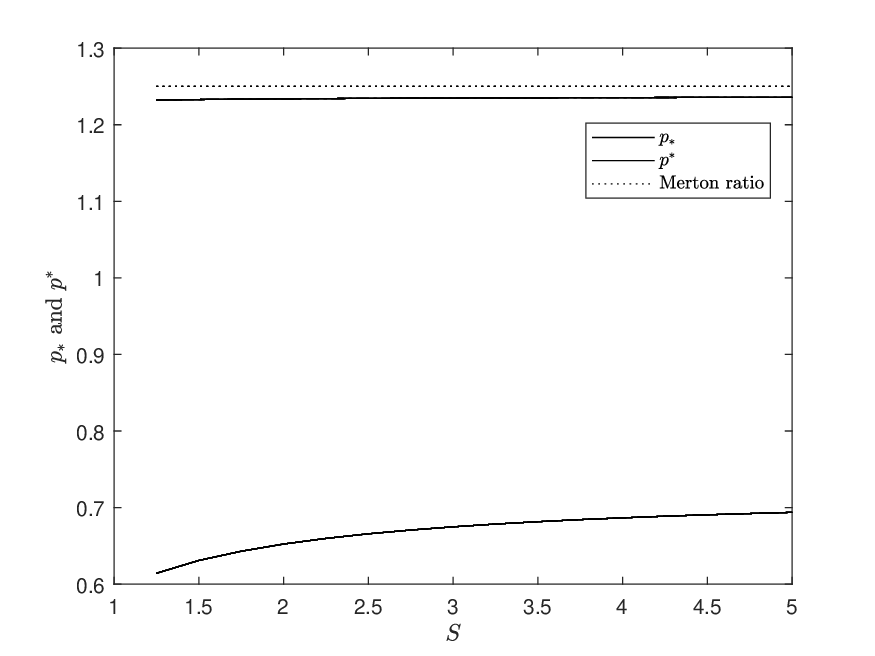}}
\subcaptionbox{$\delta=0.05$ such that \eqref{eq:othercon} holds.\label{fig:Sgeq1_decrease}}{\includegraphics[scale =0.5] {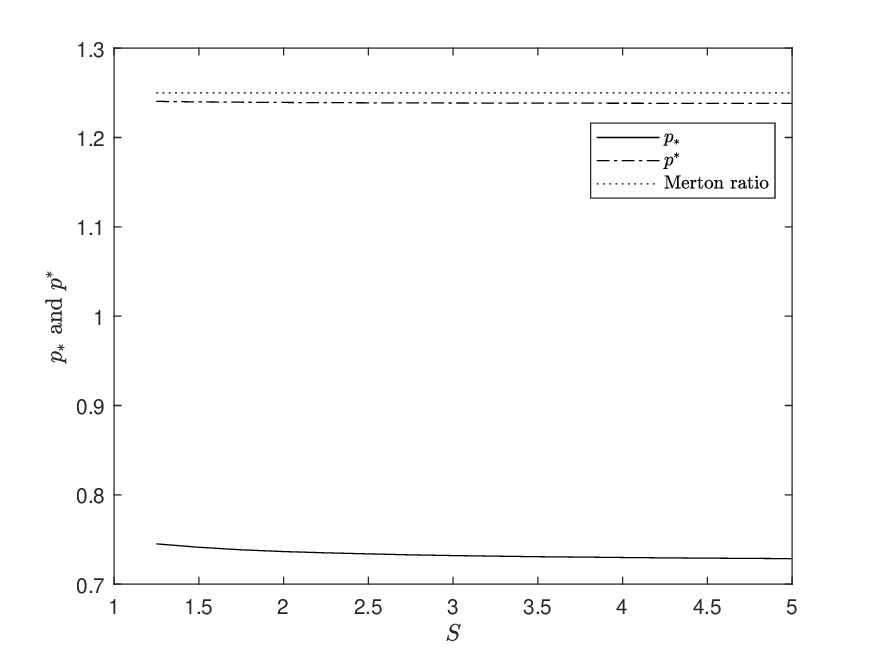}}

\caption{$p_*$ and $p^*$ as $S$ varies when $R,S>1$. Base parameters used: $\gamma^\uparrow=\gamma^\downarrow=1.1$, $R=2$, $r=0$, $\mu=0.1$, $\sigma=0.2$.}
\label{fig:compstat_Sgeq1}
\end{figure}

\subsection{Comparative statics of the boundaries in $R$ }
The proof of Proposition~\ref{prop:CompstatR} is given in Appendix~\ref{sapp:compstatR}.

\begin{prop}
Suppose $\lambda>0$. Fix all the model parameters except $R$. Let $\hat{R}>\tilde{R}$ be such that the problems under $R=\hat{R}$ and $R=\tilde{R}$ are well-posed where the optimal purchase and sale boundaries are  $(\hat{p}_*,\hat{p}^*)$ and $(\tilde{p}_*,\tilde{p}^*)$ respectively. If any of the conditions in Lemma \ref{lem:cq} holds for $R\in\{\hat{R},\tilde{R}\}$, then $\hat{p}_*\leq \tilde{p}_*$ and  $\hat{p}^*\leq \tilde{p}^*$. Consequently, subject to the well-posedness of the problem:
\begin{enumerate}
	\item If $S<1$ then $p_*$ and $p^*$ are decreasing in $R$ for any $R\in(0,1)$.
	\item If $S>1$ and $m(0)>0$ then $p_*$ and $p^*$ are decreasing in $R$ for $R\in[2,\infty)$. If also $\lambda \geq \frac{\sigma}{2}$ then $p_*$ and $p^*$ are decreasing in $R$ for $R\in(1,\infty)$.
\end{enumerate}
\label{prop:CompstatR}
\end{prop}

\begin{figure}[!htbp]
	\captionsetup[subfigure]{width=0.45\textwidth}
	\centering	
	\subcaptionbox{$R,S<1$.\label{fig:Rleq1_decrease}}{\includegraphics[scale =0.5] {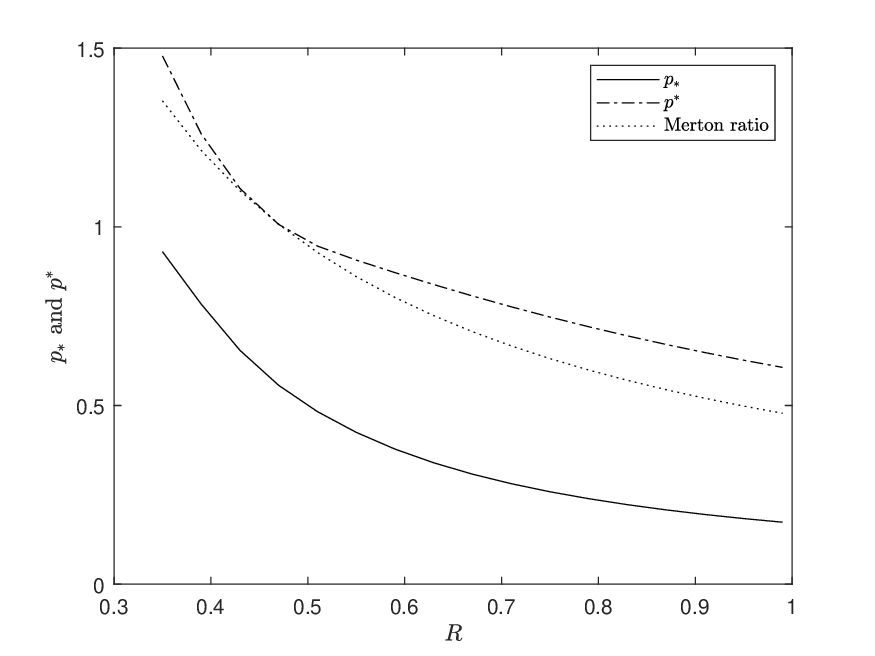}}
	\subcaptionbox{$R,S>1$\label{fig:Rgeq1_decrease}}{\includegraphics[scale =0.5] {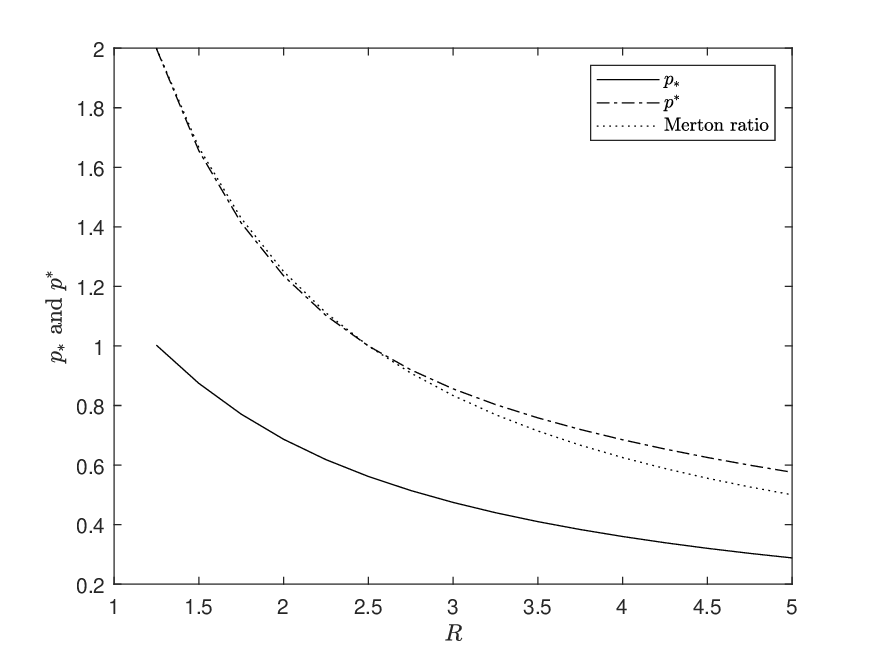}}
	\caption{$p_*$ and $p^*$ as $R$ varies. Base parameters used in Panel (a): $\gamma^\uparrow=\gamma^\downarrow=1.1$, $S=1/3$, $\delta=0.1$, $r=0$, $\mu=0.2$, $\sigma=0.6$. Base parameters used in Panel (b): $\gamma^\uparrow=\gamma^\downarrow=1.1$, $S=4$, $r=0$, $\delta=0.1$, $\mu=0.1$, $\sigma=0.2$. All parameters (and the values of $R$ considered) are such that at least one of the sufficient conditions in Lemma \ref{lem:cq} holds.}
\label{fig:compstat_Rleq1}
\end{figure}

\begin{figure}[!htbp]
\captionsetup[subfigure]{width=0.45\textwidth}
\centering
\includegraphics[scale =0.5] {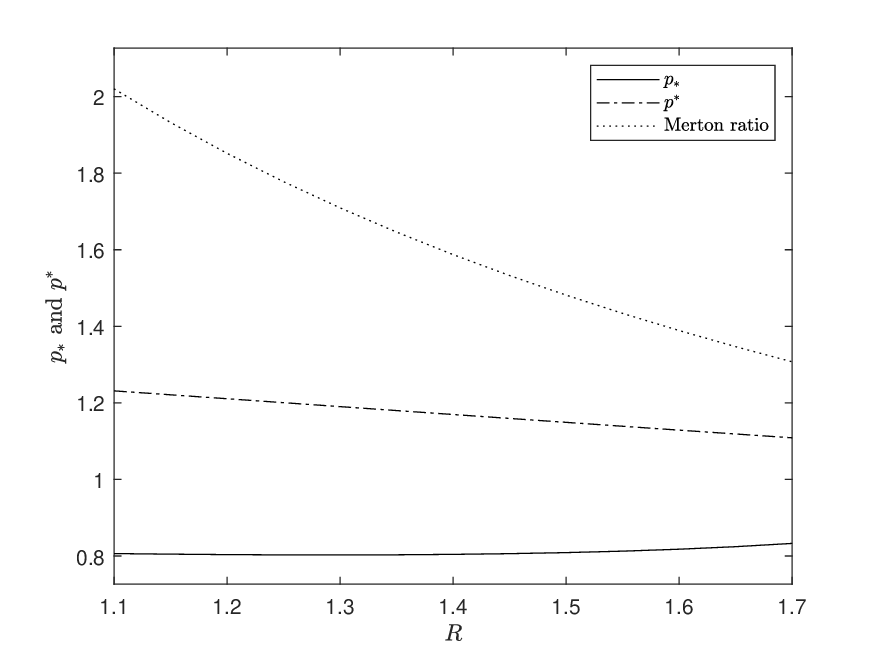}

\caption{$p_*$ and $p^*$ as $R$ varies when all the stated conditions in Lemma \ref{lem:cq} do not hold. Base parameters used: $\gamma^\uparrow=1.7$, $\gamma^\downarrow=3.33$, $S=2$, $r=0$, $\delta=-0.11$, $\mu=0.2$, $\sigma=0.3$. In this example, $p_*$ is increasing in $R$.}
\label{fig:compstat_R_wo_suffcond}
\end{figure}

The main conclusion from Proposition~\ref{prop:CompstatR} is that in the typical case $p_*$ and $p^*$ are decreasing in $R$. This result is proven to hold over the regime where the problem is well-posed whenever $R<1$ and also when $R>1$ unless either $m(0)\leq 0$ or $R \in (1,2)$ and $\lambda < \frac{\sigma}{2}$. The condition $m(0) > 0$ is equivalent to $\delta + r(S-1)<0$, which is automatically satisfied if $R,S>1$ and  $\delta,r>0$. Such monotonicity behaviour of $p_*$ and $p^*$ in $R$ is perhaps not too surprising in view of the fact that the Merton ratio $q_M=\lambda/(R\sigma)$ is decreasing in $R$. Nonetheless, it is useful to stress that the consumption motive will also affect $p_*$ and $p^*$ as $R$ varies. Note that the optimal consumption rate in the frictionless case can be rewritten as
\begin{align}
	m_M=m(q_M)=r+\frac{\delta-r}{S}+\frac{\lambda^2}{2R}\frac{S-1}{S}.
\label{eq:opt_frictionless_c}
\end{align}
Then for a fixed $S<1$, the (frictionless) consumption rate is increasing in $R$. Following similar intuition to that used in Section \ref{sect:comparative_S}, in the case with frictions we expect that the elevated consumption due to an increase in $R$ will make the agent more inclined to hold cash (i.e. a smaller fraction of wealth is invested in the risky asset). This effect is in the same direction as the impact of $R$ on the Merton ratio (when $S<1$).

However, if $S>1$ then \eqref{eq:opt_frictionless_c} suggests the optimal frictionless consumption rate is decreasing in $R$. The (lack of) consumption motive will then encourage the agent to hold less cash and hence the fraction of wealth invested in the risky asset becomes larger. The effect of $R$ on the investment level through the consumption motive is therefore opposite in sign to the direct effect of $R$ on the Merton ratio, and the net effect on the investment level becomes ambiguous. This is the reason why in the case $S>1$ we require further conditions in Proposition \ref{prop:CompstatR} to be able to conclude that $p_*$ and $p^*$ are decreasing in $R$. Indeed, there exist counterexamples such that the boundaries are not decreasing in $R$, although usually they require somewhat extreme model parameters such as a negative discount rate in conjunction with very high transaction costs. See Figure \ref{fig:compstat_R_wo_suffcond}.

\subsection{What determines whether consumption increases or decreases with the introduction of transaction costs?}
\label{ssec:WhyCdependsonS}

In this section we give three related arguments which together justify the comments made in Remarks~\ref{rem:shapeofm} and \ref{rem:consumption}(c).

Consider the continuous-time investment-consumption problem for EZ-SDU in the frictionless case.
Suppose the agent uses constant proportional strategies so that the value function is as given in \eqref{eq:valfungenstrat} where $H$ is as given in \eqref{eq:H_nu}, i.e.
\[ H(q,m) = \delta + (S-1)\left( r - m + \lambda \sigma q - \frac{1}{2} q^2 \sigma^2 R \right) . \]
Then, provided $H(q,m)>0$, and with $z = x+ y \theta$ denoting initial wealth, $V_0 = \frac{z^{1-R}}{1-R} \left( \frac{ m^{1-S}}{H(q,m)} \right)^\theta$. Maximising $V_0$ is equivalent to maximising $((1-R)V_0)^{\frac{1}{1-R}} = z  \frac{m}{H(q,m)^{\frac{1}{1-S}}}$. Note that this quantity is expressed in economically meaningful units as it represents the certainty equivalent value.

As in Section~\ref{ssec:frictionless}, suppose that the fraction of wealth invested in the risky asset is constrained to be equal to $q$, where $q$ is fixed and given, and possibly suboptimal, and consider choosing the optimal consumption rate.
The maximum certainty equivalent can be found by differentiation and we get that the optimal consumption rate $m$ is given by $m=m(q)$ and solves $\frac{1}{m} = \frac{H_m(q,m)}{(1-S)H(q,m)}$. Since $H_m(q,m)= (1-S)$ this simplifies to $m(q) = H(q,m(q))$, so that $m(q)$ is the root of the decreasing (in $m$) function $h_q$ where $h_q(m) = H(q,m)-m$.


Suppose $S<1$. Then $H(q,m)$ has a minimum in $q$ at $q_M$ (independent of the value of $m$). In particular, for $q \neq q_M$, $0 = h_q(m(q)) =  H(q,m(q)) - m(q) > H(q_M,m(q))-m(q) = h_{q_M}(m(q))$ (where the strict inequality follows from the strict convexity of $H$ in $q$). Then, since $h_{q_M}$ is decreasing we must have $m(q_M) < m(q)$. In particular, $m(\cdot)$ has a minimum at $q_M$.
Conversely, suppose $S>1$ so that $H(q,m)$ has a maximum in $q$ at $q_M$. Then, $0 = h_q(m(q)) =  H(q,m(q)) - m(q) \geq H(q_M,m(q))-m(q) = h_{q_M}(m(q))$ and since $h_{q_M}$ is decreasing we must have $m(q_M) \geq m(q)$. In particular, $m(\cdot)$ has a maximum at $q_M$ (and the uniqueness of the minimiser follows from the uniqueness of the maximiser of $H$ in $q$ as before).


It follows that (in the zero-transaction cost case) when $S<1$ the impact of using a non-optimal investment strategy, and then optimising over consumption is to increase the immediate consumption rate relative to the optimal consumption rate under an optimal investment strategy.
When $S>1$, the impact of using a non-optimal investment strategy, and then optimising over consumption is to decrease the immediate consumption rate.

One impact of transaction costs is to force the agent to use a suboptimal investment strategy when compared with the frictionless world. The above analysis then explains whether the impact of transaction costs on the consumption rate is to increase or decrease instantaneous consumption (relative to the frictionless case) and why this depends exactly on the shape of $H(q,m(q))$ (does it have a minimum or a maximum at $q_M$) which in turn is directly related to the sign of $(S-1)$.

The above argument might be considered a mathematical justification of the role of $(S-1)$ in determining the impact of the introduction of transaction costs on consumption, but it does not directly relate to the economic explanation of the fundamental role of $S-1$ given in Remark~\ref{rem:shapeofm} in terms of the interplay between two impacts of reducing the certainty equivalent value of future consumption. To understand what is happening at a more fundamental level, we consider recursive utility over one time period.

Consider an agent in a one-period economy (with times $0$ and $1$) containing a risk-less asset and a risky asset. The agent has initial wealth $x$ and recursive preferences. The time-0 controls available to the agent are to decide the amount $c \in (0,1)$ to consume at $t=0$, and the number of units $\phi$ of the risky asset to hold between times $0$ and $1$. We assume that the price of the risk-less asset is normalised so that it is constant over time, that the initial price of the risky asset is 1 and that the return on the risky asset is given by the time-1 measurable random variable $Z$ (so that 1 invested in the risky asset at time 0 yields $(1+Z)$ at time-1). It follows that the wealth of the agent at time 1 (we assume no frictions) is $x-c+\phi Z$. We assume that all wealth at time 1 is consumed at time 1.

Under quasi-arithmetic recursive preferences, the aim of the agent is to maximise the certainty equivalent value
\begin{equation}
\label{eq:RUfg} f^{-1} \left( f(c) + f \circ g^{-1} \left( \E[ g(x-c+\phi Z) ] \right) \right)
\end{equation}
where $f$ and $g$ are increasing, concave functions. The example to consider is $f(z)=\frac{z^{1-S}}{1-S}$ and $g(z)=\frac{z^{1-R}}{1-R}$ in which case the recursive utility is precisely the discrete-time (one-period) analogue of the EZ-SDU we are considering in the main body of the paper.

We specialise to the case of $g(z)=\frac{z^{1-R}}{1-R}$, but leave $f$ general. Then, with $w = \frac{\phi}{x-c}$ we find $g^{-1} \left( \E[ g(x-c+\phi Z) ] \right) = (x-c) D(w)$ where $D(w) = g^{-1}(\E[g(1+ wZ)]) = \left( \E\left[ (1+wZ)^{1-R} \right] \right)^{\frac{1}{1-R}}$. Then maximising \eqref{eq:RUfg} simplifies to first maximising $D$ with respect to $w$, and then with $\hat{w}=\argmax D(w)$, maximising
\begin{equation}
\label{eq:RUfR} f^{-1} \left( f(c) + f \left( (x-c)D(\hat{w}) \right) \right)
\end{equation}
We can also consider the problem of maximising the certainty equivalent value for a given investment strategy.
Suppose $w$ (and therefore $D=D(w)$) is fixed. Then, since $f$ is increasing, it is sufficient to maximise $f(c) + f(D(x-c))$. Assuming differentiability of $f$ (and an interior maximum, for which a sufficient condition is $f'(0)=\infty$) the maximiser $\hat{c} = \hat{c}(x,D)$ solves
\begin{equation}
\label{eq:RUfD}f'(c) = Df'(D(x-c)) = \frac{1}{(x-c)} \left. [y f'(y) ] \right|_{y = D(x-c)}.
\end{equation}

We want to consider the impact of changing $D$, and it is clear that this depends on whether $yf'(y)$ is increasing or decreasing in $y$.
By the concavity of $f$, the left expression of \eqref{eq:RUfD} is decreasing in $c$ and the middle expression is increasing in $c$. Therefore, if $y f'(y)$ is increasing in $y$ then $\hat{c}(x,D)$ is decreasing in $D$; conversely, if $y f'(y)$ is decreasing in $y$ then $\hat{c}(x,D)$ is increasing in $D$. For $f(z)=\frac{z^{1-S}}{1-S}$, $zf'(z)$ is decreasing precisely when $S>1$, so that $\hat{c}(x,D)$ is increasing in $D$ if and only if $S>1$\footnote{Alternatively for $f(z)= \frac{z^{1-S}}{1-S}$ of power law form, $\hat{c}$ solves $c^{-S} = D^{1-S}(x-c)^{1-S}$. It follows that $\hat{c} = D^{1-\frac{1}{S}} x$ which is increasing in $D$ if and only if $S>1$.}.

Finally we combine this argument with the impact of using a sub-optimal investment strategy as in the continuous-time case: if $w^*$ is optimal then $D(w)<D(w^*)$ and optimal consumption in markets with frictions is lower than in the frictionless case if and only if $S>1$.

The third argument is related to Remark~\ref{rem:shapeofm} and connects the above analysis to the impact of $1/S$ as a measure of the inter-temporal substitutability of consumption -- high $S$ indicates a strong preference for smooth consumption over time. Consider the first equality in \eqref{eq:RUfD}. $D$ enters into the expression $Df'(D(x-c))$ twice. First, the factor $D$ means that decreasing $D$ lowers the expression, resulting in higher optimal consumption. This role for $D$ is related to the impact of changes of future wealth on consumption and is independent of the preferences expressed via $f$. Second $D$ is an argument of the factor $f'(D(x-c))$ -- here lowering $D$ increases the marginal time-1 utility and therefore, in isolation, decreases time-0 consumption. This factor depends on the preferences via $f'$ and is strongest when $S$ is large. Thus, when $S$ is small, the first effect dominates and using suboptimal investment strategies leads to increased consumption. When $S$ is large, the desire for smooth consumption increases, the second effect dominates, and time-0 consumption falls.

\section{Concluding remarks}
In this paper we studied the optimal investment-consumption problem over the infinite horizon in a Black--Scholes--Merton financial market. Relative to the standard frictionless case first investigated by Merton~\cite{Merton:69,Merton:71}, the innovations are that we consider proportional transaction costs and non-additive preferences, as captured by Epstein--Zin stochastic differential utility. Our results determine when the problem is well-posed, and in this case, we give a complete description of the optimal investment and consumption strategies. A major innovation is our focus on the {\em shadow fraction of wealth in the risky asset}: it turns out that with this as the primary variable, the problem reduces to the study of a free-boundary problem for one-dimensional, singular ODE.

Even for additive (CRRA) preferences (studied first by Magill and Constantinides~\cite{MagillConstantinides:76} and Davis and Norman~\cite{DavisNorman:90}) our approach brings substantial simplifications and insights, and we are able give new (financial) interpretations to many of the quantities which arise. 

The non-additive case brings extra technical challenges (especially as regards existence and uniqueness of the utility process associated with a given consumption stream, and the comparison theorem used to facilitate the verification lemma); we overcome these challenges by making use of previous results in the frictionless case, see \cite{herdegen:hobson:jerome:23B}.

Study of the non-additive case is increasingly important as there is mounting evidence that the additive case fails to accurately describe investor behaviour. This paper is the first to give a full solution of the problem in the non-additive case (previous results only covered the small transaction cost case and only for an incomplete set of parameter combinations) and therefore is the first to make it possible to properly study the comparative statics of the problem and to distinguish the separate roles played by the risk aversion (which primarily affects the width and location of the no-transaction wedge), and the elasticity of intertemporal complementarity (which primarily affects the consumption rate).

The methods of this paper rely crucially on scaling properties of the problem (including the dynamics of the price process, the proportional transaction costs and the form of the recursive utility) to reduce the dimensionality, and it looks difficult to extend them beyond this case. Notwithstanding, looking at the extant literature in the additive case, it seems that there are at least two natural extensions which would become tractable for SDU using the methods of this paper.

The first extension is the problem with infinite transaction costs (and then it is important to distinguish between infinite transaction costs on selling and infinite transaction costs on buying), which is a precursor to an expansion solution for problems with a large transaction parameter.
The case $R,S>1$ and $m(0)<0<m(1)$ seems especially interesting as the point where $\ell$ crosses zero becomes a singular point for the solution $n_\infty$.

The second extension is the problem with multiple assets. At the general level, with transaction costs to be paid on all exchanges of assets (perhaps also requiring all transactions to be made to or from the cash account) the intuition is that the no-transaction region becomes a no-transaction cone, see for example Chen and Dai~\cite{Chen:Dai:13}. However, whereas in the one-risky asset case, the no-trade region is an interval which can be characterised by its endpoints, in the multi-asset case, the cross-section of the no-transaction region will be convex, but otherwise may be complicated in shape and difficult to characterise. However, results in the additive case (especially Hobson et al~\cite{hobson:tse:zhu:19B} and Choi~\cite{Choi:20}) give hope that it is possible to find a complete solution in a problem with multiple risky assets in which transaction costs are levied on purchases and sales of one asset only. Following the approach of this paper, the aim is to reduce the problem via the shadow fraction of wealth to the solution of a free-boundary problem for a first order differential equation. We even expect that apart from special cases, the resulting ODE is simpler than that in the single asset case studied in this paper in that there is no singularity at 1.

These extensions and other related problems are left as future work.

\appendix

\section{Additional results and proofs for the problem with frictions} 
\label{sec:app:tcrigour}

\subsection{The free boundary problem}
\label{app:ode}
The following two results lie at the heart of Proposition \ref{prop:n}.

\begin{prop}
\label{prop:n first:R<1}
Let $R < 1$. Assume that $m(0), m(1) > 0$.
\begin{enumerate}
\item Let $q_M \in (0, 1)$. Then for each $z \in \{m > 0\} \cap (0, q_M)$, there exist a unique $\zeta^+(z) \in  \{m > 0\} \cap (q_M, 1)$ and a unique continuously differentiable decreasing function $n_{(z)}^+: [z, \zeta^+(z)] \to [m(z), m(\zeta^+(z))]$ with $n_{(z)}^+(q) < \ell(q)$ for $q \in  [z, \zeta^+(z)]$ that satisfies on $ [z, \zeta^+(z)]$ the ODE
\begin{equation}
\frac{\diff}{\diff q} n_{(z)}^+(q) = O(n_{(z)}^+(q), q).
\end{equation}
as well as the boundary conditions
\begin{equation}
n_{(z)}^+(z) = m(z) \quad \text{and} \quad n_{(z)}^+(\zeta^+(z)) = m(\zeta^+(z)).
\end{equation}
Moreover, the map $\Sigma^+:\{m > 0\}  \cap (0, q_M) \to (1, \infty)$ given by
\begin{equation*}
\Sigma^+(z) := \exp\left(\int_{z}^{\zeta^+(z)} -\frac{1}{q(1-q)} \frac{m(q) -n_{(z)}^+(q)}{\ell(q)-n_{(z)}^+(q)} \dd q\right)
\end{equation*}
is continuous and decreasing, and satisfies
\begin{align}
\lim_{z \downarrow 0=\inf(\{m > 0\} \cap (0, q_M))} \Sigma^+(z) &= \infty\\
\lim_{z \uparrow\sup(\{m > 0\} \cap (0, q_M))} \Sigma^+(z) &=
\begin{cases}
\exp\left(\int_{q^m_-}^{q^m_+} -\frac{1}{q(1-q)} \frac{m(q)}{\ell(q)} \dd q\right) &\text{if } m(q_M) < 0 \\
1 & \text{if } m(q_M) \geq 0.
\end{cases}
\end{align}

\item Let $q_M =1$. Then for each $z \in \{m > 0\} \cap (0, 1)=(0,1)$, there exist a unique continuously differentiable decreasing function $n_{(z)}^+: [z, 1] \to [m(z), m(1)]$ with $n_{(z)}^+(q) < \ell(q)$ for $q \in [z, 1)$ that satisfies on $ [z, 1)$ the ODE
\begin{equation}
\frac{\diff}{\diff q} n_{(z)}^+(q) = O(n_{(z)}^+(q), q).
\end{equation}
as well as the boundary conditions
\begin{equation}
n_{(z)}^+(z) = m(z) \quad \text{and} \quad n_{(z)}^+(1) = m(1) \quad \text{and} \quad \frac{\diff}{\diff q} n_{(z)}^+(1) = m'(1)=0.
\end{equation}
Moreover, the map $\Sigma^+:\{m > 0\}  \cap (0, 1) \to (1, \infty)$ given by
\begin{equation*}
\Sigma^+(z) := \exp\left(\int_{z}^{1} -\frac{1}{q(1-q)} \frac{m(q) -n_{(z)}^+(q)}{\ell(q)-n_{(z)}^+(q)} \dd q\right)
\end{equation*}
is continuous and decreasing, and satisfies
\begin{align}
\lim_{z \downarrow 0= \inf(\{m > 0\} \cap (0, 1))} \Sigma^+(z) &= \infty\\
\lim_{z \uparrow 1 = \sup(\{m > 0\} \cap (0, 1))} \Sigma^+(z) &=1.
\end{align}

\item Let $q_M > 1$. Then for each $z \in \{m > 0\} \cap (0, q_M)$, there exist a unique $\zeta^+(z) \in  \{m > 0\} \cap (q_M, 2q_M-1)$ and a unique continuously differentiable decreasing function $n_{(z)}^+: [z, \zeta^+(z)] \to [m(z), m(\zeta^+(z))]$ with $n_{(z)}^+(q) \neq \ell(q)$ for $q \in [z, \zeta^+(z)]  \setminus \{1\}$ that satisfies on $ [z, \zeta^+(z)] \setminus \{1\}$ the ODE
\begin{equation}
\frac{\diff}{\diff q} n_{(z)}^+(q) = O(n_{(z)}^+(q), q).
\end{equation}
as well as the boundary conditions
\begin{equation}
n_{(z)}^+(z) = m(z) \quad \text{and} \quad n_{(z)}^+(\zeta^+(z)) = m(z)
\end{equation}
and, if $z<1$, the crossing condition
\begin{equation}
n_{(z)}^+(1) = m(1) \quad \text{and} \quad \frac{\diff}{\diff q} n_{(z)}^+(1) = m'(1). 
\end{equation}
Moreover, the map $\Sigma^+:\{m > 0\}  \cap (0, q_M) \to (1, \infty)$ given by
\begin{equation*}
\Sigma^+(z) := \exp\left(\int_{z}^{\zeta^+(z)} -\frac{1}{q(1-q)} \frac{m(q) -n_{(z)}^+(q)}{\ell(q)-n_{(z)}^+(q)} \dd q\right)
\end{equation*}
is continuous and decreasing, and satisfies
\begin{align}
\lim_{z \downarrow 0 = \inf(\{m > 0\} \cap (0, q_M))} \Sigma^+(z) &= \infty\\
\lim_{z \uparrow\sup(\{m > 0\} \cap (0, q_M))} \Sigma^+(z) &=
\begin{cases}
\exp\left(\int_{q^m_-}^{q^m_+} -\frac{1}{q(1-q)} \frac{m(q)}{\ell(q)} \dd q\right) &\text{if } m(q_M) < 0 \\
1 & \text{if } m(q_M) \geq 0.
\end{cases}
\end{align}

\item Let $q_M < 0$. Then for each $z \in \{m > 0\} \cap (q_M,0)$, there exist a unique $\zeta^-(z) \in  \{m > 0\} \cap (-\infty, q_M)$ and a unique continuously differentiable decreasing function $n_{(z)}^-: [\zeta^-(z),z] \to [m(\zeta^-(z)),m(z)]$ with $n_{(z)}^-(q) \neq \ell(q)$ for $q \in [\zeta^-(z),z]  \setminus \{1\}$ that satisfies on $ [\zeta^-(z),z]$ the ODE
\begin{equation}
\frac{\diff}{\diff q} n_{(z)}^-(q) = O(n_{(z)}^-(q), q).
\end{equation}
as well as the boundary conditions
\begin{equation}
n_{(z)}^-(z) = m(z) \quad \text{and} \quad n_{(z)}^-(\zeta^-(z)) = m(z).
\end{equation}
Moreover, the map $\Sigma^-:\{m > 0\}  \cap (q_M,0) \to (1, \infty)$ given by
\begin{equation*}
\Sigma^-(z) := \exp\left(\int_{\zeta^-(z)}^z -\frac{1}{q(1-q)} \frac{m(q) -n_{(z)}^+(q)}{\ell(q)-n_{(z)}^+(q)} \dd q\right)
\end{equation*}
is continuous and decreasing, and satisfies
\begin{align}
\lim_{z \uparrow 0 = \sup(\{m > 0\} \cap (q_M,0))} \Sigma^-(z) &= \infty\\
\lim_{z \downarrow\inf(\{m > 0\} \cap (q_M,0))} \Sigma^-(z)
&=
\begin{cases}
\exp\left(\int_{q^m_-}^{q^m_+} -\frac{1}{q(1-q)} \frac{m(q)}{\ell(q)} \dd q\right) &\text{if } m(q_M) <0 \\
1 & \text{if } m(q_M) \geq 0.
\end{cases}
\end{align}

\end{enumerate}

\end{prop}

\begin{prop}
\label{prop:n first:R>1}
Let $R > 1$. Assume that $m(q_M) > 0$.
\begin{enumerate}
\item Let $q_M \in (0, 1)$. Then for each $z \in \{m > 0\} \cap (0, q_M)$, there exist a unique $\zeta^+(z) \in  \{m > 0\} \cap (q_M, 1)$ and a unique continuously differentiable increasing function $n_{(z)}^+: [z, \zeta^+(z)] \to [m(z), m(\zeta^+(z))]$ with $n_{(z)}^+(q) \neq \ell(q)$ for $q \in [z, \zeta^+(z)]  $ that satisfies on $ [z, \zeta^+(z)]$ the ODE
\begin{equation}
\frac{\diff}{\diff q} n_{(z)}^+(q) = O(n_{(z)}^+(q), q)
\end{equation}
with the boundary conditions
\begin{equation}
n_{(z)}^+(z) = m(z) \quad \text{and} \quad n_{(z)}^+(\zeta^+(z)) = m(z) .
\end{equation}
Moreover, the map $\Sigma^+:\{m > 0\}  \cap (0, q_M) \to (1, \infty)$ given by
\begin{equation*}
\Sigma^+(z) := \exp\left(\int_{z}^{\zeta^+(z)} -\frac{1}{q(1-q)} \frac{m(q) -n_{(z)}^+(q)}{\ell(q)-n_{(z)}^+(q)} \dd q\right)
\end{equation*}
is continuous and decreasing, and satisfies
\begin{align}
\lim_{z \uparrow q_M=\sup(\{m > 0\} \cap (0, q_M))} \Sigma^+(z) &= 1,\\
\lim_{z \downarrow\inf(\{m > 0\} \cap (0, q_M))} \Sigma^+(z)
&=
\begin{cases}
\exp\left(\int_{q^m_-}^{q^m_+} -\frac{1}{q(1-q)} \frac{m(q)}{\ell(q)} \dd q\right) &\text{if } m(0) < 0,\\
\infty & \text{if } m(0) \geq 0.
\end{cases}
\end{align}

\item Let $q_M =1$. Then for each $z \in \{m > 0\} \cap (0, q_M)$, there exist a unique continuously differentiable increasing function $n_{(z)}^+: [z, 1] \to [m(z), m(1)]$ with $n_{(z)}^+(q) \neq \ell(q)$ for $q \in [z, 1)$ that satisfies on $ [z, 1)$ the ODE
\begin{equation}
\frac{\diff}{\diff q} n_{(z)}^+(q) = O(n_{(z)}^+(q), q)
\end{equation}
with the boundary conditions
\begin{equation}
n_{(z)}^+(z) = m(z) \quad \text{and} \quad n_{(z)}^+(1) = m(1) \quad \text{and} \quad \frac{\diff}{\diff q} n_{(z)}^+(1) = m'(1) = 0.
\end{equation}
Moreover, the map $\Sigma^+:\{m > 0\}  \cap (0, 1) \to (1, \infty)$ given by
\begin{equation*}
\Sigma^+(z) :=\exp \left(\int_{z}^{1} -\frac{1}{q(1-q)} \frac{m(q) -n_{(z)}^+(q)}{\ell(q)-n_{(z)}^+(q)} \dd q\right)
\end{equation*}
is continuous and decreasing, and satisfies
\begin{align}
\lim_{z \uparrow 1 = q_M = \sup(\{m > 0\} \cap (0, q_M))} \Sigma^+(z) &= 1,\\
\lim_{z \downarrow\inf(\{m > 0\} \cap (0, q_M))} \Sigma^+(z)
&=
\begin{cases}
\exp\left(\int_{q^m_-}^{q^m_+} -\frac{1}{q(1-q)} \frac{m(q)}{\ell(q)} \dd q\right) &\text{if } m(0) < 0,\\
\infty & \text{if } m(0) \geq 0.
\end{cases}
\end{align}
\item Let $q_M > 1$. Then for each $z \in \{m > 0\} \cap (0, q_M)$, there exist a unique $\zeta^+(z) \in  \{m > 0\} \cap (q_M, 2q_M-1)$ and a unique continuously differentiable increasing function $n_{(z)}^+: [z, \zeta^+(z)] \to [m(z), m(\zeta^+(z))]$ with $n_{(z)}^+(q) \neq \ell(q)$ for $q \in [z, \zeta^+(z)]  \setminus \{1\}$ that satisfies on $ [z, \zeta^+(z)] \setminus \{1\}$ the ODE
\begin{equation}
\frac{\diff}{\diff q} n_{(z)}^+(q) = O(n_{(z)}^+(q), q)
\end{equation}
with the boundary conditions
\begin{equation}
n_{(z)}^+(z) = m(z) \quad \text{and} \quad n_{(z)}^+(\zeta^+(z)) = m(z)
\end{equation}
and, if $z<1$, the crossing condition at 1,
\begin{equation}
n_{(z)}^+(1) = m(1) \quad \text{and} \quad \frac{\diff}{\diff q} n_{(z)}^+(1) = m'(1). 
\end{equation}
Moreover, the map $\Sigma^+:\{m > 0\}  \cap (0, q_M) \to (1, \infty)$ given by
\begin{equation*}
\Sigma^+(z) :=\exp \left(\int_{z}^{\zeta^+(z)} -\frac{1}{q(1-q)} \frac{m(q) -n_{(z)}^+(q)}{\ell(q)-n_{(z)}^+(q)} \dd q\right)
\end{equation*}
is continuous and decreasing, and satisfies
\begin{align}
\lim_{z \uparrow q_M = \sup(\{m > 0\} \cap (0, q_M))} \Sigma^+(z) &= 1,\\
\lim_{z \downarrow\inf(\{m > 0\} \cap (0, q_M))} \Sigma^+(z)
&=
\begin{cases}
\exp\left(\int_{q^m_-}^{q^m_+} -\frac{1}{q(1-q)} \frac{m(q)}{\ell(q)} \dd q\right) &\text{if } m(0) < 0,\\
\infty & \text{if } m(0) \geq 0.
\end{cases}
\end{align}

\item Let $q_M < 0$. Then for each $z \in \{m > 0\} \cap (q_M,0)$, there exist a unique $\zeta^-(z) \in  \{m > 0\} \cap (-\infty, q_M)$ and a unique, continuously differentiable, increasing function $n_{(z)}^-: [\zeta^-(z),z] \to [m(\zeta^-(z)),m(z)]$ with $n_{(z)}^-(q) \neq \ell(q)$ for $q \in [\zeta^-(z),z]  \setminus \{1\}$ that satisfies on $ [\zeta^-(z),z]$ the ODE
\begin{equation}
\frac{\diff}{\diff q} n_{(z)}^-(q) = O(n_{(z)}^-(q), q)
\end{equation}
as well as the boundary conditions
\begin{equation}
n_{(z)}^-(z) = m(z) \quad \text{and} \quad n_{(z)}^-(\zeta^-(z)) = m(z).
\end{equation}
Moreover, the map $\Sigma^-:\{m > 0\}  \cap (q_M,0) \to (1, \infty)$ given by
\begin{equation*}
\Sigma^-(z) :=\exp \left(\int_{\zeta^-(z)}^z -\frac{1}{q(1-q)} \frac{m(q) -n_{(z)}^+(q)}{\ell(q)-n_{(z)}^+(q)} \dd q\right)
\end{equation*}
is continuous and decreasing, and satisfies
\begin{align}
\lim_{z \downarrow q_M = \inf(\{m > 0\} \cap (q_M,0))} \Sigma^-(z) &= 1,\\
\lim_{z \uparrow\sup(\{m > 0\} \cap (q_M,0))} \Sigma^-(z)
&=
\begin{cases}
\exp\left(\int_{q^m_-}^{q^m_+} -\frac{1}{q(1-q)} \frac{m(q)}{\ell(q)} \dd q\right) &\text{if } m(0) < 0,\\
\infty & \text{if } m(0) \geq 0.
\end{cases}
\end{align}

\end{enumerate}

\end{prop}

\begin{proof}[Proof of Proposition \ref{prop:n first:R<1} and \ref{prop:n first:R>1}]
Complete proofs are omitted but a structurally identical family of first order ODEs is studied in \cite{hobson:tse:zhu:19A}, Section 3. See also Section 3.5 of \cite{tse16} and Section 5 of \cite{hobson:tse:zhu:19B} for the analysis of a similar problem in the context of optimal investment and consumption under proportional transaction costs with multiple assets.
\end{proof}

\begin{proof}[Proof of Proposition \ref{prop:n}] We only establish (a) -- (f) for the case $q_M \in (0, 1)$, the argument for the other three cases is similar.

First, we show existence of $q_*(\xi)$, $q^*(\xi)$ and $n_\xi$.  For $z \in \{m > 0\} \cap (0, q_M)$, define $\zeta^+(z) \in \{m > 0\} \cap (q_M, 1)$ and the function $n^+_{(z)}: [z, \zeta^+(z)] \to [m(z), m(\zeta^+(z))]$ as in Propositions \ref{prop:n first:R<1}(a) or \ref{prop:n first:R>1}(a), i.e. $n^+_{(z)}$ solves $n' = O(n,q)$ subject to $n^+_{(z)}(z)=m(z)$, and then crosses $m$ at $\zeta^+(z)$. Further, define the map $\Sigma^+:\{m > 0\}  \cap (0, q_M) \to (1, \infty)$ by
\begin{equation*}
\Sigma^+(z) := \exp\left(\int_{z}^{\zeta^+(z)} -\frac{1}{q(1-q)} \frac{m(q) -n_{(z)}^+(q)}{\ell(q)-n_{(z)}^+(q)} \dd q\right)
\end{equation*}
Since $\xi$ is in the well-posedness range of transaction costs and the function $\Sigma^+$ is increasing and continuous, it follows from Propositions \ref{prop:n first:R<1}(a) or \ref{prop:n first:R>1}(a) that there exists a unique $z \in \{m > 0\} \cap (0, q_M)$ such that $\Sigma^+(z) = \xi$. Set $q_*(\xi) := z$, $q^*(\xi) := \zeta^+(z)$ and $n_\xi := n_{(z)}^+$. Then it follows from Propositions \ref{prop:n first:R<1}(a) or \ref{prop:n first:R>1}(a) that $n_\xi$ satisfies properties (a) -- (f). Uniqueness follows since $\Sigma^+$ is strictly increasing.

Finally, we establish \eqref{eq:prop:n:q up boundE} and\eqref{eq:prop:n:q down boundE}. We only provide the argument for \eqref{eq:prop:n:q up boundE} and for the case $q_M>0$: the argument for \eqref{eq:prop:n:q down boundE} or for $q_M<0$ is similar.

Suppose $0<q^* \leq 1$. Then $(\xi - 1)(q^*)^2 \leq (\xi -1) q^*$ which can be rewritten as $q^*(1+ (\xi-1)q^*) \leq \xi q^*$. Then $q_* < q_M \leq q^* \leq \frac{\xi q^*}{(1+ (\xi-1)q^*)}$. Similarly if $q_* \leq 1 < q^*$ then $1 + (\xi-1)q^* < q^* + (\xi -1)q^* = \xi q^*$ so that $q_* \leq 1 < \frac{\xi q^*}{(1+ (\xi-1)q^*)}$. Finally, the non-trivial case is when $1 < q_* < q_M < q^*$. Then on $(q_*(\xi) , q^*(\xi) )$, $n_\xi \geq m \geq \ell$ if $R < 1$ and $n_\xi \leq m \leq \ell$ if $R >1$. In either case, this implies that $0 < \tfrac{m(q) - n_\xi(q)}{l(q) - n_\xi(q)} < 1$ on $(q_*(\xi) , q_*(\xi) )$. This together with \eqref{eq:prop:n:int const} gives
\[ \ln(\xi) < \int_{q_*(\xi)}^{q^*(\xi)} \frac{1}{q(q-1)} \dd q = \log \frac{q^*(\xi)-1}{q^*(\xi)} - \log \frac{q_*(\xi)-1}{q_*(\xi)}
\]
Rearranging terms yields the first inequality in \eqref{eq:prop:n:q up boundE}. The second inequality is trivial.
\end{proof}

The following corollary is useful in the study of the ill-posed case.

\begin{cor}
\label{cor:n}
Suppose $\ol \xi \in (1,\infty)$.

Suppose $R<1$. Then $\lim_{\xi \downarrow \ol \xi} \inf_{q \in [q_*(\xi),q^*(\xi)]} n_\xi (q) = 0$.

Suppose $R<1$. Then $\lim_{\xi \downarrow \ol \xi} \inf_{q \in [q_*(\xi),q^*(\xi)]} n_\xi (q) = 0$.
\end{cor}

\begin{proof}
This follows from the fact that $\Sigma^+$ is increasing and continuous.
\end{proof}

\begin{proof}[Proof of Corollary \ref{cor:kappa xi}]
Uniqueness of $\kappa_{\gamma^\uparrow,\gamma^\downarrow}$ follows from standard ODE theory, using the (double) boundary condition in \eqref{eq:cor:kappa:bound}.

To establish existence and argue the remaining claims, define  $\kappa_{\gamma^\uparrow,\gamma^\downarrow}: [q_*(\xi), q^*(\xi)] \to  (0, \infty)$ by
\begin{align}
\kappa_{\gamma^\uparrow,\gamma^\downarrow}(q) := \gamma^\uparrow\exp\left(\frac{S}{1-S}\int_{q_*(\xi)}^q\frac{n_{\xi}'(v)}{vn_{\xi}(v)} \dd v\right) =\frac{1}{\gamma^\downarrow}\exp\left(-\frac{S}{1-S}\int_{q}^{q^*(\xi)}\frac{n_{\xi}'(v)}{vn_{\xi}(v)}dv\right), \label{eq:pf:cor:kappa xi}
\end{align}
where the equality in \eqref{eq:pf:cor:kappa xi} follows from the integral constraint \eqref{eq:prop:n:int const} for $n_\xi$ and the fact that $\gamma^\uparrow \gamma^\downarrow = \xi$. Then $\kappa_{\gamma^\uparrow, \gamma^\downarrow}$ is well-defined (since $0 \notin [q_*(\xi), q^*(\xi)]$) and satisfies \eqref{eq:cor:kappa:bound}.
Moreover, $\kappa_{\gamma^\uparrow,\gamma^\downarrow}$ is continuously differentiable with derivative
\begin{align}
\label{eq:pf:cor:kappa:kappa prime:01}
	\kappa'_{\gamma^\uparrow,\gamma^\downarrow}(q) = \frac{S}{1-S} \frac{n'_\xi(q)}{ q n_\xi(q)} \kappa_{\gamma^\uparrow,\gamma^\downarrow}(q).
\end{align}
Plugging in $q_*(\xi)$ and $q^*(\xi)$ and using \eqref{eq:prop:n:bound:deriv} gives \eqref{eq:cor:kappa:bound:deriv}.
Moreover, using the  the ODE \eqref{eq:prop:n:ODE} for $n$, we obtain that $\kappa_{\gamma^\uparrow,\gamma^\downarrow}$ satisfies the ODE \eqref{eq:cor:kappa:ODE} on $[q_*(\xi) , q^*(\xi) ] \setminus \{1\}$. Finally, using that $(1-R)$ and $(1-S)$ have the same sign and that $(1-R) n_\xi$ is decreasing by Proposition \ref{prop:n}(b), \eqref{eq:pf:cor:kappa:kappa prime:01} directly yields (a).

The additional claim is straightforward, given \eqref{eq:pf:cor:kappa xi}.
\end{proof}

\subsection{No-transaction region}
We proceed to list some properties of the Möbius transformations $\tau_c(q)$ for $c \in (0, \infty)$ and $q \in \RR \cup \{\overline \infty\}$. First, for fixed $c$, $\tau_c$ is continuously differentiable and strictly increasing on $(\ol \infty, \frac{1}{1-c}=\tau_{\frac{1}{c}}(\ol \infty))$ and $(\frac{1}{1-c} =\tau_{\frac{1}{c}}(\ol \infty), \ol \infty)$. Moreover, for fixed $q$,  $\tau_c$ is also continuously differentiable with respect to the parameter $c$. For future reference, we note that for
$q \in \R \setminus \{\frac{1}{1-c},0,1\}$,
\begin{align}
\tau'_c(q) &= \frac{\tau_c(q)(1- \tau_c(q))}{q(1-q)}, \label{eq:tau c prime:q} \\
\frac{\partial}{\partial c} \tau_c(q) &= \frac{1}{c} \tau_c(q)(1- \tau_c(q)).  \label{eq:tau c prime:c}
\end{align}

\begin{proof}[Proof of Lemma \ref{lem:p low p up}]
If, for example, $q_M \in (0,1)$ then the inequalities follow easily since, by the properties of $\tau_c$ on $(0,1)$, $p_*(\gamma^\uparrow, \gamma^\downarrow) =  \tau_{\frac{1}{\gamma^\uparrow}}(q_*(\xi )) \leq q_*(\xi ) < q_M < q^*(\xi) \leq \tau_{\gamma^\downarrow}(q^*(\xi ))=p_*(\gamma^\uparrow, \gamma^\downarrow)$. The result for $q_M=1$ is very similar.

For other values of $q_M$ the result is a little more subtle. We consider the case when $q_M>1$, the case $q_M<0$ is similar and omitted. If $q_* \leq 1$ we have $p_*(\gamma^\uparrow, \gamma^\downarrow) \leq 1 < \tau_{\gamma^\downarrow}(q^*(\xi )) = p^*(\gamma^\uparrow, \gamma^\downarrow)$. So, suppose $1 < q_* < q_M$.
By the final result of Proposition~\ref{prop:n} we have
\[ q_*(\xi) < \frac{\xi q^*(\xi)}{1 + (\xi - 1) q^*(\xi)} = \tau_\xi(q^*(\xi)). \]
Then $p_*(\gamma^\uparrow, \gamma^\downarrow) = \tau_{\frac{1}{\gamma^\uparrow}} (q^*(\xi))< \tau_{\frac{1}{\gamma^\uparrow}} \circ \tau_\xi(q^*(\xi)) = \tau_{\gamma^\downarrow}(q^*(\xi )) = p^*(\gamma^\uparrow, \gamma^\downarrow)$.

For the final inequality in the case $q_M>1$ note that in the case $R<1$, since $n_\xi$ is decreasing on $(q_*(\xi) \vee 1,q^*(\xi))$ and $m$ is a quadratic with minimum at $q_M$, we must have $q^*(\xi)-q_M < q_M - (1 \vee q_*(\xi))$. In particular, $q^*(\xi) < 2q_M-1$ and  $p^*(\gamma^\uparrow, \gamma^\downarrow)= \tau_{\gamma^\downarrow}(q^*(\xi )) < \tau_{\gamma^\downarrow}(2 q_M-1)$. The argument for $R>1$ is very similar.
\end{proof}

\begin{proof}[Proof of Proposition\ref{prop:q of p}]
Uniqueness of $q_{\gamma^\uparrow, \gamma^\downarrow}$ follows from standard ODE theory, using the (double) boundary conditions in (b).

To establish existence of $q_{\gamma^\uparrow, \gamma^\downarrow}$, we first define its inverse function. Motivated by \eqref{eq:J q cons}, we let  $p_{\gamma^\uparrow, \gamma^\downarrow}: [q_*(\xi), q^*(\xi)] \to  \RR$ be given by
\begin{align}
	p_{\gamma^\uparrow,\gamma^\downarrow}(q) :=\frac{q}{\kappa_{\gamma^\uparrow, \gamma^{\downarrow}}(q)  + (1 - \kappa_{\gamma^\uparrow, \gamma^{\downarrow}}(q) )q} = \tau_{\frac{1}{\kappa_{\gamma^\uparrow, \gamma^{\downarrow}}(q)}}(q).
\end{align}
The boundary conditions \eqref{eq:cor:kappa:bound} of $\kappa_{\gamma^\uparrow, \gamma^{\downarrow}}$ give
\begin{equation}
p_{\gamma^\uparrow,\gamma^\downarrow}(q_*(\xi)) = \tau_{\frac{1}{\gamma^\uparrow}}(q_*(\xi)) = p_*(\gamma^\uparrow, \gamma^\downarrow)  \quad \text{and} \quad
p_{\gamma^\uparrow,\gamma^\downarrow}(q^*(\xi)) = \tau_{\gamma^\downarrow}(q_*(\xi)) =p^*(\gamma^\uparrow, \gamma^\downarrow),
\end{equation}
and the fact that $\tau_c(1) = 1$ for all $c \in (0, \infty)$, gives $p_{\gamma^\uparrow,\gamma^\downarrow}(1)  = 1$ if $1 \in [q_*(\xi), q^*(\xi)]$.
Moreover, $p_{\gamma^\uparrow,\gamma^\downarrow}$ is continuously differentiable because $\kappa_{\gamma^\uparrow,\gamma^\downarrow}$ is. Using \eqref{eq:tau c prime:q} and \eqref{eq:tau c prime:c}, the ODE \eqref{eq:cor:kappa:ODE} for $\kappa$, and the simple identity
\begin{equation*}
\frac{p(1-p)}{q(1-q)} = \frac{m(p) - \ell(p)}{m(q) - \ell(q)}, \quad p \in \RR, q \in \RR \setminus \{0, 1\},
\end{equation*}
we obtain for $q \in [q_*(\xi), q^*(\xi)] \setminus \{1\}$,
\begin{align}
	p'_{\gamma^\uparrow,\gamma^\downarrow}(q) &= 	\frac{p_{\gamma^\uparrow,\gamma^\downarrow}(q)(1 -p_{\gamma^\uparrow,\gamma^\downarrow}(q))}{q(1-q)} - \frac{p_{\gamma^\uparrow,\gamma^\downarrow}(q)(1 -p_{\gamma^\uparrow,\gamma^\downarrow}(q))}{q(1-q)} \frac{m(q) - n_\xi(q)}{\ell(q) - n_\xi(q)} \\
&= \frac{p_{\gamma^\uparrow,\gamma^\downarrow}(q)(1 -p_{\gamma^\uparrow,\gamma^\downarrow}(q))}{q(1-q)} \frac{\ell(q) - m(q)}{\ell(q) - n_\xi(q)} \label{eq:pf:prop:q of p:h prime:middle}\\
&= \frac{\ell(p_{\gamma^\uparrow,\gamma^\downarrow}(q)) - m(p_{\gamma^\uparrow,\gamma^\downarrow}(q))}{\ell(q) - n_\xi(q)}.
\label{eq:pf:prop:q of p:h prime}
\end{align}
Looking at \eqref{eq:pf:prop:q of p:h prime:middle} shows that $p'_{\gamma^\uparrow, \gamma^\downarrow}$ is positive. Indeed, the first factor in \eqref{eq:pf:prop:q of p:h prime:middle} coincides by \eqref{eq:tau c prime:q} with $\tau'_c(q) $ for $c = \frac{1}{\kappa_{\gamma^\uparrow, \gamma^{\downarrow}}(q)}$ and is therefore positive; the second factor is also positive
because $n_\xi$ and $\ell$ both lie on the same side of $m$ on $(q_*(\xi), q^*(\xi)) \setminus \{1\}$. Thus,
$p_{\gamma^\uparrow, \gamma^\downarrow}$ is strictly increasing  and maps $[q_*(\xi), q^*(\xi)]$ onto $[p_*(\gamma^\uparrow, \gamma^\downarrow), p^*(\gamma^\uparrow, \gamma^\downarrow)]$. If we now define the function $q_{\gamma^\uparrow, \gamma^\downarrow}: [q_*(\xi), q^*(\xi)] \to [q_*(\xi), q^*(\xi)]$ by
$q_{\gamma^\uparrow, \gamma^\downarrow} := p_{\gamma^\uparrow, \gamma^\downarrow}^{-1}$, then $q_{\gamma^\uparrow, \gamma^\downarrow}$ satisfies (a) to (c). (Note that since $\tau_c(1)=1$ for all $c$ we immediately get $p_{\gamma^\uparrow, \gamma^\downarrow}=1$ whence $q_{\gamma^\uparrow, \gamma^\downarrow}(1)=1$ also.)

Finally, if $\tilde \gamma^\downarrow,\tilde \gamma^\uparrow \in [1, \infty)$ are such that $\tilde \gamma^\downarrow\tilde \gamma^\uparrow = \xi$, then $\frac{\tilde \gamma^\downarrow}{\gamma^\downarrow} = \frac{\gamma^\uparrow}{\tilde \gamma^\uparrow}$ and it suffices to show that the function
$ q_{\gamma^\uparrow, \gamma^\downarrow} \circ \tau_{\frac{\tilde \gamma^\downarrow}{\gamma^\downarrow}}$ satisfies properties (a) - (c). For (a), this follows from the ODE \eqref{eq:prop:q of p} for $q_{\gamma^\uparrow, \gamma^\downarrow}$ and \eqref{eq:tau c prime:q}
\end{proof}

\subsection{The candidate value function and candidate optimal strategy}
\label{ssecapp:valuefn}
We first collect some properties of the extended functions $\ol n_\xi$ and  $\ol \kappa_{\gamma^\uparrow, \gamma^\downarrow}$ and $\ol q_{\gamma^\uparrow, \gamma^\downarrow}$. The following result follows directly from Proposition \ref{prop:n}, Corollary \ref{cor:kappa xi} and Proposition \ref{prop:q of p} together with Definition \ref{def:extension}.

\begin{prop}
\label{prop:ol kappa n q}
Let  $\gamma^\uparrow, \gamma^\downarrow \in [1, \infty)$ with $\xi:= \gamma^\uparrow \gamma^\downarrow  > 1$ be such that $\xi$ lies in the well-posedness range of transaction costs.
\begin{enumerate}
\item The functions  $\ol n_\xi$ and  $\ol \kappa_{\gamma^\uparrow, \gamma^\downarrow}$ are continuously differentiable and satisfy on $(-\infty, \infty)$ the relationship
\begin{equation*}
q \frac{\ol \kappa'_{\gamma^\uparrow, \gamma^\downarrow}(q)}{\ol \kappa_{\gamma^\uparrow, \gamma^\downarrow}(q)} = \frac{S}{1-S} \frac{\ol n'_\xi(q)}{\ol n_\xi(q)}.
\end{equation*}
\item The function $\ol q_{\gamma^\uparrow, \gamma^\downarrow}$ is continuously differentiable and on $[p_*(\gamma^\uparrow, \gamma^\downarrow), p^*(\gamma^\uparrow, \gamma^\downarrow)]$ satisfies,
\begin{equation}
p(1-p) \ol q'_{\gamma^\uparrow, \gamma^\downarrow}(p) = \frac{S}{1- S} \frac{2}{\sigma^2} \left( \ell(
\ol q_{\gamma^\uparrow, \gamma^\downarrow})(p) - n_\xi(\ol q_{\gamma^\uparrow, \gamma^\downarrow}(p))\right)
\end{equation}
and off $[p_*(\gamma^\uparrow, \gamma^\downarrow), p^*(\gamma^\uparrow, \gamma^\downarrow)]$ we have
\( 
p(1-p) \ol q'_{\gamma^\uparrow, \gamma^\downarrow}(p) = \frac{S}{1- S} \frac{2}{\sigma^2} \left( \ell(
\ol q_{\gamma^\uparrow, \gamma^\downarrow})(p) - m(
\ol q_{\gamma^\uparrow, \gamma^\downarrow}(p))\right)
\) 
\item Let $(x, y, \phi) \in \sS^\circ_{\gamma^\uparrow, \gamma^\downarrow}$. Then
\begin{equation}
x +  \phi y \ol \kappa_{\gamma^\uparrow, \gamma^\downarrow}\left(\ol q_{\gamma^\uparrow, \gamma^\downarrow} \Big(\frac{\phi y}{x + y \phi} \Big)\right) > 0,
\end{equation}
and
\begin{equation}
\frac{\phi y \ol \kappa_{\gamma^\uparrow, \gamma^\downarrow}\left(\ol q_{\gamma^\uparrow, \gamma^\downarrow} \Big(\frac{\phi y}{x + y \phi} \Big)\right)}{x +  \phi y \ol \kappa_{\gamma^\uparrow, \gamma^\downarrow}\left(\ol q_{\gamma^\uparrow, \gamma^\downarrow} \Big(\frac{\phi y}{x + y \phi} \Big)\right)} = \ol q_{\gamma^\uparrow, \gamma^\downarrow} \Big(\frac{\phi y}{x + y \phi}\Big).
\end{equation}

\end{enumerate}
\end{prop}
We proceed to study properties of the candidate value function.
\begin{prop}
\label{prop:hatV:properties}
Let  $\gamma^\uparrow, \gamma^\downarrow \in [1, \infty)$ with $\xi:= \gamma^\uparrow \gamma^\downarrow  > 1$ be such that $\xi$ lies in the well-posedness range of transaction costs. Then $\hat V^{\gamma^\uparrow, \gamma^\downarrow}$ is continuously differentiable with respect to $t, x, \phi$ and twice continuously differentiable with respect to $y$. Moreover, fix $(t, x, y, \phi) \in [0, \infty) \times \sS^\circ_{\gamma^\uparrow, \gamma^\downarrow}$. Then
\begin{enumerate}
\item $\hat V^{\gamma^\uparrow, \gamma^\downarrow}_x(t, x, y, \phi) > 0$ and $\hat V^{\gamma^\uparrow, \gamma^\downarrow}_\phi(t, x, y, \phi) > 0$, we have the identity
\begin{equation}
\frac{\hat V^{\gamma^\uparrow, \gamma^\downarrow}_\phi(t, x, y, \phi) }{y  \hat V^{\gamma^\uparrow, \gamma^\downarrow}_x(t, x, y, \phi) } = \ol \kappa_{\gamma^\uparrow, \gamma^\downarrow}\Big(\ol q_{\gamma^\uparrow, \gamma^\downarrow}\Big(\frac{\phi y}{x+ \phi y}\Big)\Big), \label{eq:prop:hatV:properties:kappa}
\end{equation}
and \eqref{eq:prop:hatV:properties:kappa} equals $\gamma^\uparrow$ if $\frac{\phi y}{x + \phi y} \in (\tau_{\frac{1}{\gamma^\uparrow}(\ol \infty)}= \frac{-1}{\gamma^\uparrow -1}, p_*(\gamma^\uparrow, \gamma^\downarrow)]$ and  $\frac{1}{\gamma^\downarrow}$ if $\frac{\phi y}{x + \phi y} \in [p^*(\gamma^\uparrow, \gamma^\downarrow), \tau_{\gamma^\downarrow}(\ol \infty)=\frac{\gamma^{\downarrow}}{\gamma^{\downarrow} -1})$, respectively.
\item 
\begin{eqnarray}
\lefteqn{g_{EZ} (t, c, \hat V^{\gamma^\uparrow, \gamma^\downarrow}(t, x, y, \phi) ) -c V^{\gamma^\uparrow, \gamma^\downarrow}_x (t, x, y, \phi)} \\
& \leq &
\theta S \hat V^{\gamma^\uparrow, \gamma^\downarrow}(t, x, y, \phi)  \ol n_{\xi}\Big(\ol q_{\gamma^\uparrow, \gamma^\downarrow}\Big(\frac{\phi y}{x+ \phi y}\Big)\Big) \label{eq:prop:hatV:properties:HJB ineq:01}
 \\
& \leq &  -\hat V^{\gamma^\uparrow, \gamma^\downarrow}_t(t, x, y, \phi)  - r x \hat V^{\gamma^\uparrow, \gamma^\downarrow}_x (t, x, y, \phi) - \mu y \hat V^{\gamma^\uparrow, \gamma^\downarrow}_y(t, x, y, \phi)  - \frac{1}{2} \sigma^2 y^2 \hat V^{\gamma^\uparrow, \gamma^\downarrow}_{yy}(t, x, y, \phi), \qquad\;\; \label{eq:prop:hatV:properties:HJB ineq:02}
\end{eqnarray}
where the inequality in \eqref{eq:prop:hatV:properties:HJB ineq:01} is an equality if $c = (x + \phi y \ol \kappa_{\gamma^\uparrow, \gamma^\downarrow}(\ol q_{\gamma^\uparrow, \gamma^\downarrow}(\frac{\phi y}{x+ \phi y}))) \ol n_{\xi}(\ol q_{\gamma^\uparrow, \gamma^\downarrow}(\frac{\phi y}{x+ \phi y}))$, and the inequality in \eqref{eq:prop:hatV:properties:HJB ineq:02} is an equality if $\frac{\phi y}{x + \phi y} \in [p_*(\gamma^\uparrow, \gamma^\downarrow), p^*(\gamma^\uparrow, \gamma^\downarrow)].$
\end{enumerate}

\end{prop}

\begin{proof}
First, it follows from the continuous differentiability of $\ol n_{\xi}$, $\ol \kappa_{\gamma^\uparrow, \gamma^\downarrow}$ and $\ol q_{\gamma^\uparrow, \gamma^\downarrow}$  that $\hat V^{\gamma^\uparrow, \gamma^\downarrow}$ is continuously differentiable with respect to all its arguments. Moreover, dropping the references to $\gamma^\uparrow$, $\gamma^\downarrow$ and $\xi$ and the arguments for simplicity and using Proposition \ref{prop:ol kappa n q}(a), we obtain
\begin{align}
\hat V_t &= -\delta \theta \hat V, \label{eq:Vt:ver} \\
\hat V_{x} & =(1-R) \hat V \frac{1}{x+\phi y \ol \kappa(\ol q)} \label{eq:Vx:ver}\\
 \hat V_{\phi} & =(1-R)\hat V \frac{y \ol \kappa(\ol q)}{x+\phi y \ol \kappa(\ol q)}, \label{eq:Vphi:ver} \\
\hat V_{y} & =(1-R) \hat V \frac{\phi \ol \kappa(\ol q)}{x+\phi y \ol \kappa(\ol q)} = (1-R) \hat V \frac{\ol q}{y} \label{eq:Vy:ver}
\end{align}
Next, it follows from the continuous differentiability of $\ol n_{\xi}$ and $\ol \kappa_{\gamma^\uparrow, \gamma^\downarrow}$ and $\ol q_{\gamma^\uparrow, \gamma^\downarrow}$ and  \eqref{eq:Vy:ver} that  $\hat{V}^{\gamma^\uparrow, \gamma^\downarrow}$ is twice continuously differentiable with respect to $y$. Again dropping the references to $\gamma^\uparrow$, $\gamma^\downarrow$ and $\xi$ and the arguments,
and using that $\hat{V}_q=0$ we obtain
\begin{eqnarray*}
y^2 \hat V_{yy} & = & \hat{V} (1-R) {\ol q} ((1-R){\ol q} - 1) + (1-R) {\hat V} y \frac{\partial p}{\partial y} \frac{\partial {\ol q}}{\partial p} \\
& = & -R(1-R) \hat V \ol q^2 + \hat{V}(1-R) \left( p(1-p) \frac{\partial {\ol q}}{\partial p} - {\ol q}(1- {\ol q}) \right)
\end{eqnarray*}
In summary,
we have
\begin{equation}
y^2 \hat V_{yy} =
\begin{cases}
-R(1-R) \hat V \ol q^2 &\text{if } \frac{\phi y}{x + \phi y} \in  (\tau_{\frac{1}{\gamma^\uparrow}}(\ol \infty), p_*(\gamma^\uparrow, \gamma^\downarrow)) \cup (p^*(\gamma^\uparrow, \gamma^\downarrow), \tau_{\gamma^\downarrow}(\ol \infty)), \\
-R(1-R) \hat V \ol q^2 + \frac{2 \theta S}{\sigma^2} \hat V (m(\ol q) - \ol n(\ol q))&\text{if } \frac{\phi y}{x + \phi y} \in
[p_*(\gamma^\uparrow, \gamma^\downarrow), p^*(\gamma^\uparrow, \gamma^\downarrow)]. \label{eq:Vyy:ver}
\end{cases}
\end{equation}

(a) This follows directly from \eqref{eq:Vx:ver} and \eqref{eq:Vphi:ver}.

(b) Using \eqref{eq:Vx:ver} and maximising $g_{EZ} (t, c, \hat V^{\gamma^\uparrow, \gamma^\downarrow}(t, x, y, \phi) ) -c V^{\gamma^\uparrow, \gamma^\downarrow}_x (t, x, y, \phi)$ over $c$, we obtain that the maximum is attained at
\begin{equation}
\hat c = \Big(x + \phi y \ol \kappa_{\gamma^\uparrow, \gamma^\downarrow}\left(\ol q_{\gamma^\uparrow, \gamma^\downarrow}\left(\frac{\phi y}{x+ \phi y}\right)\right)\Big) \ol n_{\xi}\Big(\ol q_{\gamma^\uparrow, \gamma^\downarrow}\Big(\frac{\phi y}{x+ \phi y}\Big)\Big).
\end{equation}
and that the corresponding maximal value is given by
\begin{equation}
g_{EZ} (t, \hat c, \hat V^{\gamma^\uparrow, \gamma^\downarrow}(t, x, y, \phi) ) -\hat c V^{\gamma^\uparrow, \gamma^\downarrow}_x (t, x, y, \phi) = \theta S \hat V^{\gamma^\uparrow, \gamma^\downarrow}(t, x, y, \phi)  \ol n_{\xi}\Big(\ol q_{\gamma^\uparrow, \gamma^\downarrow}\Big(\frac{\phi y}{x+ \phi y}\Big)\Big).
\end{equation}
Thus, we have \eqref{eq:prop:hatV:properties:HJB ineq:01}

Next, using \eqref{eq:Vt:ver} -- \eqref{eq:Vphi:ver} with Proposition \ref{prop:ol kappa n q}(c) and \eqref{eq:Vyy:ver} together with the definition of $m$ in \eqref{eq:m}, we obtain
\begin{align}
 -\hat V_t - r x \hat V_x - \mu y \hat V_y -\frac{1}{2} \sigma^2y^2 \hat V_{yy} =
\begin{cases}
\theta S \hat V m(\ol q)&\text{if } \frac{\phi y}{x + \phi y} \in  (\frac{1}{\gamma^\uparrow}(\ol \infty), p_*(\gamma^\uparrow, \gamma^\downarrow)) \cup (p^*(\gamma^\uparrow, \gamma^\downarrow), \tau_{\gamma^\downarrow}(\ol \infty)) ,\\
 \theta S \hat V \ol n(\ol q) &\text{if } \frac{\phi y}{x + \phi y} \in
[p_*(\gamma^\uparrow, \gamma^\downarrow), p^*(\gamma^\uparrow, \gamma^\downarrow)].
\end{cases}
\end{align}
Now, in case that $\frac{\phi y}{x + \phi y} \in
[p_*(\gamma^\uparrow, \gamma^\downarrow), p^*(\gamma^\uparrow, \gamma^\downarrow)]$, we obtain an equality in
\eqref{eq:prop:hatV:properties:HJB ineq:02}. By contrast, in case that $\frac{\phi y}{x + \phi y} \in  (\frac{1}{\gamma^\uparrow}(\ol \infty), p_*(\gamma^\uparrow, \gamma^\downarrow)) \cup (p^*(\gamma^\uparrow, \gamma^\downarrow), \tau_{\gamma^\downarrow}(\ol \infty))$, we obtain an inequality in
\eqref{eq:prop:hatV:properties:HJB ineq:02} noting that $\theta S (1-R) \hat V$ is positive and $\frac{m}{1-R}$ is a convex quadratic function so that $\frac{m(\ol q)}{1-R} \geq \frac{m(q_*(\xi))}{1-R} = \frac{\ol n(q_*(\xi))}{1-R}$ if $\ol q \leq q_*(\xi) \leq q_M$ and $\frac{m(\ol q)}{1-R} \geq \frac{m(q^*(\xi))}{1-R} = \frac{\ol n(q^*(\xi)}{1-R}= \frac{n(\ol q)}{1-R}$ if $q \geq q^*(\xi) \geq q_M$, where we have used  Proposition \ref{prop:n} for the inequalities $q_*(\xi) \leq q_M \leq q^*(\xi)$.
\end{proof}

\begin{cor}
\label{cor:hatV:ineq}
Let  $\gamma^\uparrow, \gamma^\downarrow \in [1, \infty)$ with $\xi:= \gamma^\uparrow \gamma^\downarrow  > 1$ be such that $\xi$ is in the well-posedness range of transaction costs.
\begin{enumerate}
\item
Fix $(x, y, \phi) \in  \sS^\circ_{(\gamma^\uparrow, \gamma^\downarrow)}$ and $(x', y, \phi') \in  \sS_{ \gamma^\uparrow, \gamma^\downarrow}$. Then $(x+ x', y, \phi + \phi') \in  \sS^\circ_{(\gamma^\uparrow, \gamma^\downarrow)}$ and
\begin{equation}
\hat V^{\gamma ^\uparrow, \gamma^\downarrow}(x+x', y, \phi+\phi') \geq \hat V^{\gamma ^\uparrow, \gamma^\downarrow}(x, y, \phi).
\end{equation}
\item If $(x', y, \phi') \in \sS^\circ_{\gamma^\uparrow, \gamma^\downarrow}$ is such that $x' :=  x - \gamma^\uparrow y (\phi' - \phi)^{+} + \frac{y}{\gamma^\downarrow} (\phi'- \phi)^-$ (corresponding to an initial bulk trade), then $\hat V^{\gamma^\uparrow, \gamma^\downarrow}(t, x, y, \phi) \geq \hat V^{\gamma^\uparrow, \gamma^\downarrow}(t, x', y, \phi')$. Moreover, the inequality is an equality if either $\frac{\phi y}{x + \phi y} < p_*(\gamma^\uparrow,\gamma^\downarrow)$ and $\frac{\phi' y}{x + \phi' y} = p_*(\gamma^\uparrow,\gamma^\downarrow)$ or $\frac{\phi y}{x + \phi y} > p_*(\gamma^\uparrow,\gamma^\downarrow)$ and $\frac{\phi' y}{x + \phi' y} = p_*(\gamma^\uparrow,\gamma^\downarrow)$.
\end{enumerate}
\end{cor}

\begin{proof}
(a) First, it follows from the definition of $\sS_{ \gamma^\uparrow, \gamma^\downarrow}$ and $\sS^\circ_{ \gamma^\uparrow, \gamma^\downarrow}$ that $x+\epsilon x', y, \phi+\epsilon \phi' \in \sS^\circ_{\gamma^\uparrow, \gamma^\downarrow}$ for all $\epsilon \in [0, 1]$. Hence, it suffices to show that $x' \hat V^{\gamma ^\uparrow, \gamma^\downarrow}_x(x+\epsilon x', y, \phi+\epsilon \phi') + \phi' \hat V^{\gamma ^\uparrow, \gamma^\downarrow}_\phi(x+\epsilon x', y, \phi+\epsilon \phi') \geq 0$ for all $\epsilon \in [0, 1]$. So fix $\epsilon \in [0, 1]$. Then by Proposition \ref{prop:hatV:properties}(a) and the definition of $\sS_{ \gamma^\uparrow, \gamma^\downarrow}$,
\begin{align}
&x' \hat V^{\gamma ^\uparrow, \gamma^\downarrow}_x(x+\epsilon x', y, \phi+\epsilon \phi') + \phi' \hat V^{\gamma ^\uparrow, \gamma^\downarrow}_\phi(x+\epsilon x', y, \phi+\epsilon \phi')  \\
&= V^{\gamma ^\uparrow, \gamma^\downarrow}_x(x+\epsilon x', s, \phi+\epsilon \phi') \left(x' + \phi' y \ol \kappa_{\gamma^\uparrow, \gamma^\downarrow} \Big(\ol q_{\gamma^\uparrow, \gamma^\downarrow}\Big(\frac{(\phi+\epsilon \phi')y}{x+\epsilon x' + (\phi+\epsilon \phi') y}\Big)\Big)\right) \geq 0.
\end{align}
where we use that if $\phi'>0$ then $x' + \phi'y {\ol \kappa}(\ol q) \geq x' + \frac{\phi' y}{\gamma^\uparrow} \geq 0$ since $(x',y,\phi') \in \sS_{ \gamma^\uparrow, \gamma^\downarrow}$ and if $\phi'<0$ then $x' + \phi'y {\ol \kappa}(\ol q) \geq x' + \phi' y \gamma^\downarrow \geq 0$ again since $(x',y,\phi') \in \sS_{ \gamma^\uparrow, \gamma^\downarrow}$.

(b) We only consider the case that $\phi' > \phi$, the case $\phi' < \phi$ is similar. Then $x' = x -\gamma^\uparrow (\phi' -\phi) y$
Define the function $\epsilon: [0, 1] \to (0, \infty)$ by $f(\epsilon) = V^{\gamma^\uparrow, \gamma^\downarrow}(t, (1-\epsilon) x + \epsilon x', y, (1-\epsilon) \phi + \epsilon \phi')$. Then $f(0) = \hat V^{\gamma^\uparrow, \gamma^\downarrow}(t,  x , y, \phi )$ and $f(1) = \hat V^{\gamma^\uparrow, \gamma^\downarrow}(t,  x' , y, \phi' )$. It suffices to show that $f'(\epsilon) \leq 0$ and $f'(\epsilon)= 0$ if $\frac{\phi y}{x+ \phi y} < \frac{\phi' y}{x+ \phi' y} = p_*(\gamma^\uparrow, \gamma^\downarrow)$. Dropping the arguments for simplicity, the claim follows from part (a) via
\begin{align*}
f'(\epsilon) &= \hat V^{\gamma^\uparrow, \gamma^\downarrow}_x (x' -x) + \hat V^{\gamma^\uparrow, \gamma^\downarrow}_\phi (\phi' - \phi) = -\hat V^{\gamma^\uparrow, \gamma^\downarrow}_x y \gamma^\uparrow (\phi' -\phi) + \hat V^{\gamma^\uparrow, \gamma^\downarrow}_\phi (\phi' - \phi) \\
&= \hat V^{\gamma^\uparrow, \gamma^\downarrow}_x y (\phi'- \phi) \left(\ol \kappa_{\gamma^\uparrow, \gamma^\downarrow}\Big(\ol q_{\gamma^\uparrow, \gamma^\downarrow}\Big(\frac{((1- \epsilon) \phi + \epsilon \phi')y}{(1-\epsilon) x + \epsilon x' + ((1- \epsilon) \phi + \epsilon \phi') y}\Big)\Big) - \gamma^\uparrow\right).
\end{align*}

\end{proof}

\begin{proof}[Proof of Theorem \ref{thm:cand solution 1}]
For better readability, we are going to drop the subscripts and arguments $\xi$ and $\gamma^\uparrow, \gamma^\downarrow$. For future reference define
\begin{eqnarray}
a_Q(q) & = &  \frac{2S}{\sigma(1-S)} (\ell(q) - n(q)) = \sigma \frac{\ell(q)-n_\xi(q)}{D} \label{eq:defaQ}; \\
b_Q(q) & = & q \left( n(q) - 2 \theta S \{\ell(q )-n(q) \} \right), \label{eq:defbQ}
\end{eqnarray}
so that $d \hat{Q}_t = a_Q(\hat{Q}_t) \dd B_t + b_Q(\hat{Q}_t) \dd t + \dd G^\uparrow_t - \dd G^\downarrow_t$.

(a) Existence and uniqueness of a solution to \eqref{eq:dQdef} follows from standard results on reflecting one-dimensional diffusions, see, for example, \cite[Theorem 4.1]{tanaka:79}. The additional claim is trivial on $\llbracket 0, \tau \llbracket$. For the argument on $\llbracket \tau, \infty \llbracket$, first note that $a_Q(1) = 0$ and $b_Q(1) = n_\xi(1) > 0$ by Proposition \ref{prop:n}(e) and the fact that $m(1) = \ell(1)$.

If $q=1$, then $\hat Q_\tau = 1$ on $\{\tau < \infty\}$ since $\hat Q \in [q_*, q^*]$. Moreover, it is not difficult to check that  $\dd G_t  = b_Q(1) \dd t$ if $t \geq \tau$, and the process $\hat Q$ remains constant on $\llbracket \tau, \infty \llbracket$.

By contrast, if  $q^* > 1$, then $\hat Q_t > 1$ on $\{\tau < t\}$ by Feller's boundary test for diffusions (see e.g.~\cite[Proposition 5.5.22 (a)]{karatzas1991brownian}) since $b_Q(q) > 0$ and $a_Q(q) = O(1 - q)$ in a neighbourhood of $1$ because $\ell'(1) - n_\xi'(1) = \ell'(1) - m'(1) \neq 0$ by Proposition \eqref{eq:dQdef}(e) and the fact that $q_M \neq 1$.

(b) Let $(\hPhi, \hC, \hX)$ be defined as in the statement of the proposition. First, we argue that we cannot have both $ \phi y \gamma^\uparrow (1 - q_*) < q_* x$ and $\frac{\phi y}{\gamma^\downarrow} (1 - q^*) > q^* x$. We argue this in the case $q_*,q^*>0$; the case $q_*,q^* < 0$ is similar, but easier.

If $q_*,q^*>0$ then it is sufficient to show that $\gamma^\uparrow \frac{(1 - q_*)}{q_*} > \frac{1}{\gamma^\downarrow} \frac{(1 - q^*)}{q^*}$. If $0<q_* \leq 1 < q^*$ then this is immediate, and if $0<q_*<q^* \leq 1$ it follows from the fact that $\frac{1-q}{q}$ is decreasing on $(0,1]$ and then $\frac{1}{\gamma^\downarrow} \frac{(1 - q^*)}{q^*} \leq \frac{(1 - q^*)}{q^*} < \frac{(1 - q_*)}{q_*} \leq \gamma^\uparrow \frac{(1 - q_*)}{q_*}$.
Finally, if $1 \leq q_* < q^*$ the result follows from $\xi < \frac{(q^*-1)}{q^*} \frac{q^*)}{(q_*-1)}$ as shown in the proof of Proposition~\ref{prop:n}.

Now consider $(\hPhi_0, \hX_0)$. If both $\frac{\phi y}{\gamma^\downarrow}(1-q^*) \leq q^* x$ and $q_* x \leq \phi y \gamma^\uparrow (1 - q_*)$ then $\hPhi_0=\phi$ and using $\kappa_{\gamma^\uparrow,\gamma^\downarrow}(q) = \frac{q(1-p)}{(1-q)p}$ and $p = \frac{\phi x}{x+\phi y}$ we find
$\hX_0 = \frac{1-\hQ_0}{\hQ_0} Y_0 \kappa_{\gamma^\uparrow,\gamma^\downarrow}(\hQ_0) \hPhi_0= \frac{(1-p)}{p} y \phi = x$. In particular, $(\hPhi_0, \hX_0)=(\phi,x)$ and $\hX_0$ trivially satisfies the initial condition (2nd unlabelled equation in Section 4).

If $\frac{\phi y}{\gamma^\downarrow}(1-q^*) > q^* x$ or if $q_* x > \phi y \gamma^\uparrow (1 - q_*)$ then we have to account for the initial re-balancing of the portfolio. We describe what happens in the latter case, the details for the former case are similar.

Suppose $\phi y \gamma^\uparrow(1-q_*) < q_* x$. In this case $\hPhi_0 = q_* \left( \phi + \frac{x}{y \gamma^\uparrow} \right) > \phi$ and $\hQ_0 = q_*$. Then from $\hX_0= \frac{1-q_*}{q_*} y \gamma^\uparrow q_* \left( \phi + \frac{x}{y \gamma^\uparrow} \right)$ we find
$\hX_0 = x - q_* \left( \hPhi_0 - \phi \right)$. In particular $\hX_0$ satisfies the initial condition \eqref{eq:frictionIC}.

Rewriting \eqref{eq:Xdef1} we find that $\hQ$ solves
\[ \hQ_t = \frac{Y_t \kappa_{\gamma^\uparrow,\gamma^\downarrow}(\hQ_t) \hPhi_t}{\hX_t + Y_t \kappa_{\gamma^\uparrow,\gamma^\downarrow}(\hQ_t) \hPhi_t} ,\]
and hence that, provided $\hX$ represents the cash wealth of the agent (and $\kappa_{\gamma^\uparrow,\gamma^\downarrow}(\hQ)$ is the relative shadow price), $\hQ$ solves \eqref{eq:Qselfconsistent} and has the interpretation as the shadow fraction of wealth.
Since $\kappa_{\gamma^\uparrow,\gamma^\downarrow}(q_*) = \gamma^\uparrow$ and $\kappa_{\gamma^\uparrow,\gamma^\downarrow}(q^*)= \frac{1}{\gamma^\downarrow}$, the claim that $\hX = X^{x, y, \phi, \hat \Phi, \hat C}$ will follow if we can show that
\begin{equation}
\label{eq:Xdef2} \dd \hX_t = \left( r \hat{X}_t - \hat{C}_t \right) \dd t - Y_t \kappa_{\gamma^\uparrow,\gamma^\downarrow}(\hQ_t) \dd \hPhi_t,
\end{equation}
so that $\hX$ is the wealth process associated with consumption $\hC$ and trading strategy $\hPhi$.

We look for an expression for $\dd \hX_t$ of the form $\dd \hX_t = a_t \dd B_t + b_t \dd t + c_t \dd G_t$ where $a_t = a_X(\hX_t,Y_t,\hQ_t, \hPhi_t)$ and similarly for $b_t$ and $c_t$.

In the sequel, we only consider the case $0 < q^*(\xi) < q^*_(\xi) < 1$ so that $\hat Q \in (0, 1)$. The other cases in Proposition \ref{prop:n} can be treated in a similar way, replacing $\log(1-  \hat Q)$ by $\log(\hat Q - 1)$ if $\hat Q > 1$ and replacing $\log \hat Q$ and $\log \hat \Phi$ by $\log(-\hat Q)$ and $\log(-\hat \Phi)$, respectively, if $\hat Q < 0$. Special care is needed if $1 \in [q^*(\xi) < q^*(\xi)]$. In that case, by the final part of (a), one has to argue separately on $\llbracket 0, \tau \llbracket$ and $\rrbracket \tau, \infty \llbracket$. We leave the details to the reader.

From \eqref{eq:Xdef1} we have $\ln \hX_t = \ln(1- \hQ_t) - \ln \hQ_t + \ln Y_t + \ln \kappa_{\gamma^\uparrow,\gamma^\downarrow}(\hQ_t) + \ln \hPhi_t$ and then
\begin{eqnarray*}
\frac{\dd \hX_t}{\hX_t} - \frac{1}{2} \frac{ \dd [\hX]_t }{\hX_t^2}
& = & - \frac{\dd \hQ_t}{(1 -\hQ_t)} - \frac{1}{2} \frac{ \dd [\hQ]_t }{(1-\hQ_t)^2}
 - \frac{\dd \hQ_t}{\hQ_t} + \frac{1}{2} \frac{ \dd [\hQ]_t }{\hQ_t^2} +\frac{\dd Y_t}{Y_t} - \frac{1}{2} \frac{ \dd [Y]_t }{Y_t^2} \\
 && \hspace{5mm} + \frac{\kappa_{\gamma^\uparrow,\gamma^\downarrow}'(\hQ_t)}{\kappa_{\gamma^\uparrow,\gamma^\downarrow}(\hQ_t)}\dd \hQ_t + \frac{1}{2} \left( \frac{\kappa_{\gamma^\uparrow,\gamma^\downarrow}''(\hQ)_t}{\kappa_{\gamma^\uparrow,\gamma^\downarrow}(\hQ_t)} - \left( \frac{\kappa_{\gamma^\uparrow,\gamma^\downarrow}'(\hQ_t)}{\kappa_{\gamma^\uparrow,\gamma^\downarrow}(\hQ_t)} \right)^2 \right) \dd [\hQ]_t + \frac{\dd \hPhi_t}{\hPhi_t}.
\end{eqnarray*}

Equating co-efficients, and using $\frac{\kappa_{\gamma^\uparrow,\gamma^\downarrow}'(q)}{\kappa_{\gamma^\uparrow,\gamma^\downarrow}(q)} - \frac{1}{q(1-q)} = -  \frac{D}{\ell(q)-n_\xi(q)}$ and the definition of $a_Q(q)$ we find
\[ \frac{a_X(x,y,q,\phi)}{x} = \left( a_Q(q) \left( \frac{\kappa_{\gamma^\uparrow,\gamma^\downarrow}'(q)}{\kappa_{\gamma^\uparrow,\gamma^\downarrow}(q)} - \frac{1}{q(1-q)} \right) + \sigma \right)
= 0 \]
In particular, $\hX$ has no quadratic variation term. Further, using $\frac{\dd \hPhi_t}{\hPhi_t} = \frac{\dd G^\uparrow_t}{\hQ_t} - \frac{\dd G^\downarrow_t}{\hQ_t} = \frac{\dd G}{\hQ_t}$ and the fact that $\kappa_{\gamma^\uparrow,\gamma^\downarrow}'=0$ at $q_*$ and $q^*$,
\[ \frac{c_X(x,y,q,\phi)}{x} =  \frac{\kappa_{\gamma^\uparrow,\gamma^\downarrow}'(q)}{\kappa_{\gamma^\uparrow,\gamma^\downarrow}(q)} - \frac{1}{q(1-q)}  + \frac{1}{q} = \frac{-1}{1-q}. \]
Then $c_t \dd G_t = - \frac{\hX_t}{1- \hQ_t}\dd G_t = - \frac{\hX_t \hQ_t}{1- \hQ_t} \frac{\dd \hPhi_t}{\hPhi_t} = - Y_t \kappa_{\gamma^\uparrow,\gamma^\downarrow}(\hQ_t) \dd \hPhi_t$.
Finally,
\begin{eqnarray*} \frac{b_X(x,y,q,\phi)}{x} & = & \left( \frac{\kappa_{\gamma^\uparrow,\gamma^\downarrow}'(q)}{\kappa_{\gamma^\uparrow,\gamma^\downarrow}(q)} - \frac{1}{(1-q)}- \frac{1}{q} \right) b_Q(q) + \mu - \frac{\sigma^2}{2} \\
&& \hspace{20mm} + \frac{a_Q(q)^2}{2} \left( - \frac{1}{(1-q)^2} + \frac{1}{q^2} + \frac{\kappa_{\gamma^\uparrow,\gamma^\downarrow}''(q)}{\kappa_{\gamma^\uparrow,\gamma^\downarrow}(q)} - \left(\frac{\kappa_{\gamma^\uparrow,\gamma^\downarrow}'(q)}{\kappa_{\gamma^\uparrow,\gamma^\downarrow}(q)}\right)^2 \right).
\end{eqnarray*}
Then, differentiating $(1-q)q \frac{\kappa_{\gamma^\uparrow,\gamma^\downarrow}'(q)}{\kappa_{\gamma^\uparrow,\gamma^\downarrow}(q)} = \frac{m(q)-n_\xi(q)}{\ell(q)-n_\xi(q)}$ and using $\ell(q) = m(q)+Dq(1-q)$, $\ell'(q) - m'(q) = D(1-2q)$ and $\ell(q)m'(q) - m(q) \ell'(q) = D(q(1-q)m'(q) - (1-2q)m)$ yields
\begin{eqnarray*}
\lefteqn{
 \frac{\kappa_{\gamma^\uparrow,\gamma^\downarrow}''(q)}{\kappa_{\gamma^\uparrow,\gamma^\downarrow}(q)} - \left( \frac{\kappa_{\gamma^\uparrow,\gamma^\downarrow}'(q)}{\kappa_{\gamma^\uparrow,\gamma^\downarrow}(q)} \right)^2} \\
  & = &
 - \frac{(1-2q)}{q^2(1-q)^2} \frac{m(q)-n_\xi(q)}{\ell(q)-n_\xi(q)} \\
&& \hspace{10mm}+ \frac{(\ell(q)m'(q)-m(q)\ell'(q)) + n_\xi(q)(\ell'(q)-m'(q)) - n'(q)(\ell(q)-m(q))}{q(1-q)(\ell(q)-n_\xi(q))^2} \\
& = &  - \frac{(1-2q)}{q^2(1-q)^2}\frac{(m(q)-n_\xi(q))}{(\ell(q)-n_\xi(q))} \\
&& \hspace{1mm} + \frac{D}{q(1-q)(\ell(q)-n_\xi(q))^2}\left( q(1-q)m'(q) - (1-2q)(m(q)-n_\xi(q)) - q\frac{1-S}{S} n_\xi(q) \frac{m(q)-n_\xi(q)}{\ell(q)-n_\xi(q)} \right)
\end{eqnarray*}
It follows that, using $m'(q)= \frac{R(1-S)}{S} \sigma^2 q - \frac{(1-S)}{S}(\mu - r)$,
\begin{eqnarray*}
\lefteqn{\frac{b_X(x,y,q,\phi)}{x} } \\
 & = & - \frac{Dq}{\ell(q)-n_\xi(q)}\left( n_\xi(q) - 2 \theta S(\ell(q)-n_\xi(q)) \right)  + \mu - \frac{\sigma^2}{2}
+ \frac{\sigma^2 (\ell(q)-n_\xi(q))^2}{2D^2}  \frac{(1-q)^2 - q^2}{(1-q)^2q^2} \\
&& \hspace{10mm}
 - \frac{1}{2} \sigma^2 (\ell(q)-n_\xi(q))^2 \frac{(1-2q) (m(q)-n_\xi(q))}{(\ell(q)-m(q))^2(\ell(q)-n_\xi(q))} \\
&& \hspace{10mm} + \frac{\sigma^2}{2(\ell(q)-m(q))} \left( q(1-q) \left( \frac{R(1-S)}{S} \sigma^2 q - \frac{(1-S)}{S}(\mu - r) \right) - (1-2q)(m(q)-n_\xi(q)) \right. \\
&& \hspace{50mm} \left. - q \frac{1-S}{S} \frac{n_\xi(q)(m(q)-n_\xi(q))}{\ell(q)-n_\xi(q)} \right) \\
&=& \mu - \frac{\sigma^2}{2} + q \sigma^2(1-R) - \frac{Dqn_\xi(q)}{(\ell(q)-n_\xi(q))} + 
\left( R \sigma^2 q - (\mu - r) \right) + (1-2q) \frac{ \sigma^2 }{2} \frac{ (\ell(q)-n_\xi(q))^2}{(\ell(q)-m(q))^2} \\ \\
&& \hspace{10mm} - (1-2q) \frac{ \sigma^2 }{2} \frac{ (\ell(q)-n_\xi(q))(m(q)-n_\xi(q))}{(\ell(q)-m(q))^2} -(1-2q) \frac{ \sigma^2 }{2} \frac{(m(q)-n_\xi(q))}{(\ell(q)-m(q))} \\
&& \hspace{10mm} -  q \frac{ \sigma^2 }{2} \frac{(1-S)}{S} \frac{n_\xi(q)(m(q)-n_\xi(q)}{(\ell(q)-m(q))(\ell(q)-n_\xi(q))} \\
& = & r  - \frac{Dqn_\xi(q)}{\ell(q)-n_\xi(q)} \left( 1 + \frac{m(q)-n_\xi(q)}{\ell(q)-m(q)} \right) \\
&& \hspace{10mm} +(1-2q) \frac{\sigma^2}{2} \left[ -1 + \frac{(\ell(q)-n_\xi(q))^2}{(\ell(q)-m(q))^2} - \frac{ (\ell(q)-n_\xi(q))(m(q)-n_\xi(q))}{(\ell(q)-m(q))^2} - \frac{(m(q)-n_\xi(q))}{(\ell(q)-m(q))}  \right] \\
& = & r  - \frac{n_\xi(q)}{1-q}
\end{eqnarray*}
Then $b_t = r \hX_t - \frac{\hX_t}{1-\hQ_t} n(\hQ_t) = r \hX_t - \frac{Y_t \kappa_{\gamma^\uparrow,\gamma^\downarrow}(\hQ_t) \hPhi_t}{\hQ_t} n(\hQ_t) = r \hX_t - \hC_t$.
Combining the expressions for $a_t$, $b_t$ and $c_t$ we conclude that $\hX$ does indeed solve \eqref{eq:Xdef2}.

It only remains to show that $(\hPhi, \hC) \in \sA$. We consider the case where $q_*,q^*>0$. Then $\Phi_t > 0$ and it is sufficient to show that $\hX_t + \frac{\hPhi_t Y_t}{\gamma^\downarrow}>0$. Substituting for $\hX_t$ using \eqref{eq:Xdef1} we have $\hX_t + \frac{\hPhi_t Y_t}{\gamma^\downarrow} = \hPhi_t Y_t \left( \frac{1-\hQ_t}{\hQ_t} \kappa_{\gamma^\uparrow,\gamma^\downarrow}(\hQ_t) + \frac{1}{\gamma^\downarrow} \right)$. Set $\chi(q) = \frac{(1-q)}{q}\kappa_{\gamma^\uparrow,\gamma^\downarrow}(q) + \frac{1}{\gamma^\downarrow}$: we want to show that $\chi(q)>0$ on $(q_*,q^*)$. The result is immediate on $(q_* \wedge 1,1]$ so it is sufficient to consider $\chi$ on $[q_* \vee 1,q^*)$. But on this interval $\chi'(q) = - \frac{\kappa_{\gamma^\uparrow,\gamma^\downarrow}(q)}{q^2} \left( \frac{\ell(q)-m(q)}{\ell(q)-n_\xi(q)} \right) <0$ and the positivity of $\chi$  on $[q_* \vee 1,q^*)$ follows from the monotonicity of $\chi$ and the fact that $\chi(q^*) = \frac{1}{q^* \gamma^\downarrow}$.
\end{proof}

\begin{prop}
\label{app:prop:cand solution}
Let  $\gamma^\uparrow, \gamma^\downarrow \in [1, \infty)$ with $\xi:= \gamma^\uparrow \gamma^\downarrow  > 1$ be such that $\xi$ is in the well-posedness range of transaction costs. Let $(x, y, \phi) \in  \sS^\circ_{(\gamma^\uparrow, \gamma^\downarrow)}$.
Let  $(\hat \Phi, \hat C) \in \sA_{(\gamma^\uparrow, \gamma^\downarrow)}(x, y, \phi)$ be the candidate optimal investment-consumption pair, $\hat X = (\hat X_t)_{t \geq 0} := (X^{x, y, \phi, \hat \Phi, \hat C}_t)_{t \geq 0}$ the corresponding cash process, and $\hat Q =   (\hat Q_t)_{t \geq 0}$ the corresponding shadow fraction of wealth.
Then
\begin{equation}
\label{eq:app:prop:cand solution:hat V}
\hat V^{\gamma^\uparrow, \gamma^\downarrow} (t,\hat{X}_t,Y_t, \hat{\Phi}_t) = \hat{V}^{\gamma^\uparrow, \gamma^\downarrow}(0,\hat{X}_0,Y_0, \hat{\Phi}_0) {\mathcal E}(  \sigma (1-R) \hat Q \bullet B)_t e^{- \theta \int_0^t n_\xi(\hat Q_s) \dd s}, \quad t \geq 0.
\end{equation}
Moreover, the process $\hat M = (\hat M_t)_{t \geq 0}$ defined by
\begin{equation*}
  \hat M_t := \int_0^t g_{EZ}(s, \hat C_s, \hat V^{\gamma ^\uparrow, \gamma^\downarrow}(s, \hat X_s, Y_s, \hat \Phi_s)) \dd s + \hat V^{\gamma ^\uparrow, \gamma^\downarrow}(t, \hat X_t, Y_t, \hat \Phi_t).
  \end{equation*}
is a uniformly integrable martingale with $\hat M_\infty :=  \int_0^\infty g_{EZ}(s, \hat C_s, \hat V^{\gamma ^\uparrow, \gamma^\downarrow}(s, \hat X_s, Y_s, \hat \Phi_s)) \dd s$

\end{prop}

\begin{proof}
First, it follows from the definition of $\hat C$ in Theorem \ref{thm:cand solution 1}(c), \eqref{eq:prop:hatV:properties:HJB ineq:01} and \eqref{eq:Vx:ver} that
\begin{equation}
\label{eq:pf:app:prop:cand solution:gEZ}
g_{EZ}(t, \hat C_t, \hat V^{\gamma ^\uparrow, \gamma^\downarrow}(t, \hat X_t, Y_t, \hat \Phi_t)) = \theta n_\xi(\hat Q_t) \hat V^{\gamma^\uparrow, \gamma^\downarrow} (t,\hat{X}_t,Y_t, \hat{\Phi}_t), \quad t \geq 0.
\end{equation}
Second, It\^o's formula, Proposition \ref{prop:hatV:properties}(a) and (b) and the properties of $\hat \Phi$, $\hat C$ from Theorem \ref{thm:cand solution 1} give
\begin{equation}
\label{eq:pf:app:prop:cand solution:dM}
\dd \hat M_t = \sigma (1-R) \hat{V}^{\gamma ^\uparrow, \gamma^\downarrow}(t, \hat X_t, Y_t, \hat \Phi_t) \hat Q_t \dd B_t,
\end{equation}
Third, combining \eqref{eq:pf:app:prop:cand solution:dM} with \eqref{eq:pf:app:prop:cand solution:gEZ} yields
\begin{align}
\dd \hat V^{\gamma^\uparrow, \gamma^\downarrow} (t,\hat{X}_t,Y_t, \hat{\Phi}_t) &= \dd \hat M_t - g_{EZ}(t, \hat C_t, \hat V^{\gamma ^\uparrow, \gamma^\downarrow}(t, \hat X_t, Y_t, \hat \Phi_t)) \dd t \\
&= \hat{V}^{\gamma ^\uparrow, \gamma^\downarrow}(t, \hat X_t, Y_t, \hat \Phi_t) \left( \sigma (1-R) \hat Q_t \dd B_t -  \theta n_\xi(\hat Q_t) \dd t \right).
\end{align}
The formula for the stochastic exponential now gives \eqref{eq:app:prop:cand solution:hat V}.

Next, fix $\beta > 0$ and let $C > 0$ be a uniform bound for $\hat Q$ and $\nu > 0$ a lower bound for $n_\xi$. Since $\hat Q$ is uniformly bounded, Novikov's condition implies that ${\mathcal E}(\beta  \sigma (1-R) \hat Q \bullet B)$ is a true martingale. Hence, \eqref{eq:app:prop:cand solution:hat V} yields
\begin{align}
\EX{|\hat V^{\gamma^\uparrow, \gamma^\downarrow} (t,\hat{X}_t,Y_t, \hat{\Phi}_t)|^\beta} &= |\hat V^{\gamma^\uparrow, \gamma^\downarrow} (0,\hat{X}_0,Y_0, \hat{\Phi}_0)|^\beta \EX{e^{\beta(\beta - 1) \frac{\sigma^2(1-R)^2}{2} \int_0^t \hat Q_s^2 \dd s - \theta \int_0^T n_\xi(\hat Q_s) \dd s }} \\
&\leq  |\hat V^{\gamma^\uparrow, \gamma^\downarrow} (0,\hat{X}_0,Y_0, \hat{\Phi}_0)|^\beta  e^{\beta(\beta - 1) \frac{\sigma^2(1-R)^2}{2} C^2 - \theta \nu) t }, \label{eq:pf:app:prop:cand solution:V beta}
\end{align}
where $C$ is a (upper or lower as required) bound on $\hat{Q}$, and $\nu$ is a lower bound on $n_\xi$.
For $\beta = 2$, \eqref{eq:pf:app:prop:cand solution:V beta} shows that $\hat M$ is a square-integrable martingale. For $\beta = 1$, \eqref{eq:pf:app:prop:cand solution:V beta} shows that $\lim_{t \to \infty} \EX{|\hat V^{\gamma^\uparrow, \gamma^\downarrow} (t,\hat{X}_t,Y_t, \hat{\Phi}_t)|} = 0$ and together with \eqref{eq:pf:app:prop:cand solution:gEZ}  and the fact that $n_\xi$ is uniformly bounded, this implies that $\hat M_t$ converges in $L^1$ to $\hat M_\infty := \int_0^t g_{EZ}(s, \hat C_s, \hat V^{\gamma ^\uparrow, \gamma^\downarrow}(s, \hat X_s, Y_s, \hat \Phi_s)) \dd s$. Hence, $\hat M$ is a uniformly integrable martingale.
\end{proof}

\begin{defn}\label{def:O equivalence relation}
	Suppose that $X=(X_t)_{t\geq0}$ and $Y=(Y_t)_{t\geq0}$ are nonnegative progressive processes. We say that $X$ has the same order as $Y$ (and write $X \stackrel{\OO}{=} Y$, noting that $\stackrel{\OO}{=}$ is an equivalence relation) if there exist constants $k,K\in(0,\infty)$ such that
	\begin{equation}\label{eq:O(Y) condition}
	0\leq k Y \leq X \leq K Y.
	\end{equation}
	Denote the set of progressive processes with the same order as $X$ by $\OO(X)\subseteq\sP_+$.
\end{defn}

\begin{defn}\label{def:selforder}
	For $X \in \sP_+$, define $J^{X} = (J^{X}_t)_{t\geq0}$ by $J^{X}_t = \cEX[t]{\int_t^{\infty} X_s \dd s}$. Then the set $\bS \OO$ of \textit{self-order} consumption streams is given by	
	\begin{equation}
	\bS\OO \coloneqq \left\{X\in\sP_{++}:~ \E \left[ \int_0^\infty X_t \dd t\right] < \infty \mbox{ and } X \stackrel{\OO}{=} J^X\right\}.
	\end{equation}
Moreover, set
	$$\bS\OO_\nu  \coloneqq \left\{X\in\bS\OO:~  \big(e^{\nu t} X_t\big)_{t \geq 0} \in \bS\OO \right\}\quad \text{for }\nu\geq0, \qquad \text{and} \qquad \bS\OO_+ = \cup_{\nu > 0} \bS\OO_\nu.$$
\end{defn}

\begin{lemma}
\label{lem:SO}
Let $X, Y \in \sP_+$ with $X \stackrel{\OO}{=} Y$. Then for $\nu \geq 0$, $X \in \bS\OO_\nu$ if and only if $Y \in \bS\OO_\nu$. In particular, $X \in \bS\OO_+$ if and only if $Y \in \bS\OO_+$.
\end{lemma}

\begin{cor}
\label{cor:app:proper}
Let  $\gamma^\uparrow, \gamma^\downarrow \in [1, \infty)$ with $\xi:= \gamma^\uparrow \gamma^\downarrow  > 1$ be such that $\xi$ lies in the well-posedness range of transaction costs. Let $(x, y, \phi) \in  \sS^\circ_{(\gamma^\uparrow, \gamma^\downarrow)}$.
Let  $(\hat \Phi, \hat C) \in \sA_{(\gamma^\uparrow, \gamma^\downarrow)}(x, y, \phi)$ be the candidate optimal investment-consumption pair, $\hat X = (\hat X_t)_{t \geq 0} := (X^{x, y, \phi, \hat \Phi, \hat C}_t)_{t \geq 0}$ the corresponding cash process, and $\hat Q =   (\hat Q_t)_{t \geq 0}$ the corresponding shadow fraction of wealth. Then the process $(e^{-\delta \theta t} \hat C^{1-R}_t)_{t \geq 0}$ lies in $\bS\OO_+$.

Further, $\hat{C}$ is uniquely proper.
\end{cor}

\begin{proof}
First, note that
\begin{equation*}
e^{-\delta \theta t} \hat C^{1-R}_t = e^{-\delta \theta t} (\hat X_t + \hat \Phi_t \tilde Y_t)^{1-R} n_\xi(\hat Q_t)^{1-R}  = (1-R) \hat V^{\gamma^\uparrow, \gamma^\downarrow}(t,\hat{X}_t,Y_t, \hat{\Phi}_t) n_\xi(\hat Q_t)^{\theta}
\end{equation*}
Since $n_\xi$ is uniformly bounded from above and from below, by Lemma~ \ref{lem:SO} it remains to show that
 that $((1-R)\hat V (t,\hat{X}_t,Y_t, \hat{\Phi}_t))_{t \geq 0} \in \bS\OO_+$.

To this end, it suffices to show that $((1-R)\hat V (t,\hat{X}_t,Y_t, \hat{\Phi}_t))_{t \geq 0} \in \bS\OO_\nu$ for some $\nu > 0$. Since $n$ is uniformly bounded from above and below by a positive constants, there is  $D > \nu > 0$ such that $D \geq \theta n + \nu \geq \nu > 0$. Now using that ${\mathcal E}(  \sigma (1-R) \hat Q \bullet B)_t$ is a martingale by Novikov's condition and the fact that $\hat Q$ is uniformly bounded, we obtain
\begin{align}
J_t^{\exp(\nu \cdot) \hat V}&= \cEX[t]{\int_t^\infty \exp(\nu s) (1-R)\hat V (s,\hat{X}_s,Y_s, \hat{\Phi}_s) \dd s} \notag \\
&=(1-R)\hat V (t,\hat{X}_t,Y_t, \hat{\Phi}_t) \cEX[t]{\int_t^\infty \exp\left(- \int_t^s  \theta n(\hat Q_u) - \nu \dd u \right)  \dd s}  \label{eq:pf:app:prop:cand solution V in S0+}
\end{align}
The claim that $(e^{-\delta \theta t} \hat C^{1-R}_t)_{t \geq 0}$ lies in $\bS\OO_+$ now follows from the fact that the last factor on the right hand side of \eqref{eq:pf:app:prop:cand solution V in S0+} lies in $[\frac{1}{D}, \frac{1}{\nu}]$.

The fact that there is a unique utility process $V^{\hat{C}}$ associated to $\hat{C}$ now follows from Herdegen et al~\cite[Theorem 4.7]{herdegen:hobson:jerome:23C}.
\end{proof}

\begin{proof}[Proof of Theorem \ref{thm:cand solution 2}]
First, it follows from Proposition \ref{cor:app:proper} and Theorem \ref{thm:cand solution 1} that $\hat C$ is uniquely proper and has continuous paths.

Next, define the process $\hat M = (\hat M_t)_{t \geq 0}$ by
\begin{equation*}
  \hat M_t := \int_0^t g_{EZ}(s, \hat C_s, \hat V^{\gamma ^\uparrow, \gamma^\downarrow}(s, \hat X_s, Y_s, \hat \Phi_s)) \dd s + \hat V^{\gamma ^\uparrow, \gamma^\downarrow}(t, \hat X_t, Y_t, \hat \Phi_t).
  \end{equation*}
By Proposition~\ref{app:prop:cand solution}, $\hat M$ is a uniformly integrable martingale with
\[ \hat M_\infty = \int_0^\infty e^{-\delta s} \frac{(\hat C_s)^{1-S}}{1-S} ((1-R)\hat V^{\gamma ^\uparrow, \gamma^\downarrow}(s, \hat X_s, Y_s, \hat \Phi_s))^\rho \dd s. \]
This implies that
\begin{align}
\hat V^{\gamma ^\uparrow, \gamma^\downarrow}(t, \hat X_t, Y_t, \hat \Phi_t) &= \hat M_t -  \int_0^t g_{EZ}(s, \hat C_s, \hat V^{\gamma ^\uparrow, \gamma^\downarrow}(s, \hat X_s, Y_s, \hat \Phi_s)) \dd s \\
&= \cEX[t]{\hat M_\infty} - \int_0^t g_{EZ}(s, \hat C_s, \hat V^{\gamma ^\uparrow, \gamma^\downarrow}(s, \hat X_s, Y_s, \hat \Phi_s)) \dd s \\
&= \cEX[t]{ \int_t^\infty g_{EZ}(s, \hat C_s, \hat V^{\gamma ^\uparrow, \gamma^\downarrow}(s, \hat X_s, Y_s, \hat \Phi_s)) }.
\end{align}
Thus $(\hat V^{\gamma ^\uparrow, \gamma^\downarrow}(t, \hat X_t, Y_t, \hat \Phi_t)_t)_{t \geq 0}$ is a utility process for $\hat{C}$. Note that it is never vanishes and hence is a proper solution. Thus, it coincides with $V^{\hat C}$ by uniqueness of utility processes (for proper solutions).
\end{proof}

\subsection{Verification argument}

\begin{proof}[Proof of Proposition \ref{prop:ver:supersol}]
First, note that the dynamics of $X + \epsilon \hat X$ are given by
\begin{equation}
 \diff (X_t + \epsilon \hat X_t) = r(X_t + \epsilon \hat X_t) \dd t - (C_t + \epsilon \hat C_t) \dd t - Y_t\gamma^\uparrow\dd( \Phi_t^\uparrow + \epsilon \hat \Phi^\uparrow_t)  + \frac{Y_t}{\gamma^\downarrow}\dd( \Phi_t^\downarrow + \epsilon \hat \Phi_t^\downarrow).
\end{equation}
Next, fix arbitrary bounded stopping times $\tau_1\leq\tau_2$. Using It\^o's formula, Proposition \ref{prop:hatV:properties}, Corollary \ref{cor:hatV:ineq} and the fact that $\Phi_t^\uparrow  + \epsilon \hat \Phi_t^\uparrow - (\Phi_t  + \epsilon \hat \Phi_t)^\uparrow$ and $\Phi_t^\downarrow  + \epsilon \hat \Phi_t^\downarrow - (\Phi_t  + \epsilon \hat \Phi_t)^\downarrow$ are nondecreasing and identical for $t \geq 0$, we obtain
\begin{align}
&\hat V(X_{\tau_2} + \epsilon \hat X_{\tau_2}, Y_{\tau_2}, \Phi_{\tau_2} + \epsilon \hat \Phi_{\tau_2}) - \hat V(X_{\tau_1} + \epsilon \hat X_{\tau_1}, Y_{\tau_1}, \Phi_{\tau_1} + \epsilon \hat \Phi_{\tau_1}) \\
&\quad + \int_{\tau_1}^{\tau_2} g_{EZ}(C_u + \epsilon \hat C_u,  \hat V(X_u + \epsilon \hat X_u, Y_u, \Phi_u + \epsilon \hat \Phi_u)) \dd u \\
&= \int_{\tau_1}^{\tau_2} \left(\hat V_u \dd u + \hat V_x \dd (X_u + \epsilon \hat X_u)  + \hat V_y \dd Y_u + \hat V_\phi \dd (\Phi_u + \epsilon \hat \Phi_u) + \frac{1}{2} \hat V_{yy} \dd \langle Y \rangle_u + g_{EZ}(C_u + \epsilon \hat C_u, \hat V) \dd u \right)\\
&=  \int_{\tau_1}^{\tau_2}  \left(\hat V_t + r(X_u+ \epsilon \hat X_u) \hat V_x -(C_u + \epsilon \hat C_u) \hat V_x + \mu Y_u \hat V_y + \frac{1}{2} \sigma^2 Y_u^2 \hat V_{yy} + g_{EZ}(C_u + \epsilon \hat C_u, \hat V) \right) \dd u \\
&\quad +  \int_{\tau_1}^{\tau_2} \left(\hat V_\phi - \gamma^\uparrow Y_u \hat V_x \right)  \dd( \Phi_u  + \epsilon \hat \Phi_u)^\uparrow -  \int_{\tau_1}^{\tau_2}\gamma^\uparrow Y_u \hat V_x  \dd\left( \Phi_u^\uparrow  + \epsilon \hat \Phi_u^\uparrow - (\Phi_u  + \epsilon \hat \Phi_u)^\uparrow \right) \\
&\quad +  \int_{\tau_1}^{\tau_2} \left( \frac{Y_u}{\gamma^\downarrow} \hat V_x  - \hat V_\phi\right)  \dd( \Phi_u + \epsilon \hat \Phi_u)^\downarrow  + \int_{\tau_1}^{\tau_2} \frac{Y_u}{\gamma^\downarrow}\hat V_x \dd\left( \Phi_u^\downarrow  + \epsilon \hat \Phi_u^\downarrow - (\Phi_u  + \epsilon \hat \Phi_u)^\downarrow \right) \\
&\quad+ \int_{\tau_1}^{\tau_2} \sigma Y_u \hat V_{s} \dd W_u \\
&\leq - \int_{\tau_1}^{\tau_2} \hat V_x Y_u \left ( \gamma^\uparrow - \frac{1}{\gamma^\downarrow}\right)\dd\left( \Phi_u^\uparrow  + \epsilon \hat \Phi_u^\uparrow - (\Phi_u  + \epsilon \hat \Phi_u)^\uparrow \right) +\int_{\tau_1}^{\tau_2} \sigma Y_u \hat V_{y} \dd W_u \leq \int_{\tau_1}^{\tau_2} \sigma Y_u \hat V_{y} \dd W_u.
\end{align}
Define the local martingale $N=(N_t)_{t \geq 0}$ by
\begin{equation}
N_t = \int_0^t \sigma Y_u \hat V_{y}(u,X_u + \epsilon \hat X_u, Y_u, \Phi_u + \epsilon \hat \Phi_u) \dd W_u.
\end{equation}
Moreover, for $n\in\N$, define the stopping time $\zeta_n\coloneqq\inf\{s\geq \tau_1:  \langle {N} \rangle_s - \langle {N} \rangle_{\tau_1} \geq n\}$. Then
		\begin{align*}
		 &\hat V(X_{\tau_1} + \epsilon \hat X_{\tau_1}, Y_{\tau_1}, \Phi_{\tau_1} + \epsilon \hat \Phi_{\tau_1}) \geq \hat V(X_{\tau_2\wedge\zeta_n} + \epsilon \hat X_{\tau_2\wedge\zeta_n}, Y_{\tau_2\wedge\zeta_n}, \Phi_{\tau_2\wedge\zeta_n} + \epsilon \hat \Phi_{\tau_2\wedge\zeta_n}) \\
&\quad + \int_{\tau_1}^{\tau_2\wedge\zeta_n} e^{-\delta u} \frac{(C_u + \epsilon \hat C_u)^{1-S}}{1-S} ((1-R) \hat V(u,X_u + \epsilon \hat X_u, S_u, \Phi_u + \epsilon \hat \Phi_u))^\rho \dd u + N_{\tau_1} - N_{{\tau_2\wedge\zeta_n}}.
		\end{align*}
		Taking conditional expectations and using that $(N_{t \wedge \zeta_n} - N_{t \wedge \tau_1})_{t \geq 0}$ is an $L^2$-bounded martingale, the Optional Sampling Theorem gives
		\begin{equation}\label{eq:pf:prop:ver:supersol:local supersol}
			\hat{V}_{\tau_1} \geq \cEX[{\tau_1}]{\hat{V}_{\tau_2 \wedge \zeta_n} + \int_{\tau_1}^{{\tau_2}\wedge\zeta_n} e^{-\delta u}\frac{(C_u+\epsilon\hat C_u)^{1-S}}{1-S} ((1-R)\hat{V}_u)^\rho \dd s}.
		\end{equation}
where in a slight abuse of notation we have written $\hat V_u$ as shorthand for $\hat V(u,X_u + \epsilon \hat X_u, Y_u, \Phi_u + \epsilon \hat \Phi_u)$.
		Since $\hat{V}_u  \geq \hat{V}(u,\epsilon \hat X_u, Y_u, \epsilon \hat \Phi_u)$ by Corollary \ref{cor:hatV:ineq} and using $(\hat{V}(t,\epsilon \hat X_t, Y_t, \epsilon \hat \Phi_t))_{t \geq 0}$ is bounded below by zero if $R < 1$ and by a uniformly integrable martingale if $R > 1$ by Proposition \ref{app:prop:cand solution}, taking the liminf as $n \to \infty$ in \eqref{eq:pf:prop:ver:supersol:local supersol}, the conditional version of Fatou's Lemma (with a UI martingale lower bound if $R> 1$) and the conditional Monotone Convergence Theorem yield
		\begin{equation}\label{eq:hat V supersol}
\hat{V}_{\tau_1} \geq \cEX[{\tau_1}]{\hat{V}_{\tau_2 } + \int_{\tau_1}^{\tau_2} e^{-\delta u}\frac{(C_u+\epsilon \hat C_u)^{1-S}}{1-S} ((1-R)\hat{V}_u)^\rho \dd u}.
		\end{equation}
		Furthermore, $\liminf_{t\to\infty}\E[\hat{V}_{t}]\geq \lim_{t\to\infty}\E[\hat{V}(t,\epsilon \hat X_{t}, S_{t}, \epsilon \hat \Phi_{t}))] = 0$. Consequently, $\hat{V}(I,X+\epsilon \hat X, Y, \Phi + \epsilon \hat \Phi)$ is a supersolution for $C+\epsilon \hat C$.
\end{proof}

\begin{proof}[Proof of Theorem \ref{thm:verification}]
The argument is different for $\theta \leq 1$ and $\theta > 1$.

\textbf{Case $\theta \leq 1$.}
By Theorem \ref{thm:cand solution 2}, Proposition \ref{prop:hatV:properties}(c) and continuity of $\hat V^{\gamma^\uparrow, \gamma^\downarrow}$, it suffices to show that for each $\epsilon > 0$ and $C \in \sC_{\gamma^\uparrow, \gamma^\downarrow}(x, y, \phi)$,
\begin{equation}
\label{eq:pf:thm:verification:ineq}
V^C_0 \leq\hat V^{\gamma^\uparrow, \gamma^\downarrow}(0, x (1+\epsilon), y, \phi (1+ \epsilon) ).
\end{equation}
So fix $\epsilon > 0$ and $C \in \sC^{(\gamma^\uparrow, \gamma^\downarrow)}(x, y, \phi)$ with corresponding investment process $\Phi$ and cash process $X = X^{x, y, \phi, \Phi, C}$. Moreover, set $\hat X = \hat X^{x, y, \phi, \hat \Phi, \hat C}$
The cases $R < 1$ and $R >1$ require slightly different arguments.

First, assume $R<1$. Then $\hat V^{\gamma^\uparrow, \gamma^\downarrow}(I, X+\epsilon \hat X, Y, \Phi + \epsilon \hat \Phi)$ is a supersolution for the pair $(g_{EZ}, C+\epsilon \hat C)$ by Proposition \ref{prop:ver:supersol}. Since $C+\epsilon \hat C > C$ and $\frac{c^{1-S}}{1-S}$ is increasing in $c$, $\hat V_{\gamma^\uparrow, \gamma^\downarrow}(I, X+\epsilon \hat X, Y, \Phi + \epsilon \hat \Phi)$ is also supersolution for $C$ by \eqref{eq:hat V supersol}. Moreover, the (generalised) utility process $V^{C}$ for $C$ is the minimal supersolution for $C$ by \cite[Theorem 6.5]{herdegen:hobson:jerome:23B}. Consequently, $\hat V^{\gamma^\uparrow, \gamma^\downarrow}(I, X+\epsilon \hat X, Y,  \Phi + \epsilon \hat \Phi)\geq V^{C}$, and \eqref{eq:pf:thm:verification:ineq} follows from considering $t = 0$ and using Corollary \ref{cor:hatV:ineq}(b).

Next, assume $R>1$, which implies $S>1$ by the fact that $\theta >0$. Since $(C+\epsilon \hat C)^{1-S} \leq (\epsilon \hat C)^{1-S}$, $C+\epsilon \hat C$ is uniquely evaluable by Corollary \ref{cor:app:proper} and \cite[Corollary 6.3]{herdegen:hobson:jerome:23B}. Hence, there exists a uniformly integrable utility process $V^{C+\epsilon \hat C}$, which is in particular a subsolution for $C$. By the comparison theorem for sub- and supersolutions in the form of \cite[Theorem 5.8]{herdegen:hobson:jerome:23B}, it follows that $\hat V(I, X+\epsilon \hat X, Y, \Phi + \epsilon \hat \Phi)\geq V^{C+\epsilon C}$. Moreover, since $V^{C+\epsilon C} \geq V^{C}$ by \cite[Proposition 6.8]{herdegen:hobson:jerome:23B}, $\hat V^{\gamma^\uparrow, \gamma^\downarrow}(I, X+\epsilon \hat X, Y, \Phi + \epsilon \hat \Phi)\geq V^{C}$, and \eqref{eq:pf:thm:verification:ineq} follows from considering $t = 0$ and using Corollary \ref{cor:hatV:ineq}(b).

\medskip
\textbf{Case $\theta >  1$.}
By Theorem \ref{thm:cand solution 2} and Proposition \ref{prop:hatV:properties}(c), it suffices to show that for each $C \in \sC^{\gamma^\uparrow, \gamma^\downarrow}(x, y, \phi) \cap \UU\PP^*$,
\begin{equation}
\label{eq:pf:thm:verification:theta>1:ineq}
V^C_0 \leq \hat V^{\gamma^\uparrow, \gamma^\downarrow}(0,x, y, \phi),
\end{equation}
So fix $C \in \sC^{(\gamma^\uparrow, \gamma^\downarrow)}(x, y, \phi)\cap \UU\PP^*$ with corresponding investment process $\Phi$ and cash process $X = X^{x, y, \phi, \Phi, C}$. Moreover, set $\hat X = \hat X^{x, y, \phi, \hat \Phi, \hat C}$.
As in the case $\theta\leq 1$, the cases $R < 1$ and $R >1$ require slightly different arguments.

First, assume $R>1$. For each $\epsilon > 0$, $\hat V^{\gamma^\uparrow, \gamma^\downarrow}(I, X+\epsilon \hat X, Y, \Phi + \epsilon \hat \Phi)$ is a supersolution for the pair $(g_{EZ}, C+\epsilon \hat C)$ by Proposition \ref{prop:ver:supersol}.
Since $e^{-\delta \theta t} (C_t+ \epsilon \hat C_t)^{1-R} \leq e^{-\delta \theta t} (\epsilon \hat C)^{1-R}$ for all $t \geq 0$, it follows from Corollary \ref{cor:app:proper} and \cite[Proposition E.3]{herdegen:hobson:jerome:23C} that there exists a unique extremal solution $V^{C+\epsilon \hat C}$ for $(C + \epsilon \hat C)$ and $V^{C+\epsilon \hat C}$ is increasing in $\epsilon$. It is the minimal supersolution by \cite[Corollary E.4]{herdegen:hobson:jerome:23C}. Hence, by minimality, $V^{C+\epsilon \hat C} \leq \hat V^{\gamma^\uparrow, \gamma^\downarrow}(I, X+\epsilon \hat X, Y, \Phi + \epsilon \hat \Phi) <0$. In particular, $V^{C+\epsilon \hat C}_0\leq\hat V^{\gamma^\uparrow, \gamma^\downarrow}(0, X_0+\epsilon \hat X_0, y, \Phi_0 + \epsilon \hat \Phi_0) $, and letting $\epsilon \downarrow 0$ and using Proposition \ref{prop:hatV:properties}(c), we obtain
\begin{equation}
\label{eq:pf:thm:verification:theta>1:R>1}
\lim_{\epsilon\to0}V^{C+\epsilon \hat C}_0 \leq  V^{\gamma^\uparrow, \gamma^\downarrow}(0,x, y, \phi).
\end{equation}
	Let $V^*=\lim_{\epsilon\to0}V^{C+\epsilon \hat C}$. Then  $V^*_0\leq \hat V^{\gamma^\uparrow, \gamma^\downarrow}(0, X_0+\epsilon \hat X_0, y, \Phi_0 + \epsilon \hat \Phi_0) $. Consequently, since $g_{EZ}$ is increasing in its last two arguments, and $C+\epsilon \hat C$ and $V^{C+\epsilon \hat C}$ are increasing in $\epsilon$, $g_{EZ}(\cdot,C+\epsilon \hat C,V^{C +\epsilon\hat C})$ is increasing in~$\epsilon$. Applying the Monotone Convergence Theorem for conditional expectations yields
	\begin{equation*}
	V^*_t = \lim_{\epsilon\to0}V^{C+\epsilon \hat C}_t = \lim_{\epsilon\to0}\cEX[t]{\int_t^\infty g_{EZ}(s,C_s+\epsilon \hat C_s,V^{C+\epsilon \hat C}_s) \dd s} = \cEX[t]{\int_t^\infty  g_{EZ}(s,C_s,V^*_s) \dd s}.
	\end{equation*}
	Therefore, $V^*$ is a solution associated to $C$. Moreover, for each $t \geq 0$, $V^*_t <0$ $\as{\P}$ by the fact that $V^{C + \epsilon \hat C}_t < 0$ $\as{\P}$ for each $\epsilon$ and  $V^{C + \epsilon \hat C}$ is increasing in $\epsilon$. Thus, $V^*$ is proper. It therefore agrees with the unique proper solution $V^C$ and via \eqref{eq:pf:thm:verification:theta>1:R>1} we obtain \eqref{eq:pf:thm:verification:theta>1:ineq}.

Next, assume $R<1$. For $\nu > 0$ and	let $C^\nu_t \coloneqq C_t\wedge \nu^{-1} e^{- t} \hat C_t$. Then $\hat V^{\gamma^\uparrow, \gamma^\downarrow}(I, X+\epsilon \hat X, Y, \Phi + \epsilon \hat \Phi)$ is a supersolution for $C+\epsilon \hat C$ by Proposition \ref{prop:ver:supersol}, and a fortiori a supersolution for $C^\nu$. 

Moreover, there exists an extremal solution $V^{C^\nu}$ associated to $C^\nu$  which is decreasing in $\nu$ by \cite[Proposition E.3]{herdegen:hobson:jerome:23C}. Furthermore, since $C$ is right-continuous, $C^\nu$ is right-continuous. The extremal solution $V^{C^\nu}$ associated to $C^\epsilon$ is therefore proper by \cite[Theorem 4.5 and Remark 5.7]{herdegen:hobson:jerome:23C}.

Furthermore, using that $(C^\nu_t)^{1-R}\leq \nu^{R-1} e^{(R-1)t} \hat{C_t}^{1-R}$, we obtain
	\begin{align}
	(1-R)V^{\gamma^\uparrow, \gamma^\downarrow}(t, X_t+\epsilon \hat X_t, Y, \Phi_t + \epsilon \hat \Phi_t) &\geq e^{-\delta\theta t} (\nu \hat C_t)^{1-R} n_\xi(\hat Q_t)^{-\theta} \\
	&\geq \nu^{2(1-R)} n_\xi(Q_t)^{-\theta}  e^{(1-R-\delta\theta)t} (C^\nu_t)^{1-R} \\
&\geq \nu^{2(1-R)} D^{-\theta}  e^{(1-R-\delta\theta)t} (C^\nu_t)^{1-R},
	\end{align}
where $D := m(q_*(\xi))$ is an upper bound for $n_\xi$.

Now using that
	$V^{C^\nu}$ is uniformly integrable by \cite[Remark 2.2]{herdegen:hobson:jerome:23C}, the comparison result \cite[Corollary B.6]{herdegen:hobson:jerome:23C}, gives  $\hat V^{\gamma^\uparrow, \gamma^\downarrow}(I, X+\epsilon \hat X, Y, \Phi + \epsilon \hat \Phi)\geq V^{C^\nu}$. In particular, $V^{\gamma^\uparrow, \gamma^\downarrow}(0, X_0+\epsilon \hat X_0, y, \Phi_0 + \epsilon \hat \Phi_0)\geq V^{C^\nu}_0$, and letting $\epsilon, \nu \downarrow 0$ and using Proposition \ref{cor:hatV:ineq}(b), we obtain
\begin{equation}
\label{eq:pf:thm:verification:theta>1:R<1}
\lim_{\epsilon\to0}V^{C^\nu}_0 \leq  V^{\gamma^\uparrow, \gamma^\downarrow}(0,x, y, \phi).
\end{equation}

Finally, set $V^*=\lim_{\nu \to 0}V^{C^\nu}$. Using that $\frac{c^{1-S}}{1-S} ((1-R)v)^\rho$ is increasing both in $c$ and $v$ (recall $\rho>0$), monotone convergence gives, for $t \geq 0$
	\begin{align}
	V^*_t &= \lim_{n\to\infty}\cEX[t]{\int_t^\infty e^{-\delta s} \frac{(C^\nu_s)^{1-S}}{1-S} ((1-R)V^{C^\nu}_s)^\rho \dd s} =\cEX[t]{\int_t^\infty e^{-\delta s} \frac{C_s^{1-S}}{1-S} ((1-R)V^*_s)^\rho  \dd s}.
	\end{align}
	Hence, $V^*$ is a solution for $C$. It is a proper solution since for each $t \geq 0$, $V^*_t > 0$ if $V^{C^\nu}_t>0$ for some $\nu$ and $V^{C^\nu}_t>0$ on $\{J^{(C^\nu)^{1-R}} > 0\} = \{J^{C^{1-R}} > 0\}$ up to null sets by the fact that each $V^{C^\nu}$ is proper and $\hat C$ is strictly positive. Therefore, $V^*$ must agree with the unique proper solution $V^C$ associated to $C$, and via \eqref{eq:pf:thm:verification:theta>1:R<1} we obtain \eqref{eq:pf:thm:verification:theta>1:ineq}.
\end{proof}

\subsection{The ill posed case}

\begin{proof}[Proof of Proposition \ref{prop:ill posed:theta<1}]
Define $\ol \xi$ as in \eqref{eq:ol xi}.

(a) We only provide an argument for the case $S \leq R < 1$; the argument for the case $R < S < 1$ is very similar and hence omitted. We distinguish three cases.

\medskip{}
{\bf Case 1: $m(0)\leq 0$.} Without loss of generality we may assume that $x>0$ and $\phi=0$, since otherwise one can liquidate the entire risky position at $t=0$ and end up with a strictly positive pure-cash position, which lies in the interior of the solvency region. Consider a sequence of investment-consumption strategies $(\Phi^n,C^{n})_{n \in \NN}$, where $\Phi^n \equiv 0$ and $C^n=\eta_n Z^n$ where $(\eta_n)_{n \in \NN}$ is a decreasing sequence of positive real numbers converging to $-\frac{S}{1-S}m(0) \geq 0$ and $Z^n := Z^{x, \Phi^n, C^n}$; see \eqref{eqn:frictionless:wealth process}. It is easy to verify that $C^n\in\mathscr{C}_{\gamma^\uparrow, \gamma^\downarrow}(x, \phi, y)$ for each $n$. Then by the calculations in Section 3.1, the (unique) utility process $V^{C^n} = (V^{C^n}_t)_{t \geq 0}$ associated with $C^n$ satisfies $$V^{C^n}_0 =\left(\frac{\eta_n^{1-S}}{H(0,\eta_n)}\right)^{\theta}\frac{x^{1-R}}{1-R} = \left(\frac{ \eta_n^{1-S}}{S(m(0)+\frac{1-S}{S}\eta_n)}\right)^{\theta}\frac{x^{1-R}}{1-R}  ,$$ where $H(\cdot,\cdot)$ is defined in \eqref{eq:H_nu}. Now letting $n \to \infty$, the claim follows.

\medskip{}
 {\bf Case 2: $m(0) > 0$ and $m(1)\leq 0$.} Without loss of generality we may assume that $x=0$ and $\phi>0$, since otherwise at $t=0$ one can invest  all the current wealth into the risky asset and end up with a pure-stock position, which lies in the interior of the solvency region. Consider a sequence of investment-consumption strategies $(\Phi^n,C^{n})_{n \in \NN}$, where $\Phi^n = \frac{Z^n}{Y}$ and $C^n=\frac{\eta_n}{\gamma^\downarrow} Z^n$ where $(\eta_n)_{n \in \NN}$ is a decreasing sequence of positive real numbers converging to $-\frac{S}{1-S}m(1) \geq 0$ and $Z^n := Z^{x, \Phi^n, C^n}$;  see \eqref{eqn:frictionless:wealth process}. It is easy to verify that $C^n\in\mathscr{C}_{\gamma^\downarrow, \gamma^\downarrow}(x, \phi, y)$ for each $n$. Then by similar calculations as in Section 3.1, the (unique) utility process $V^{C^n} = (V^{C^n}_t)_{t \geq 0}$ associated with $C^n$ satisfies $$V^{C^n}_0 =\left(\frac{ \eta_n^{1-S}}{H(1,\eta_n)}\right)^{\theta}\frac{\left(\frac{x}{\gamma^\downarrow}\right)^{1-R}}{1-R} = \left(\frac{ \eta_n^{1-S}}{S(m(1)+\frac{1-S}{S}\eta_n)}\right)^{\theta}\frac{\left(\frac{x}{\gamma^\downarrow}\right)^{1-R}}{1-R}  ,$$ where $H(\cdot,\cdot)$ is defined in \eqref{eq:H_nu}. Now letting $n \to \infty$, the claim follows.

\medskip{}

{\bf Case 3: $m(0) > 0$, $m(1) > 0$, $m(q_M)<0$ and $\xi\leq \bar{\xi}$}. Consider a decreasing sequence $(\xi_n)_{n \in \NN}$ with $\xi_n\downarrow \bar{\xi}$. For each $n \in \NN$, set $\gamma^\uparrow_n:=\xi_n/\gamma^\downarrow$ and denote by $(\Phi^n,C^n)$ the optimal strategy from Theorem \ref{thm:cand solution 1} corresponding to the (well-posed) problem with transaction cost $\gamma^\uparrow_n$ on purchases and $\gamma^\downarrow$ on sales. Note that each $C^n$ is $(\gamma^\uparrow,\gamma^\downarrow)$-attainable as $\mathscr{C}_{\gamma^\uparrow_n,\gamma^\downarrow}(x,\phi,y)\subseteq \mathscr{C}_{\gamma^\uparrow,\gamma^\downarrow}(x,\phi,y)$.
Since $\underline{x}:= x + \phi^+ y \frac{1}{\gamma^\downarrow} - \phi^- y \gamma^{\uparrow} > 0$ we have that for large enough $n$,
$x + \phi^+ y \frac{1}{\gamma^\downarrow} - \phi^- y \gamma_n^{\uparrow} > \frac{\underline{x}}{2}>0$ and
then using 
Theorem \ref{thm:cand solution 2} together with the lower bound $x+ \phi y \ol \kappa_{\gamma^\uparrow_n, \gamma^\downarrow}\left(\ol q_{\gamma^\uparrow_n, \gamma^\downarrow}\left( \frac{\phi y}{x+\phi y} \right)\right) \geq x + \phi^+ y \frac{1}{\gamma^\downarrow} - \phi^- y \gamma^{\uparrow}_n  \geq \frac{\ul x}{2}$ gives
\begin{eqnarray*}
	V^{C^n}_0 &= & \frac{\Big(x+\phi y \ol \kappa_{\gamma^\uparrow_n, \gamma^\downarrow}\left(\ol q_{\gamma^\uparrow_n, \gamma^\downarrow}\left( \frac{\phi y}{x+\phi y} \right)\right)\Big)^{1-R}}{1-R}  \ol n_{\xi_n}\left(\ol q_{\gamma^\uparrow_n, \gamma^\downarrow}\left( \frac{\phi y}{x+\phi y} \right)\right)^{-\theta S}
\geq \frac{ {\ul x}^{1-R}}{(1-R)2^{1-R}} \left(\Vert \ol n_{\xi_n} \Vert_\infty\right)^{-\theta S}
\end{eqnarray*}
Using that $\lim_{n \to \infty} \Vert \ol n_{\xi_n} \Vert_\infty = 0$ by Corollary \ref{cor:n}, the result follows.

(b) We only provide an argument for the case $S \geq R < 1$; the argument for the case $R >S > 1$ is very similar and hence omitted. We distinguish two cases.

 {\bf Case 1: $m(q_M)\leq 0$}. Since $\mathscr{C}_{\gamma^\uparrow,\gamma^\downarrow}(x,\phi,y)\subseteq \mathscr{C}_{0}(x + \phi,y)$, the claim follows from the fact that for each $C \in \mathscr{C}_{0}(x + \phi,y)$, we have $V^C_0=-\infty$ by \cite[Corollary 8.2]{herdegen:hobson:jerome:23B}.

\medskip{}

{\bf Case 2: $m(q_M)>0$ and $\xi \geq\bar{\xi}$.}
Note that the finiteness of transaction costs implies $\bar{\xi}<\infty$ while $m(q_M)>0$ implies $\bar{\xi}>1$. Let $(\xi_n)_{n \in \NN}$  be an increasing sequence with $\xi_n \uparrow \bar{\xi}$. For each $n \in \NN$ set $\gamma^\uparrow_n =\xi_n/\gamma^{\downarrow}$. By Theorem \ref{thm:verification}, the problem with transaction costs $(\gamma^\uparrow_n,\gamma^\downarrow)$ is well-posed for every $n \in \NN$ with value function given by \eqref{def:hat V}. Fix $C \in \mathscr{C}_{\gamma^\uparrow,\gamma^\downarrow}(x,\phi,y)$.
We have the upper bound $x+ \phi y \ol \kappa_{\gamma^\uparrow_n, \gamma^\downarrow}\left(\ol q_{\gamma^\uparrow_n, \gamma^\downarrow}\left( \frac{\phi y}{x+ \phi y} \right)\right) \leq  \xi x + \phi^+ y \gamma^\uparrow_n - \phi^- y \frac{1}{\gamma^\downarrow}  \leq \xi x + \phi^+ y \gamma^\uparrow - \phi^- y \frac{1}{\gamma^\downarrow} = \xi \ul x$
and then, since $\mathscr{C}_{\gamma^\uparrow,\gamma^\downarrow}(x,\phi,y) \subset \mathscr{C}_{\gamma^\uparrow_n,\gamma^\downarrow}(x,\phi,y)$ by Remark~\ref{rem:inclusion}, Theorem \ref{thm:cand solution 2} gives
\begin{eqnarray*}
	V^{C}_0 &\leq &\sup_{C\in\mathscr{C}_{\gamma^\uparrow_n, \gamma^\downarrow}(x, \phi, y)}V^C_0
	=\frac{\Big(x + \phi y \ol \kappa_{\gamma^\uparrow_n, \gamma^\downarrow}\left(\ol q_{\gamma^\uparrow_n, \gamma^\downarrow}\left( \frac{\phi y}{x+\phi y} \right)\right)\Big)^{1-R}}{1-R}  \ol n_{\xi_i}\left(\ol q_{\gamma^\uparrow_n, \gamma^\downarrow}\left( \frac{\phi y}{x+ \phi y} \right)\right)^{-\theta S} \\
&\leq & \frac{ \left( \xi \ul x \right)^{1-R}}{1-R} \left(\Vert \ol n_{\xi_n} \Vert_\infty\right)^{-\theta S}.
\end{eqnarray*}
By Corollary \ref{cor:n} we have that $\lim_{n \to \infty} \Vert \ol n_{\xi_n} \Vert_\infty = 0$ and the result follows.
\end{proof}


\section{Proofs for small transaction costs}
\label{app:smalltc}

\begin{proof}[Proof of Proposition \ref{prop:smalltcQ}]
We consider the case that $q_M = \frac{\lambda}{\sigma R} \in (0, \infty)\setminus \{1 \}$; the case $q_M < 0$ can be treated in a similar way. The case $q_M=1$ is broadly similar, but the leading terms are of different order. See Hobson et al~\cite{hobson:tse:zhu:19A} for a derivation of the expansion for the case $q_M=1$ in the additive setting.

By Definition \ref{def:well posed}, the assumptions that the problem \eqref{eq:frictional problem} is well-posed for arbitrarily small transaction cost implies that $\hat{m} = m(q_M) \geq 0$ if $R<1$ and $m(q_M)>0$ if $R>1$.

Recall $n^+_{(z)}$ solves $n'=O(q,n)$ subject to $n^+_{(z)}(z) = m(z)$. (Here $z \in (0,q_M) \cap \{ q:m(q)>0 \}$.) Then $\zeta^+= \zeta^+(z) > q_M$ is such that $n^+_{(z)}(\zeta^+) = m(\zeta^+)$ and moreover $(1-R) n^{+}_{(z)} > (1-R)m$ on $(z,\zeta^+) \setminus \{ 1 \}$. Further, $\Sigma^+(z) = \exp \left( \int_z^{\zeta^+(z)}\frac{1}{q(1-q)} \frac{ n^+_{(z)}(q) - m(q)}{\ell(q) - n^+_{(z)}(q)} dq \right)$.

For $u \in (0,q_M) \cap \{ q:m(q)>0 \}$ introduce $N_u : [-u , \infty) \to \R^+$ defined by  $N_u(v) =n^+_{(\hq-u)}(\hq+v)- \hatm$. Introduce also $M$ and $L$ given by $M(v) = m(\hq+v)- \hatm$ and $L(v) = \ell(\hq+v)- \hatm$. Note that $M(v) = \frac{R(1-S)}{2S} \sigma^2 v^2 =RDv^2$ where $D = \frac{(1-S)}{2S} \sigma^2$.

Define $G_u:[-1,\infty)$ by $G_u(v) = \frac{1}{D}(N_u(uv)-M(uv))$. (The normalisation $\frac{1}{D}$ is largely cosmetic but will help highlight the point at which the elasticity of intertemporal consumption $S$ enters the expansion.) Then $G_u(-1)=0$. Further, for $v \geq -1$,
\begin{eqnarray}
G_u'(v) & = & \frac{1}{D} u O(\hq + uv,DG_u(v) + M(uv) + \hatm) - \frac{1}{D} u M'(uv) \nonumber \\
&=& - \frac{2}{\sigma^2} u \frac{D(G_u(v) + R u^2 v^2) + \hatm}{1-\hq - vu} \frac{G_u(v)}{( \hq(1-\hq)+ (1-2\hq)uv - u^2 v^2 - G_u(v))}  - 2 R v u^2 \label{eq:defGdash}
\end{eqnarray}

We look for a solution to \eqref{eq:defGdash} of the form $G_u(v) = \sum_{k \geq 0} u^k g_k(v)$ defined for $v \geq -1$. We find that the zeroth order term $g_0$ satisfies
\[ g'_0(v) = 0   \hspace{20mm}  g_0(-1)= 0  \]
and conclude that $g_0(v)=0$. Similarly, using that $g_0 \equiv 0$ we find
\[ g'_1(v) = 0   \hspace{20mm}  g_1(-1)= 0  \]
so that again $g_1  \equiv 0$.
The next three terms (clearly, the expansion method can be continued to arbitrary order) solve

\begin{tabular}{lcl}
$g'_2(v) = -2Rv$   & \hspace{10mm} &  $g_2(-1)= 0$ \\
$g'_3(v) = - \frac{2\hatm}{\sigma^2 \hq(1-\hq)^2} g_2(v)$  &  &  $g_3(-1)= 0$ \\
$g'_4(v) = - \frac{2\hatm}{\sigma^2\hq(1-\hq)^2} \left( g_3(v) + v g_2(v) \frac{(3 \hq - 1)}{\hq(1-\hq)} \right)$ & & $g_4(-1)=0$
\end{tabular}

\noindent{and} we find
\begin{eqnarray*}
g_2(v) & = & R(1-v^2) \\
g_3(v) & = & -A (1+v)^2(2-v) \\
g_4(v) & = & B_1 (1+v)^3(3-v) + B_2(1-v^2)^2
\end{eqnarray*}
where
\begin{eqnarray*}
A & = &  \frac{2R}{3\sigma^2} \frac{\hatm}{\hq(1-\hq)^2} \\
B_1 & = & \frac{R}{3 \sigma^4 } \frac{ \hatm^2}{\hq^2(1- \hq)^4} \\
B_2 & = & \frac{R}{2 \sigma^2}\frac{\hatm(3\hq-1)}{\hq^2(1-\hq)^3}
\end{eqnarray*}

Note that if we truncate the expansion at the second order term (i.e. take only the non-zero first term of the expansion the we find $G_u(v) \approx u^2(1-v^2)R$ so that $N_u(uv) \approx u^2RD$ which is independent of $v$. It follows that for $z$ close to $\hq$, $n^+_{(z)}(q) = m(z)$ and then $\zeta^+(z)= 2 \hq -z$, and to leading order (in the shadow fraction of wealth co-ordinates) the no-transaction region is symmetric about the Merton ratio.

The expansion for $G_u$ takes the form
\[ G_u(v) = R(1-v^2) u^2 - A(1+v)^2(2-v) u^3 + \left[ B_1 (1+v)^3(3-v) + B_2(1-v^2)^2 \right] u^4 + \sO(u^5),   \]
and if $\zeta_{(u)} = \inf \{ v > -1 : G_u(v) = 0 \}$ then $\zeta_{(u)}$ is approximately 1. We look for a root of the form $\zeta_{(u)} = 1 - \zeta_1 u + \zeta_2 u^2 + \sO(u^3)$. Then
from $G_u(\zeta_{(u)})=0$ we find
\[ 0 = ( 2R \zeta_1 - 4A) u^3  +  (- 2R \zeta_2 - \zeta_1^2 R + 16 B_1) u^4 + O(u^5) \]
It follows that
\[ \zeta_1 =  \frac{2A}{R} = \frac{4\hatm}{3 \sigma^2\hq(1-\hq)^2}, \hspace{20mm}
\zeta_2 = \zeta_1^2. \]

We can now fix $u= \hq - q_*$ via the identity
\begin{eqnarray*}
\ln(\xi) & = & u \int_{-1}^{1 - \zeta_1 u + \zeta_1^2 u^2 + \sO(u^3)} dv \frac{1}{(\hq + uv)(1- \hq - uv)}\frac{G_u(v)}{(\hq+uv)(1 - \hq - vu) - G_u(v)} \\
& = & \int_{-1}^{1 - \zeta_1 u + \zeta_1^2 u^2 + \sO(u^3)} \frac{dv}{\hq^2(1-\hq)^2} u \left( 1 + \frac{uv}{\hq} \right)^{-1}
\left( 1 - \frac{uv}{1-\hq} \right)^{-1} G_u(v)  \\
&& \hspace{30mm} \times \left( \left(1 + \frac{uv}{\hq} \right) \left(1 - \frac{uv}{1-\hq} \right) - \frac{G_u(v)}{\hq(1-\hq)} \right)^{-1} \\
& = & \int_{-1}^{1 - \zeta_1 u + \zeta_1^2 u^2 + \sO(u^3)} \frac{dv}{\hq^2(1-\hq)^2} \\
&& \hspace{40mm} \times \left( g_2(v)u^3  + \left( g_3(v) + \frac{2(2 \hq - 1)}{\hq(1-\hq)} vg_2(v) \right)u^4 +  \Psi(v) u^5 + \sO(u^6) \right)
\end{eqnarray*}
where
\[ \Psi(v) = g_4(v) + \frac{2(2 \hq - 1)}{\hq(1-\hq)} vg_3(v) + \frac{1}{\hq(1-\hq)} g_2(v)^2 + \frac{3 -10\hq + 10 \hq^2}{\hq^2(1-\hq)^2} v^2 g_2(v) \]
Now $\int_{-1}^z (1-v^2)dv = \frac{1}{3}(z+1)^2(2-z)$ so that
\[ \int_{-1}^{1 - \zeta_1 u + \zeta_1^2 u^2 + \sO(u^3)}g_2(v)dv =\frac{R}{3} (4 - 3 \zeta_1^2 u^2) + \sO(u^3)  \]
Similarly, $\int_{-1}^z (1+v)^2(2-v)  dv = \frac{1}{4}(z+1)^3(3-z)$ and $\int_{-1}^z v(1-v^2)dv = -\frac{1}{4}(z^2-1)^2$ so that
\begin{eqnarray*}
\lefteqn{\int_{-1}^{1 - \zeta_1 u + \sO(u^2)}\left( g_3(v) + \frac{2(2 \hq-1)}{\hq(1-\hq)} v g_2(v)  \right)dv} \\
 & = &
-4 A (1- \zeta_1 u)  + \sO(u^2) \\
& = & - \frac{8R}{3\sigma^2} \frac{\hatm}{\hq(1-\hq)^2} + \left( \frac{8R}{3\sigma^2} \frac{\hatm}{\hq(1-\hq)^2} \right) \zeta_1 u + \sO(u^2)
\end{eqnarray*}
Finally, $\int_{-1}^1 (1+v)^3(3-v)  dv = \frac{48}{5}$, $\int_{-1}^1 (1-v^2)^2  dv = \frac{16}{15}$, $\int_{-1}^1 v(1+v)^2(2-v)  dv = \frac{8}{5}$, $\int_{-1}^1 v^2(1-v^2)  dv = \frac{4}{15}$ so that
\begin{eqnarray*}
 \int_{-1}^{1 + \sO(u)} \Psi(v) dv & = & \frac{16}{5} \frac{R}{\sigma^4} \frac{ \hatm^2}{\hq^2(1- \hq)^4}
+ \frac{8R}{15\sigma^2} \frac{\hatm(3\hq-1)}{\hq^2(1-\hq)^3} - \frac{32 R}{15 \sigma^2} \frac{\hatm(2\hq-1)}{\hq^2(1-\hq)^3} \\
&& \hspace{10mm}
        + \frac{16}{15} \frac{1}{\hq(1-\hq)} R^2 +\frac{4}{15} \frac{(3 -10\hq + 10 \hq^2)}{\hq^2(1-\hq)^2} R
+\sO(u)
\\
& = & \frac{16R}{5\sigma^4} \frac{ \hatm^2}{\hq^2(1- \hq)^4}
- \frac{8R}{15 \sigma^2}  \hatm \frac{(5\hq-3)}{\hq^2(1-\hq)^3}  \\
&& \hspace{10mm}
        + \frac{16}{15} \frac{1}{\hq(1-\hq)} R^2  +\frac{4}{15} \frac{(3 -10\hq + 10 \hq^2)}{\hq^2(1-\hq)^2} R
+\sO(u)
\end{eqnarray*}
Putting this all together, writing $\xi = 1+\epsilon$, and then using the expression for $\zeta_1$,
\begin{eqnarray*}
\lefteqn{\ln(1+\epsilon) \left( = \epsilon - \frac{\epsilon^2}{2} + \sO(\epsilon^3) =  \epsilon^\uparrow + \epsilon^\downarrow - \frac{(\epsilon^\uparrow)^2  - (\epsilon^\downarrow)^2}{2} + + \sO(\epsilon^3) \right) } \\
& = & \frac{4R}{3 \hq^2(1-\hq)^2}u^3 -  \frac{8R}{3\sigma^2} \frac{\hatm}{\hq^3(1-\hq)^4} u^4  \\
& & \hspace{5mm} + \frac{R}{\hq^2(1-\hq)^2} \left( - \zeta_1^2 + \frac{8}{3\sigma^2} \frac{\hatm \zeta_1}{\hq(1-\hq)^2} +\frac{4}{15} \frac{(3 -10\hq + 10 \hq^2)}{\hq^2(1-\hq)^2}  + \frac{16}{15} R \frac{1}{\hq(1-\hq)} \right. \\
&& \hspace{35mm} \left. \frac{16}{5\sigma^4}\frac{ \hatm^2}{\hq^2(1- \hq)^4} - \frac{8}{15\sigma^2} \hatm \frac{(5\hq-3)}{\hq^2(1-\hq)^3} \right) + \sO(u^6) \\
& =& \frac{4R}{3 \hq^2(1-\hq)^2}u^3 -  \frac{2R}{\hq^2(1-\hq)^2} \zeta_1 u^4  \\
& & \hspace{5mm} + \frac{R}{\hq^2(1-\hq)^2} \left( \frac{14}{5} \zeta_1^2 - \frac{2}{5} \frac{(5\hq-3)}{\hq(1-\hq)} \zeta_1 + \frac{4}{15} \frac{(3 -10\hq + 10 \hq^2)}{\hq^2(1-\hq)^2}  + \frac{16}{15} R \frac{1}{\hq(1-\hq)} \right)u^5 + \sO(u^6) \\
& =: & C_3 u^3 + C_4 u^4 + C_5 u^5 + \sO(u^6)
\end{eqnarray*}

Comparing the two sides of the equation we see that for the equation to balance we must have that $u = \hq - q_*(1+\epsilon) = \sO(\epsilon^{1/3})$.
We look for an expressing for $q_*$ of the form
\[ u  =  \Delta^{\frac{1}{3}} \epsilon^{\frac{1}{3}} + \Sigma \Delta^{\frac{2}{3}} \epsilon^{\frac{2}{3}} + \Psi \Delta \epsilon + \sO\left( \epsilon^{\frac{4}{3}} \right) \]
or equivalently
\[ q_*(1+\epsilon)  = \hq - \Delta^{\frac{1}{3}} \epsilon^{\frac{1}{3}} - \Sigma \Delta^{\frac{2}{3}} \epsilon^{\frac{2}{3}} - \Psi \Delta \epsilon + \sO\left( \epsilon^{\frac{4}{3}} \right) \]
Then
\[ \epsilon = C_3 \left( \Delta \epsilon + 3 \Sigma (\Delta \epsilon)^{4/3} + 3(\Sigma^2 + \Psi)(\Delta \epsilon)^{5/3} \right) + C_4 \left( (\Delta \epsilon)^{4/3} + 4 \Sigma (\Delta \epsilon)^{5/3} \right) + C_5 (\Delta \epsilon)^{5/3} + \sO(\epsilon^2) \]
The constants $\Sigma$, $\Delta$ and $\Psi$ can be fixed by the identities $1 = C_3 \Delta$, $0=C_4 + 3 \Sigma C_3$ and $0= 3(\Sigma^2 + \Psi)C_3 + 4 \Sigma C_4 + C_5$.
We find
\begin{eqnarray*}
\Delta & = & \frac{1}{C_3} = \frac{3 \hq^2 (1 - \hq)^2}{4R} \\
\Sigma & = & - \frac{C_4}{3C_3} = \frac{ 2 \hatm}{3 \sigma^2 \hq(1 - \hq)^2} = \frac{\zeta_1}{2}  \\
\Psi & = & - \frac{C_5 + 4 \Sigma C_4}{3C_3} - \Sigma^2 = 3 \Sigma^2 - \frac{C_5}{3C_3}\\
&=& 
\frac{1}{5} \Sigma^2 + \frac{1}{5} \frac{(5\hq-3)}{\hq(1-\hq)} \Sigma - \frac{(3 -10\hq + 10 \hq^2)+ 4R\hq(1-\hq)}{15\hq^2(1-\hq)^2}
\end{eqnarray*}
Further,
\begin{eqnarray*}
q^*(1+\epsilon) -\hq & = & u(\epsilon) \zeta_{(u(\epsilon))} = u(\epsilon)(1- \zeta_1 u(\epsilon) + \zeta_1^2 u(\epsilon)^2) + \sO (u(\epsilon)^4) \\
& = & \Delta^{\frac{1}{3}} \epsilon^{\frac{1}{3}} + \Sigma \Delta^{\frac{2}{3}} \epsilon^{\frac{2}{3}} + \Psi \Delta \epsilon - \zeta_1 \left( \Delta^{\frac{2}{3}} \epsilon^{\frac{2}{3}}  + 2\Sigma \Delta \epsilon \right)+ \zeta_1^2 \Delta \epsilon + \sO( \epsilon^{4/3})
\end{eqnarray*}
It follows that
\begin{equation*} q^*(1+\epsilon) = \hq  + \Delta^{\frac{1}{3}} \epsilon^{\frac{1}{3}} - \Sigma \Delta^{\frac{2}{3}} \epsilon^{\frac{2}{3}} + \Psi \Delta \epsilon + \sO\left( \epsilon^{\frac{4}{3}} \right) \qedhere
\end{equation*}
\end{proof}


\begin{proof}[Proof of Corollary \ref{cor:smalltcP}]
Setting $\epsilon := (1 +\epsilon^\uparrow)(1+\epsilon^\downarrow) -1$ and using \eqref{eq:def:p lower star} and \eqref{eq:def:p upper star}, we obtain
\begin{eqnarray*} p_*(1 + \epsilon^\uparrow,1 + \epsilon^\downarrow) & = & \frac{ q_*(1+\epsilon)}{1 + \epsilon^\uparrow(1-q_*(1+\epsilon))} = q_*(1+\epsilon) - \epsilon^\uparrow q_*(1+\epsilon)(1- q_*(1+\epsilon)) + \sO(\epsilon^2) \\
p^*(1 + \epsilon^\uparrow,1 + \epsilon^\downarrow) & = &
   \frac{(1 + \epsilon^\downarrow) q^*(1+\epsilon)}{1 + \epsilon^\downarrow q^*(1+\epsilon)} =  q^*(1+\epsilon) + \epsilon^\downarrow q^*(1+\epsilon)(1- q^*(1+\epsilon)) + \sO(\epsilon^2)
\end{eqnarray*}
Now the result follows from Proposition \ref{prop:smalltcC}.
\end{proof}

\begin{proof}[Proof of Proposition \ref{prop:smalltcC}]
Setting $\gamma^\uparrow : = 1  +\epsilon^\uparrow$, $\gamma^\downarrow : = 1  +\epsilon^\downarrow$, $\xi := \gamma^\uparrow \gamma^\downarrow$ and $\epsilon := \xi - 1$, and
looking at the expression \eqref{eq:defC*} for the optimal consumption, we need to consider how $\kappa(q(p)) = \kappa_{\gamma^\uparrow, \gamma^\downarrow}(q(p))$ and $n_{\xi}(q(p))$ vary with $\epsilon^\uparrow$ and $\epsilon^{\downarrow}$ as functions of $p$.

Set $u=u(\epsilon)=\hq - q_*(1+\epsilon)$. Recalling that $\kappa_{\gamma^\uparrow, \gamma^\downarrow}:[q_*(\xi),q^*(\xi)] \to [ \frac{1}{\gamma^\downarrow},\gamma^\uparrow]$ is monotonically decreasing and hence $\kappa_{\gamma^\uparrow, \gamma^\downarrow}(q(p)) = 1 + \sO(\epsilon)$, set
\[ \tilde{\kappa}_{u}(v) = \kappa_{\gamma^\uparrow, \gamma^\downarrow}(\hq + u v) = \kappa_{\gamma^\uparrow, \gamma^\downarrow}(\hq +  (\hq - q_*(1+\epsilon))v). \]
Then $\tilde{\kappa}_{u}(-1) = \kappa_{\gamma^\uparrow, \gamma^\downarrow}(q_*(\xi)) = \gamma^\uparrow$. Moreover,
\begin{eqnarray*}
\frac{\tilde{\kappa}_u'(v)}{\tilde{\kappa}_u(v)} & = &  u \frac{\kappa'_{\gamma^\uparrow, \gamma^\downarrow}(\hq + u v)}{\kappa_{\gamma^\uparrow, \gamma^\downarrow}(\hq + u v)} \\
& = & - \frac{u}{(\hq + u v)(1- \hq - u v)} \frac{(n^+_{(q_*(1+\epsilon))}(\hq + u v) - m(\hq + u v))}{(\ell(\hq + u v) - n^+_{(q_*(1+\epsilon))}(\hq + u v))} \\
& = & - \frac{u}{(\hq + u v)(1- \hq - u v)} \frac{N_u(uv)-M(uv)}{L(uv)-N_u(uv)} \\
& = & - \frac{u}{(\hq + u v)(1- \hq - u v)} \frac{DG_u(v)}{D(\hq + u v)(1- \hq - u v) - DG_u(v)} \\
& = & - \frac{u^3 g_2(v)}{\hq^2(1-\hq)^2} + \sO(u^4)
\end{eqnarray*}
We find, with $h_3(v) = (1+v)^2(2-v) = 3 \int_{-1}^v (1-w^2) dw$, for $v \in [-1, \zeta_{(u)}]$,
\begin{eqnarray*}
\tilde{\kappa}_u(v) & = & \gamma^\uparrow \exp \left( - \frac{u^3}{\hq^2(1-\hq)^2} \int_{-1}^v g_2(w) dw  + \sO(u^4) \right) \\
&  = &  \gamma^\uparrow \exp \left( - \frac{R \epsilon \Delta}{3 \hq^2(1-\hq)^2} h_3(v) + \sO(\epsilon^{4/3}) \right) \\
&  = &  \gamma^\uparrow \exp \left( - \frac{\epsilon }{4} h_3(v) + \sO(\epsilon^{4/3}) \right)
\end{eqnarray*}
Note that $\kappa_u(1) = \gamma^\uparrow \exp \left( - \epsilon   \right) + \sO(\epsilon^{4/3} ) =   \gamma^\uparrow \frac{1}{1+ \epsilon} + \sO(\epsilon^{4/3}) = \frac{1}{\gamma^\downarrow} + \sO(\epsilon^{4/3})$.

Since $\kappa_{\gamma^\uparrow, \gamma^\downarrow}(q(p)) = 1 + \sO(\epsilon)$ it follows that $p(q) = \frac{q}{1 - q(1- \kappa_{\gamma^\uparrow, \gamma^\downarrow}(q))} = q + \sO(\epsilon)$ which can be inverted to give $q(p) = p + \sO(\epsilon)$. Further, $q$ and $q$ both lie in intervals of width of order $\epsilon^{1/3}$ around $\hq$, so $q(p) = \hq + \sO(\epsilon^{1/3})$. It follows that, with $p = \frac{y \phi}{x+y \phi}$ and $v$ as shorthand for $\frac{p-\hq}{u}$ in later lines,
\begin{eqnarray*}
x + y \phi \kappa_{\gamma^\uparrow, \gamma^\downarrow}(q(p)) & = & (x + y \phi)\left( 1 + p (\kappa_{\gamma^\uparrow, \gamma^\downarrow}(q(p)) - 1) \right) \\
& = & (x + y \phi) \left( 1 + \hq (\kappa_{\gamma^\uparrow, \gamma^\downarrow}(p) - 1) + \sO(\epsilon^{4/3}) \right) \\
& = & (x+ y \phi)\left( 1 + \hq (\tilde{\kappa}_{u}(v) - 1) + \sO(\epsilon^{4/3}) \right) \\
& = & (x+ y \phi)\left( 1 + \hq \left( \gamma^\uparrow \exp \left( - \frac{\epsilon}{4}h_3(v) \right)    - 1\right) + \sO(\epsilon^{4/3}) \right) \\
& = & (x+ y \phi)\left( 1 + \hq \left( (1+\epsilon^\uparrow)\left( 1 - \frac{\epsilon^\uparrow + \epsilon^{\downarrow}}{4}h_3(v) \right)    - 1\right) + \sO(\epsilon^{4/3}) \right) \\
& = & (x+ y \phi)\left( 1 + \hq \left( 1 - \frac{h_3(v)}{4} \right) \epsilon^\uparrow - \hq \frac{h_3(v)}{4} \epsilon^{\downarrow} \right) + \sO(\epsilon^{4/3})
\end{eqnarray*}

Now we turn to the expansion of $n_{\xi}(q(p))$.
We have $n_{\xi}(q(p)) = n^+_{(\hq - u(\epsilon))}(q(p))$.
Then, using $g_3 = - Ah_3$ we have
\begin{eqnarray*}
n^+_{(\hq - u(\epsilon))}(q(p)) 
& = & \hatm + N_{u(\epsilon)}(q(p)-\hq) \\
& = & \hatm + M(q(p)-\hq) + DG_{u(\epsilon)} \left( \frac{(q(p)-\hq)}{u(\epsilon)} \right) \\
& = & \hatm+ RD (q(p)-\hq)^2 + D u(\epsilon)^2 g_2\left(\frac{(q(p)-\hq)}{u(\epsilon)}\right) + D u(\epsilon)^3 g_3\left(\frac{(q(p)-\hq)}{u(\epsilon)}\right) + \sO(\epsilon^{4/3}) \\
& = & \hatm + \frac{R(1-S)}{S} \frac{\sigma^2}{2} u(\epsilon)^2 - D Au(\epsilon)^3 h_3\left(\frac{(q(p)-\hq)}{u(\epsilon)}\right) + \sO(\epsilon^{4/3}) \\
& = & \hatm \left( 1 + \frac{R(1-S)}{S} \frac{\sigma^2}{2 \hatm} \left( \Delta^{2/3} \epsilon^{2/3} + 2 \Delta \Sigma \epsilon \right) - \frac{D A \Delta}{\hatm} \epsilon h_3 \left(\frac{p-\hq}{\Delta^{1/3} \epsilon^{1/3}}\right) \right) + \sO(\epsilon^{4/3})
\end{eqnarray*}
Note that $\frac{ D A \Delta}{\hatm}$ can be written more simply as $\frac{\hq}{4}\frac{1-S}{S}$ and $\frac{R \Delta \Sigma \sigma^2}{\hat{m}} = \frac{\hq}{2}$. Hence this expression simplifies to
\[ n^+_{(\hq - u(\epsilon))}(q(p)) = \hatm \left( 1 + \frac{R(1-S)}{S} \frac{\sigma^2}{2 \hatm} \Delta^{2/3} \epsilon^{2/3}
+ \hq \frac{(1-S)}{2S} \left(1  - \frac{1}{2} h_3\left(\frac{p-\hq}{\Delta^{1/3} \epsilon^{1/3}}\right) \right) \epsilon \right) + \sO(\epsilon^{4/3}). \]

Putting this all together and using \eqref{eq:defC*},
\begin{eqnarray*}
\lefteqn{C^*_{1+\epsilon^\uparrow,1+\epsilon^\downarrow}(t,x,\phi,y)} \\
& = & (x+y \phi) \hatm \left( 1 + \hq \left( 1 - \frac{1}{4} h_3\left(\frac{p-\hq}{\Delta^{1/3} \epsilon^{1/3}}\right) \right) \epsilon^\uparrow - \hq \frac{1}{4} h_3\left(\frac{p-\hq}{\Delta^{1/3} \epsilon^{1/3}}\right) \epsilon^{\downarrow}  \right) \\
&& \hspace{10mm} \times \left( 1 + \frac{R(1-S)}{S} \frac{\sigma^2}{2 \hatm} \Delta^{2/3} \epsilon^{2/3}
+ \hq \frac{(1-S)}{2S} \left( 1 - \frac{1}{2} h_3\left(\frac{p-\hq}{\Delta^{1/3} \epsilon^{1/3}}\right) \right) \epsilon \right) + \sO(\epsilon^{4/3}) \\
& = & (x+y \phi) \hatm \left( 1 +  \frac{R(1-S)}{S} \frac{\sigma^2}{2 \hatm} \Delta^{2/3} (\epsilon^\uparrow+ \epsilon^\downarrow)^{2/3} +
\frac{\hq}{2S} \left( 1  - \frac{1}{2} h_3\left(\frac{p-\hq}{\Delta^{1/3} \epsilon^{1/3}}\right) \right) (\epsilon^\uparrow+ \epsilon^\downarrow) + \frac{\hq}{2} (\epsilon^\uparrow - \epsilon^\downarrow) \right) \\
&& \hspace{40mm}+ \sO(\epsilon^{4/3})
\end{eqnarray*}
\end{proof}

\section{Proofs for the results on comparative statics}
\label{app:compstat}

\subsection{Comparative statics in S}
\label{sapp:compstatS}

\begin{proof}[Proof of Proposition~\ref{prop:compstat_S}]
	
	We will establish the ordering of the optimal purchase/sale boundaries in the $q$-coordinate, which can then be extended to that of $\hat{p}_*^*$ and $\tilde{p}_*^*$ by the monotonicity of the Möbius transformation over the relevant domain.
	
	Let $$\overline{m}(q):=\frac{S}{1-S}\left(m(q)-\frac{\delta}{S}\right)=-r-\lambda \sigma q+\frac{R\sigma^2}{2} q^{2}$$ which is independent of $S$. Similarly, let $$\overline{\ell}(q):=\frac{S}{1-S}\left(\ell(q)-\frac{\delta}{S}\right)=-r-\left(\lambda \sigma-\frac{\sigma^2}{2}\right)q-(1-R)\frac{\sigma^2}{2}q^{2}$$ which is also independent of $S$.
	
	Suppose $(n(\cdot),q_*,q^*)$ is the solution to the free boundary problem in Proposition \ref{prop:n}. Then  $n'(q)=O(q,n(q))$ on $q\in[q_*,q^*]$, $n(q_*)=m(q_*)$, $n(q^*)=m(q^*)$, and
	\begin{align*}
		\ln\gamma^{\uparrow}\gamma^{\downarrow}=\ln\xi=\int_{q_{*}}^{q^{*}}\frac{1}{q(1-q)}\frac{n(q)-m(q)}{\ell(q)-n(q)}\dd q.
	\end{align*}
	Define $\bar{n}(q):=\frac{S}{1-S}\left(n(q)-\frac{\delta}{S}\right)$. We have
	\begin{align*}
		\bar{n}'(q)=\frac{S}{1-S}n'(q)=\frac{S}{1-S}O(q,n(q))&=\frac{S}{1-S}O(q,(1-S)\bar{n}(q)/S+\delta/S)\\
		&=\frac{\bar{n}(q)}{1-q}\frac{\bar{n}(q)-\bar{m}(q)}{\bar{\ell}(q)-\bar{n}(q)}-\frac{1}{S}\frac{1}{1-q}\frac{\bar{n}(q)-\bar{m}(q)}{\bar{\ell}(q)-\bar{n}(q)}\left(\bar{n}(q)+\delta\right)\\
		&=\bar{O}(q,\bar{n}(q))
	\end{align*}
	over $q\in[q_*,q^*]$, where $$\bar{O}(q,\bar{n}):=\frac{\bar{n}}{1-q}\frac{\bar{n}-\bar{m}(q)}{\bar{\ell}(q)-\bar{n}}-\frac{1}{S}\frac{1}{1-q}\frac{\bar{n}-\bar{m}(q)}{\bar{\ell}(q)-\bar{n}}\left(\bar{n}+\delta\right).$$ The boundary conditions at $q=q^*_*$ can be restated as $\bar{n}(q_*)=\bar{m}(q_*)$ and $\bar{n}(q^*)=\bar{m}(q^*)$. The integral constraint becomes
	\begin{align*}
		\ln \xi=\int_{q_{*}}^{q^{*}}\frac{1}{q(1-q)}\frac{\bar{n}(q)-\bar{m}(q)}{\bar{\ell}(q)-\bar{n}(q)} \dd q.
	\end{align*}
	Finally, away from $1$, the solution trajectory $n=(n(q))_{q_*\leq q\leq q^*}$ always satisfies $\frac{1}{1-q}\frac{n(q)-m(q)}{\ell(q)-n(q)}\sgn(q_M)\geq 0$ and in turn  $\frac{1}{1-q}\frac{\bar{n}(q)-\bar{m}(q)}{\bar{\ell}(q)-\bar{n}(q)}\sgn(q_M)\geq 0$.
	
	\begin{enumerate}
		
		\item If \eqref{eq:suffcon} holds, then
		\begin{align*}
			\bar{n}(q)+\delta=\frac{S(n(q)-\delta)}{1-S}> \frac{S(m_M-\delta)}{1-S}=\delta-r-\frac{\lambda^2}{2R}\geq 0
		\end{align*}
		as $\sgn(1-S)n(q)> \sgn(1-S)m_M$. If we view $\bar{O}(q,\bar{n})=\bar{O}(q,\bar{n};S)$ as a function of $S$, then clearly $\bar{O}$ is strictly increasing (resp. decreasing) in $S$ over $(q,\bar{n})$ on which $\frac{1}{1-q}\frac{\bar{n}-\bar{m}(q)}{\bar{\ell}(q)-\bar{n}}\left(\bar{n}+\delta\right)$ is positive (resp. negative). Hence given \eqref{eq:suffcon}, we must have $\sgn(q_M)\bar{O}(q,\overline{n}(q);\hat{S})> \sgn(q_M)\bar{O}(q,\overline{n}(q);\tilde{S})$ over $q\in(q_*,q^*)\setminus\{1\}$ for $\hat{S}>\tilde{S}$.
		
		Recall the definitions of $(n^{+}_{(z)}(q))_{q\in[z,\zeta^{+}(q)]}$, $\zeta^{+}(z)$ and $\Sigma^{+}(z)$ as well as the relevant domain of $z$ defined in Proposition \ref{prop:n first:R<1} or \ref{prop:n first:R>1}. Define $(a^{+}_{(z)}(q))_{q\in[z,\zeta^{+}(z)]}$ via $a^{+}_{(z)}(q):=\frac{S}{1-S}(n^{+}_{(z)}(q)-\frac{\delta}{S})$. By construction,
		\begin{align*}
			\frac{\dd}{\dd q}a^{+}_{(z)}(q)=\bar{O}(q,a^{+}_{(z)}(q);S),\qquad q\in[z,\zeta^+(z)],
		\end{align*}
		$a^+_{(z)}(z)=\bar{m}(z)$, $a^+_{(z)}(\zeta^{+}(z))=\bar{m}(\zeta^+(z))$,
		and $$\Sigma^{+}(z)=\exp\left(\int_{z}^{\zeta^{+}(z)}\frac{1}{q(1-q)}\frac{a^{+}_{(z)}(q)-\bar{m}(q)}{\bar{\ell}(q)-a^{+}_{(z)}(q)}\dd q\right).$$ 
		
		Let $\hat{n}_{(z)}^{+}$ and $\tilde{n}_{(z)}^{+}$ be the two versions of $n^+_{(z)}$ under $S=\hat{S}$ and $S=\tilde{S}< \hat{S}$ respectively while all other model parameters are fixed. We will apply the same overhead notations to $\zeta^{+}(z)$, $a^{+}_{(z)}$ and $\Sigma^{+}$ which are defined relative to the underlying parameters $\hat{S}$ and $\tilde{S}$. Since $\bar{m}(\cdot)$ does not depend on $S$, $\hat{a}^+_{(z)}(z)=\bar{m}(z)=\tilde{a}^+_{(z)}(z)$. Also $(\hat{a}^+_{(z)})'(z)=0=(\tilde{a}^+_{(z)})'(z)$ for $z\neq 1$.
		
		Suppose that $q_M\in(0,1)$. We first show that $\hat{a}^+_{(z)}(q)> \tilde{a}^+_{(z)}(q)$ near $q=z$ for any $z\in(0,q_M)$. Using the ODE satisfied by $a^+_{(z)}$ and L’Hospital rule, we obtain
			\begin{align*}
				\lim_{q\downarrow z}\frac{a^+_{(z)}(q)-\bar{m}(z)}{(q-z)^2}=\frac{\bar{m}'(z)}{\sigma^2 z(1-z)^2}\left[\frac{1}{S}(\bar{m}(z)+\delta)-\bar{m}(z)\right],
			\end{align*}
			where the leading factor is negative and in turn the whole expression is strictly increasing in $S$ as $\bar{m}(z)+\delta=\frac{S(m(z)-\delta)}{1-S}> \frac{S(m_M-\delta)}{1-S}=\delta-r-\frac{\lambda^2}{2R}\geq 0$ for all $z\neq q_M$ under condition \eqref{eq:suffcon}. Hence $\hat{a}^+_{(z)}(q)$ must be initially strictly larger than $\tilde{a}^+_{(z)}(q)$. Together with the facts that $\bar{O}(q,a;\hat{S})> \bar{O}(q,a;\tilde{S})$ along any solution trajectory (excluding the end points), we conclude $\hat{\zeta}^+(z)\geq \tilde{\zeta}^+(z)$ and $\hat{a}^+_{(z)}(q)\geq \tilde{a}^+_{(z)}(q)$ at least up to $q=\tilde{\zeta}^+(z)$<1. Then
		\begin{align*}
			\hat{\Sigma}^{+}(z)=\exp\left(\int_{z}^{\hat{\zeta}^{+}(z)}\frac{1}{q(1-q)}\frac{\hat{a}^+_{(z)}(q)-\bar{m}(q)}{\bar{\ell}(q)-\hat{a}^+_{(z)}(q)}\dd q \right)
			&\geq \exp\left(\int_{z}^{\tilde{\zeta}^{+}(z)}\frac{1}{q(1-q)}\frac{\hat{a}^+_{(z)}(q)-\bar{m}(q)}{\bar{\ell}(q)-\hat{a}^+_{(z)}(q)}\dd q \right)\\
			&\geq \exp\left(\int_{z}^{\tilde{\zeta}^+(z)}\frac{1}{q(1-q)}\frac{\tilde{a}^+_{(z)}(q)-\bar{m}(q)}{\bar{\ell}(q)-\tilde{a}^+_{(z)}(q)}\dd q \right) \\
			&=\tilde{\Sigma}^+(z),
		\end{align*}
		and in turn $\hat{q}_{*}=(\hat{\Sigma}^+)^{-1}(\xi)\geq (\tilde{\Sigma}^+)^{-1}(\xi)=\tilde{q}_{*}$ since $\Sigma^+(\cdot)$ is decreasing.
		
		The arguments are similar in the case of $q_M>1$ except some extra care is needed around $q=1$. For any $z\in(0,1)$, $a^+_{(z)}(1)=\bar{m}(1)$ and $(a^+_{(z)})'(1)=\bar{m}'(1)$ such that the ordering of $\hat{a}^+_{(z)}(q)$ and $\tilde{a}^+_{(z)}(q)$ is not immediately clear on $q>1$. But again we can deduce via L’Hospital rule that
			\begin{align*}
				\lim_{q\to 1}\frac{a^+_{(z)}(q)-\bar{m}(q)}{(1-q)^2}=-\frac{\sigma^2}{2}\frac{\bar{m}'(1)}{\frac{1}{S}(\bar{m}(1)+\delta)-\bar{m}(1)},
			\end{align*}	
			which is positive and strictly increasing in $S$ under condition \eqref{eq:suffcon}. Hence we conclude $\hat{a}^+_{(z)}(q)> \tilde{a}^+_{(z)}(q)$ in the neighbourhood of $q=1$.
		
		The results for $q^*$ can be established similarly by parameterising the solutions to the family of ODEs via their right boundary points. Results for the case of $q_M<0$ can be derived using similar ideas.
		
		\item 	If instead the parameters are such that \eqref{eq:othercon} holds, then the argument in part (a) does not work as $\bar{n}(q)+\delta$ may change sign and hence the monotonicity of $\bar{O}(q,\bar{n};S)$ in $S$ is ambiguous. But suppose $q_M>0$ and the problems under $S\in\{\hat{S},\tilde{S}\}$ are well-posed for all sufficiently small transaction costs. Then there exists $\delta>0$ such that for all $z\in(q_M-\delta,q_M)$, $n^{+}_{(z)}(q)$ is well-defined and is strictly positive over $q\in[z,\zeta^{+}(z)]$. Moreover, since the family of solutions $(n^{+}_{(z)}(\cdot))_z$ is monotonic in $z$ with $n^{+}_{(z)}(q)\to m_M$ as $z\to q_M$, one can choose $\Delta$ such that $|n^+_{(z)}(q)-m_M|<\omega$ and $|m(q)-m_M|<\omega$ on $q\in[z,\zeta^+(z)]$ for all $z\in(q_M-\Delta,q_M)$ where $\omega>0$ is arbitrary. Then
		\begin{align*}
			a^+_{(z)}(q)+\delta=\frac{S(n^{+}_{(z)}(q)-\delta)}{1-S}&\leq \frac{S(m_M+\sgn(1-S)\omega-\delta)}{1-S}\\
			&=\delta-r-\frac{\lambda^2}{2R}+\frac{S}{|1-S|}\omega< 0
		\end{align*}
		and similarly $\bar{m}(q)+\delta<0$ if $\omega$ is chosen to be sufficiently small. In such case, $\bar{O}(q,a^{+}_{(z)}(q);\hat{S})< \bar{O}(q,a^{+}_{(z)}(q);\tilde{S})$ along the solution trajectory $(a^{+}_{(z)}(q))_{z<q<\zeta^+(z)}$, and that $\hat{a}^{+}_{(z)}(q)<\tilde{a}^{+}_{(z)}(q)$ near $q=z$ (and near $q=1$ as well if $q_M>1$). Following similar arguments in part (a), we can conclude $\hat{\zeta}^+(z)\leq \tilde{\zeta}^+(z)$ and $\hat{a}^{+}_{(z)}(q)\leq \tilde{a}^{+}_{(z)}(q)$ up to $q=\hat{\zeta}^+(z)$, and in turn $\hat{\Sigma}^+(z)\leq \tilde{\Sigma}^+(z)$ for $z\in(q_M-\delta,q_M)$. Thus there exists $\epsilon:=\hat{\Sigma}^+(q_M-\Delta)-1>0$ such that $\hat{q}_{*}=(\hat{\Sigma}^+)^{-1}(\xi)\leq (\tilde{\Sigma}^+)^{-1}(\xi)=\tilde{q}_{*}$ for all $\xi\leq 1+\epsilon$. 	
		
		The proof for $q^*$ is similar upon parameterising the solutions to the family of ODEs via the right boundary points. The analysis can be analogously extended to $q_M<0$.
	\end{enumerate}
\end{proof}

\subsection{Comparative statics in R}
\label{sapp:compstatR}
Now we turn to comparative statics in $R$.
For the case $R,S>1$ it is useful to introduce an auxiliary function. For $q\in (0,\infty)\setminus \{1,2S/(S+1)\}$, define $c=c(q)$ as the solution to the equation $O(q,c)/c=2/q$. Straightforward algebra leads to
\begin{align}
c(q) = \ell(q) + \frac{(S-1)q}{2S - (S+1)q} (m(q) - \ell(q)) = m(q) - \frac{2S(1-q)}{2S-(S+1)}(m(q)-\ell(q)).
\label{eq:cfun}
\end{align}
The definition of $c$ is extended to $q=0$ and $q=1$ via $c(0):=\lim_{q\downarrow 0}c(q)=\ell(0)=m(0)$ and $c(1):=\lim_{q\to 1}c(q)=m(1)=\ell(1)$. Note that $c(q)$ is continuous on $q\in[0,2S/(S+1))\supset (0,1)$ and we have $c'(0)=\ell'(0)$ and $c'(1)=m'(1)$. We also have
\begin{align*}
\begin{cases}
\ell(q) < c(q) < m(q),  &q \in (0,1); \\
c(q)<m(q)<\ell(q),      & q \in \left(1,\frac{2S}{S+1}\right); \\
m(q)<\ell(q) < c(q),      & q \in \left(\frac{2S}{S+1},\infty \right).
\end{cases}
\end{align*}
Also, on $(0,1)\times (0,\infty)$, $\frac{O(q,z)}{z} > \frac{2}{q}$ for $z \in (\ell(q),c(q))$ and $\frac{O(q,z)}{z} < \frac{2}{q}$ for $z \in (c(q),m(q))$. Further, on $(1,\frac{2S}{S+1})\times(0,m(q))$, $\frac{O(q,z)}{z} < \frac{2}{q}$ on $(0,c(q)^+)$ and  $\frac{O(q,z)}{z} > \frac{2}{q}$ on $(c(q)^+,m(q))$. Finally, on $(\frac{2S}{S+1}, \infty)$, $\frac{O(q,z)}{z} < \frac{2}{q}$ on $(0,m(q))$.


\begin{lemma}
Suppose $\lambda>0$ and the problem is well-posed such that $(n(\cdot),q_*,q^*)$ is the solution to the free boundary problem in Proposition \ref{prop:n}. Then any of the following is a sufficient condition for $2n(q)-qn'(q)> 0$ on $q\in[q_*,q^*]$.
\begin{enumerate}
\item $S<1$.
\item $S>1$, $m(0)>0$, and either (recall the definition of $c(q)$ in \eqref{eq:cfun}):
\begin{enumerate}
	\item $2c(q)- qc'(q)>0$ on $(0,1)$ if $q_M<1$ or on $(0,\frac{2S}{S+1})$ if $q_M>1$,
	\item $R\geq 2$.
\end{enumerate}
\end{enumerate}
A further sufficient condition for (b)(i) is that $\ell'(0) \geq 0$, or equivalently $(1-S)(\sigma-2\lambda)\geq 0$.
\label{lem:cq}
\end{lemma}

\begin{proof} 
\begin{enumerate}
\item The result under $R,S<1$ is trivial since the solution $n(q)$ is positive and decreasing (under $\lambda>0$). 

\item We first show that $\ell'(0)\geq 0$ is sufficient for condition (b)(i) to hold. By direct calculation,
\begin{align*}
	2c(q)-qc'(q)&=2\left[\frac{m(q)-\ell(q)}{\frac{2S}{S-1}\frac{1-q}{q}+1}+\ell(q)\right]-q\left[\frac{\left(\frac{2S}{S-1}\frac{1-q}{q}+1\right)(m'(q)-\ell'(q))+(m(q)-\ell(q))\frac{2S}{S-1}q^{-2}}{\left(\frac{2S}{S-1}\frac{1-q}{q}+1\right)^2}+\ell'(q)\right]\\
	&=\frac{\sigma^2}{2}\frac{(S-1)q}{S\left(\frac{2S}{S-1}\frac{1-q}{q}+1\right)^2}+2\ell(q)-q\ell'(q)\\
	&=\frac{\sigma^2}{2}\frac{(S-1)q}{S\left(\frac{2S}{S-1}\frac{1-q}{q}+1\right)^2}+\frac{2(\delta+r(S-1))}{S}+\frac{S-1}{S}\left(\lambda\sigma-\frac{\sigma^2}{2}\right)q \\
	&=\frac{\sigma^2}{2}\frac{(S-1)q}{S\left(\frac{2S}{S-1}\frac{1-q}{q}+1\right)^2}+2m(0)+\ell'(0)q> 0
\end{align*}
on $q>0$ away from $q=2S/(S+1)$.

\begin{enumerate}
	
	
	\item 	Suppose $q_M\in(0,1)$ such that $[q_*,q^*]\subset (0,1)$. By the given condition, $O(q,c(q))=2c(q)/q> c'(q)$ for $q\in[q_*,q^*]$. If $n^+_{(z)}=(n^+_{(z)}(q))_{q\geq z}$ is the solution to the initial value problem $n'=O(q,n)$ with $n^+_{(z)}(z)=m(z)$, $n^+_{(z)}(q)$ can then only upcross $c(q)$ from below on $q\in(0,1)$. Since $c(0)=\ell(0)=m(0)$ and the family of solutions $(n^+_{(z)}(\cdot))_z$ is increasing in $z$, we must have $n(q)=n^+_{(q_*)}(q)> n^+_{(0)}(q)\geq c(q)>\ell(q)$ for all $q\in[q_*,q^*]\subset (0,1)$. Finally, $\Xi(q,n):=\frac{O(q,n)}{n}=\frac{S-1}{S}\frac{1}{1-q}\frac{n-m(q)}{\ell(q)-n}$ is continuously decreasing in $n\in(\ell(q),\infty)$ for any fixed $q\in(0,1)$ and therefore
	\begin{align*}
		\frac{n'(q)}{n(q)}=\Xi(q,n(q))< \Xi(q,c(q))=\frac{2}{q}.
	\end{align*}
	
	Suppose $q_M>1$ and the transaction costs are sufficiently large such that $q_*<1<q_M<q^*$. Using the same argument as in the case of $q_M\in(0,1)$, we can still conclude $n(q)\geq c(q)$ and in turn $n'(q)/n(q)\leq 2/q$ up to $q<1$. But since $n$ now passes through the singular point $(1,m(1))$ at which we have $c(1)=m(1)=n(1)$ and $c'(1)=m'(1)=n'(1)$, the ordering of $c(q)$ and $n(q)$ is ambiguous to the right of $q=1$. We thus need to first argue that $n(q)>c(q)$ in the neighbourhood of $q=1$.
	
	Write $f(q):=n(q)-c(q)$. Then
	\begin{align*}
		f'(q)=n'(q)-c'(q)=\frac{S-1}{S}\frac{f(q)+c(q)}{1-q}\frac{f(q)+c(q)-m(q)}{\ell(q)-f(q)-c(q)}-c'(q)
	\end{align*}
	and a rearrangement of terms yields
	\begin{align}
		\frac{f(q)}{(q-1)^2}=\frac{S(f'(q)+c'(q))}{(S-1)(f(q)+c(q))}\left(\frac{\ell(q)-c(q)}{1-q}-\frac{f(q)}{1-q}\right)-\frac{c(q)-m(q)}{(1-q)^2}.
		\label{eq:f_lim}
	\end{align}
	Since it is known that $n(1)=m(1)=c(1)$ and $n'(1)=m'(1)=c'(1)$, we have $f(1)=f'(1)=0$ and $\lim_{q\to 1}f(q)/(1-q)=0$. Upon taking limit, we get
	\begin{align*}
		\lim_{q\to 1}\frac{f(q)}{(q-1)^2}&=\frac{S(m'(1)-\ell'(1))m'(1)}{(S-1)m(1)}-\frac{c''(1)-m''(1)}{2}\\
		&=\frac{S(m'(1)-\ell'(1))m'(1)}{(S-1)m(1)}-\frac{2S(m'(1)-\ell'(1))}{S-1}\\
		&=\frac{S(\ell'(1)-m'(1))}{S-1}\left[2-\frac{m'(1)}{m(1)}\right]>0
	\end{align*}
	because $\ell'(1)>m'(1)$ and $2m(1)>m'(1)$ under the assumption of $m(0)>0$, $q_M>1$ and $S>1$. Hence we conclude $n(q)>c(q)$ near $q=1$.
	
	Using similar arguments in the case of $q_M\in(0,1)$, we must have $\ell(q)>m(q)>n(q)>c(q)$ on $1< q<q^*\wedge \frac{2S}{S+1}$. As $\Xi(q,n)=\frac{O(q,n)}{n}=\frac{S-1}{S}\frac{1}{1-q}\frac{n-m(q)}{\ell(q)-n}$ is continuously decreasing in $n\in(0,\ell(q))$ for any fixed $q>1$, we have $\frac{n'(q)}{n(q)}=\frac{O(q,n(q))}{n(q)}< \frac{O(q,c(q))}{c(q)}=\frac{2}{q}$ for $1<q<q^*\wedge \frac{2S}{S+1}$. If $q^*>\frac{2S}{S+1}$, then $c(q)>\ell(q)>m(q)>n(q)>0$ on $q\in(\frac{2S}{S+1},q^*)$. Note that $\Xi(q,n)$ is discontinuous at $n=\ell(q)$. But since $H(q,n)$ is still locally continuous and decreasing on $n\in [0,\ell(q))$ and $n\in(\ell(q),\infty)$ respectively under fixed $q>1$, we have
	\begin{align*}
		\frac{2}{q}=\Xi(q,c(q))> \lim_{n\to\infty} \Xi(q,n)=\frac{S-1}{S}\frac{1}{q-1}\geq \frac{S-1}{S}\frac{1}{q-1}\frac{m(q)}{\ell(q)}= \Xi(q,0)> \Xi(q,n(q))=\frac{n'(q)}{n(q)}
	\end{align*}
	on $\frac{2S}{S+1}<q<q^*$. By continuity of $n(q)$ and $n'(q)$, the above inequality holds at $q=\frac{2S}{S+1}$ and $q=1$ as well.
	
	Finally, the case of $1\leq q_*<q_M<q^*$ under small transaction costs can be handled similarly using the fact that the solution $n(q)$ in such case is bounded below by $n^+_{(1)}(q)$.

	\item If $m(0) > 0$ and $\ell'(0) \geq 0$ then we are in Case b(i), so we only need to consider $m(0) > 0$ and $\ell'(0) < 0$. Since $\ell$ is concave we have $m(1)=\ell(1)<\ell(0)=m(0)$ so necessarily $q_M \in (0,1)$ and $[q_*,q^*] \subset (0,1)$. In what follows we only consider functions over $(0,1)$.
	
	Define $A(q) := qm(q)+(1-q)\ell(q)$. $A$ is chosen so that $A(z)=c(z)$ and $A'(z)=c'(z)$ for $z \in \{0,1\}$, and to be a simple function.
	It is easily seen that $A(q)>c(q)$. So, if we can show that $n$ is always above $A$ it will follow that $2n(q)-qn'(q)>0$ on $(q_*,q^*)$ as required.
	
	Note that $O(q,A(q)) = \frac{(S-1)}{S} \frac{A(q)}{q}$ so that if $n=A$ then $n'>A'$ is equivalent to $\frac{(S-1)}{S} \frac{A(q)}{q}> A'$. So, a sufficient condition for $n>A$ on $(0,1)$ is (here we use $m(q)-\ell(q) = Dq(1-q)$ where $D = \frac{(S-1)}{S} \frac{\sigma^2}{2}$)
	\begin{eqnarray*}
		\lefteqn{\frac{(S-1)}{S} \frac{A(q)}{q}  > \ell'(q) + (m(q)- \ell(q)) + q  (m'(q)- \ell'(q)) } \\
		& \Leftrightarrow & (S-1) ( \ell(q) + q(m(q)- \ell(q)) ) > Sq \ell'(q) + Sq(m(q)- \ell(q)) + S q^2  (m'(q)- \ell'(q)) \\
		& \Leftrightarrow & (S-1) \ell(q)  > Sq \ell'(q) + q(m(q)- \ell(q)) + S q^2  (m'(q)- \ell'(q)) \\
		& \Leftrightarrow & (S-1) \ell(0) + (S-1)q \ell'(0) +(S-1) \frac{q^2}{2} \ell''(0)  > Sq (\ell'(0) + q \ell''(0)) + Dq^2(1-q) + S D q^2(1-2q)  \\
		& \Leftrightarrow & (S-1) \ell(0)  > q \ell'(0) + \frac{q^2}{2} \ell''(0)(2S - (S-1)) + Dq^2((1-q) + S(1-2q)) \\
		& \Leftrightarrow & (S-1) \ell(0)  > q \ell'(0) + \frac{q^2}{2} \ell''(0)(S+1) + Dq^2((S+1)-q(1+2S)).
	\end{eqnarray*}	
	Thus we only need to establish $(S-1)\ell(0)>\sup_{q\in(0,1)} \varphi(q)$ where $\varphi(q)$ is a cubic function defined via
\[ \varphi(q):=q\ell'(0)+ \frac{\sigma^2}{2S}(S+1)(S-1)(2-R) q^2- \frac{\sigma^2}{2S} (S-1)(1+2S)q^3,
\] where we use that $\frac{1}{2} \ell''(0)+D = \frac{\sigma^2}{2S}(S-1)(2-R)$. Note that $\varphi(0)=0<(S-1)\ell(0)$ and $\varphi'(0)=\ell'(0)<0$. Hence we are done if the coefficient of $q^2$ is negative, which is equivalent to $R\geq 2$.

\end{enumerate}

\end{enumerate}

\end{proof}

\begin{proof}[Proof of Proposition~\ref{prop:CompstatR}]
We will show that $\hat{p}_*\leq \tilde{p}_*$ and  $\hat{p}^*\leq \tilde{p}^*$ under the stated conditions in Lemma \ref{lem:cq}. The last part of the proposition will then follow trivially.
	
In what follows, we write
\begin{align*}
m(q)&=m(q;R)=\frac{\delta-r(1-S)}{S}-\frac{1-S}{S}\lambda \sigma q+R\frac{1-S}{S}\frac{\sigma^2}{2}q^{2}=m(q;0)+R\frac{1-S}{S}\frac{\sigma^2}{2}q^{2}, \\
\ell(q)&=\ell(q;R)=\frac{\delta-r(1-S)}{S}-\frac{1-S}{S}\left(\lambda\sigma-\frac{\sigma^2}{2}\right)q-(1-R)\frac{1-S}{S}\frac{\sigma^2}{2}q^{2}=\ell(q;0)+R\frac{1-S}{S}\frac{\sigma^2}{2}q^{2},
\end{align*}	
where $m(q;0)$ and $\ell(q;0)$ do not depend on $R$. Note that $\ell(q)-m(q)=\ell(q;0)-m(q;0)$.

Let $f(q):=\frac{n(q)-m(q)}{\ell(q)-n(q)}$ where $n$ is a solution to $n'=O(q,n)$. Then
\begin{align*}
f'(q)&=\frac{(\ell(q)-n(q))(n'(q)-m'(q))-(n(q)-m(q))(\ell'(q)-n'(q))}{(\ell(q)-n(q))^2}\\
&=\frac{1+f(q)}{\ell(q)-m(q)}\left[\frac{S-1}{S(1-q)}f(q)(m(q)+\ell(q)f(q))-m'(q)-\ell'(q)f(q)\right]\\
&=\frac{1+f(q)}{\ell(q;0)-m(q;0)}\Biggl[\frac{S-1}{S(1-q)}f(q)\left(m(q;0)+R\frac{1-S}{S}\frac{\sigma^2}{2}q^{2}+\left(\ell(q;0)+R\frac{1-S}{S}\frac{\sigma^2}{2}q^{2}\right)f(q)\right)\\
&\qquad -m'(q;0)-R\frac{1-S}{S}\sigma^2 q-\left(\ell'(q;0)+R\frac{1-S}{S}\sigma^2 q\right)f(q)\Biggr]\\
&=\Psi(q,f(q);0)-\frac{R (1+f(q))^2}{S(1-q)^2}\left((1-S)qf(q)+2S(1-q)\right)=:\Psi(q,f(q);R)
\end{align*}
for some function $\Psi(q,f;R)$ with $\Psi(q,f;0)$ being independent of $R$. In particular, $\Psi$ depends on $R$ linearly and we have $$\frac{\partial}{\partial R}\Psi(q,f;R)=- \frac{(1+f)^2}{S(1-q)^2}((1-S)qf+2S(1-q)).$$

Along a solution trajectory of the free boundary value problem,
\begin{align*}
(1-S)qf(q)+2S(1-q)&=(1-S)q\frac{n(q)-m(q)}{\ell(q)-n(q)}+2S(1-q)\\
&=-Sq(1-q)\frac{O(q,n(q))}{n(q)}+2S(1-q)\\
&=\frac{S(1-q)}{n(q)}\left(2n(q)-qn'(q)\right).
\end{align*}
Then if any condition in Lemma \ref{lem:cq} holds such that $2n-qn'>0$, we have $\sgn\left(\frac{\partial}{\partial R}\Psi(q,f(q);R)\right)=\sgn(q-1)$.

Recall the definitions of $(n^{+}_{(z)}(q))_{q\in[z,\zeta^{+}(q)]}$, $\zeta^{+}(z)$ and $\Sigma^{+}(z)$ defined in Proposition \ref{prop:n first:R<1} or \ref{prop:n first:R>1}. Set $f^{+}_{(z)}(q):=\frac{n^+_{(z)}(q)-m(q)}{\ell(q)-n^+_{(z)}(q)}$. Note that $\sgn(f^+_{(z)}(q))=\sgn(1-q)$. By construction,
\begin{align*}
\frac{\dd}{\dd q}f^{+}_{(z)}(q)=\Psi(q,f^{+}_{(z)}(q);R),\qquad q\in[z,\zeta^+(z)],
\end{align*}
$f^{+}_{(z)}(z)=0=f^{+}_{(z)}(\zeta(z))$ and $$\Sigma^{+}(z)=\exp\left(\int_z^{\zeta^+(z)}\frac{f^+_{(z)}(q)}{q(1-q)}\dd q\right).$$
Fix $\tilde{R}< \hat{R}$, and we will use the corresponding overhead notations to denote the particular quantities evaluated under $R=\tilde{R}$ and $R=\hat{R}$ respectively similar to what we have done in the proof of Proposition \ref{prop:compstat_S}.

If $\tilde{R}<\hat{R}$ are such that $\tilde{q}_M\geq 1\geq \hat{q}_M$, then we obviously have $\hat{q}^*\leq 1\leq \tilde{q}^*$. It is therefore sufficient to only consider values of $\tilde{R},\hat{R}$ such that $\tilde{q}_M, \hat{q}_M<1$ or $\tilde{q}_M, \hat{q}_M>1$.

If $\tilde{q}_M,\hat{q}_M<1$, then the domain of a solution to the free boundary problem must be contained within $(0,1)$. Since $\Psi(q,f^+_{(z)}(q);R)$ is strictly decreasing in $R$ on $q\in(0,1)$ and $\hat{f}^{+}_{(z)}(z)=0=\tilde{f}^{+}_{(z)}(z)$ with
\begin{align}
(f^+_{(z)})'(z)=-\frac{m'(z;0)+R\frac{1-S}{S}\sigma^2 z}{\frac{1-S}{S}\frac{\sigma^2}{2}z(1-z)}
\label{eq:fz}
\end{align}
which is decreasing in $R$ such that $(\hat{f}^{+}_{(z)})'(z)<(\tilde{f}^{+}_{(z)})'(z)$ for $z\in(0,1)$, we must have $\hat{\zeta}(z)\leq \tilde{\zeta}(z)<1$ and $0\leq  \hat{f}_{(z)}(q)\leq \tilde{f}^+_{(z)}(q)$ for $q\in[z,\hat{\zeta}^+(z)]$. Then
\begin{align*}
\tilde{\Sigma}^+(z)&
= \exp\left(\int_{z}^{\tilde{\zeta}^+(z)}\frac{\tilde{f}^+_{(z)}(q)}{q(1-q)}\dd q\right)
\geq \exp\left(\int_{z}^{\hat{\zeta}^+(z)}\frac{\tilde{f}^+_{(z)}(q)}{q(1-q)}\dd q\right)
\geq \exp\left(\int_{z}^{\hat{\zeta}^+(z)}\frac{\hat{f}^+_{(z)}(q)}{q(1-q)}\dd q\right)
=\hat{\Sigma}^+(z)
\end{align*}
and then $\hat{q}_{*}=(\hat{\Sigma}^+)^{-1}(\xi)\leq (\tilde{\Sigma}^+)^{-1}(\xi)=\tilde{q}_{*}$.

Suppose $\tilde{q}_M, \hat{q}_M>1$. In case of $z> 1$ such that the domain of the solution lives on a subset of $(1,\infty)$, $\Psi(q,f^+_{(z)}(q);R)$ is now increasing in $R$ and  $(\hat{f}^{+}_{(z)})'(z)>(\tilde{f}^{+}_{(z)})'(z)$ for $z>1$ using \eqref{eq:fz}. Hence $\tilde{\zeta}(z)\geq \hat{\zeta}(z)> q_M$ and $0\geq \hat{f}^+_{(z)}(q)\geq \tilde{f}^+_{(z)}(q)$ for $1< z\leq q\leq \hat{\zeta}^+(z)$. Then we can deduce $\tilde{\Sigma}^+(z)\geq \hat{\Sigma}^+(z) $ on $z> 1$. The same conclusion can be extended to $z=1$ upon showing that $(\hat{f}^{+}_{(1)})'(1)>(\tilde{f}^{+}_{(1)})'(1)$. It is sufficient to verify that $$\lim_{q\to 1}\frac{f(q)}{q-1}=\frac{S}{1-S}\frac{m'(1)}{m(1)}=\frac{R\sigma^2-\lambda\sigma}{R\frac{1-S}{S}\frac{\sigma^2}{2}+\frac{\delta-r(1-S)}{S}-\frac{1-S}{S}\lambda\sigma}$$ is increasing in $R$. One can check that the first order derivative of this expression with respect to $R$ is proportional to
\begin{align*}
	\frac{\delta-r(1-S)}{S}-\frac{\lambda\sigma(1-S)}{2S}=m(0)+\frac{m'(0)}{2}=\ell(1)+\frac{1-S}{2S}R\sigma^2(q_M-1),
\end{align*}
which is strictly positive in both cases of $S<1$ (as $\ell(1)>0$ under well-posedness and we work with $q_M>1$) and $S>1$ (since $m'(0)>0$, and $m(0)>0$ as a a required condition).

Finally, for $\tilde{q}_M, \hat{q}_M>1$ and $z<1$ such that $\tilde{\zeta}^+(z)=\tilde{\zeta}^+(1)\geq\hat{\zeta}^+(1)=\hat{\zeta}^+(z)$, we have $\tilde{f}^+_{(z)}(q)\geq \hat{f}^+_{(z)}(q)\geq 0$ on $q\leq 1$, $\tilde{f}^+_{(z)}(1)= \hat{f}^+_{(z)}(1)=0$ and $0\geq \hat{f}^+_{(z)}(q)= \hat{f}^+_{(1)}(q)\geq\tilde{f}^+_{(1)}(q)= \tilde{f}^+_{(z)}(q)$ on $1\leq q\leq \hat{\zeta}(1)$. Then
\begin{align*}
\tilde{\Sigma}^+(z)
= \exp\left(\int_{z}^{\tilde{\zeta}^+(z)}\frac{\tilde{f}^+_{(z)}(q)}{q(1-q)}\dd q  \right)
&=\exp\left(\int_{z}^{1}\frac{\tilde{f}^+_{(z)}(q)}{q(1-q)}\dd q+\int_{1}^{\tilde{\zeta}^+(1)}\frac{\tilde{f}^+_{(z)}(q)}{q(1-q)}\dd q\right)\\
&\geq\exp\left(\int_{z}^{1}\frac{\hat{f}^+_{(z)}(q)}{q(1-q)}\dd q+\int_{1}^{\hat{\zeta}^+(1)}\frac{\hat{f}_{(z)}(q)}{q(1-q)}\dd q\right)
=\hat{\Sigma}^+(z).
\end{align*}
Hence $\tilde{\Sigma}^+(z)\geq \hat{\Sigma}^+(z) $ for all $z$ which establishes $\hat{q}_{*}\leq \tilde{q}_{*}$.

The result for the sale boundary can again be established by parametrising the solutions via their right boundary points.
\end{proof}

\small
\bibliography{SDU-TCfeb24}
\bibliographystyle{amsplain}

\end{document}